\documentclass[accepted]{uai2023} %

\newif\ifbiblatex
    \biblatexfalse %
    
\newif\ifvfull
    \vfulltrue %
    
    \newif\ifvfullred %
    \vfullredtrue

\usepackage[american]{babel}

\ifbiblatex\else
    \usepackage{natbib} %
    \fi
\usepackage{booktabs} %
\usepackage{algorithm}
\usepackage{algorithmic}

\newcommand\vfull[1]{{\ifvfullred\color{red!75!black}\fi\ifvfull#1\fi}}

\usepackage{microtype}
\relax %
    \newif\ifmarginprooflinks
    	\marginprooflinkstrue

\relax %
    \usepackage[dvipsnames]{xcolor}
    
    \usepackage{mathtools}
    \usepackage{amssymb}
		\DeclareMathSymbol{\shortminus}{\mathbin}{AMSa}{"39}
    
    \usepackage{bbm}
    \usepackage{faktor}
    \usepackage{booktabs}
    \usepackage{graphicx}
    \usepackage{scalerel}
    \usepackage{enumitem}
    \usepackage{wrapfig}
    \usepackage{nicefrac}\let\nf\nicefrac

    \usepackage{hyperref} %
        \hypersetup{colorlinks=true, linkcolor=blue!75!black, urlcolor=magenta, citecolor=green!50!black}

\relax %
     \let\H\relax %
	\DeclareMathOperator{\H}{\mathrm{H}} %
	\DeclareMathOperator{\I}{\mathrm{I}} %
	\DeclareMathOperator*{\Ex}{\mathbb{E}} %

	\DeclareMathOperator{\supp}{\mathrm{supp}} %
	\DeclareMathOperator*{\argmin}{arg\,min} %

    \DeclarePairedDelimiterX{\infdivx}[2]{(}{)}{%
        #1\;\delimsize\|\;#2} %
	\newcommand{\thickD}{I\mkern-8muD}  %
	\newcommand{\kldiv}{\thickD\infdivx} %
	\newcommand{\tto}{\rightarrow\mathrel{\mspace{-15mu}}\rightarrow} %
    
    \newcommand{\Rext}{\mskip1mu\overline{\mskip-1mu\mathbb R\!}\,}
    \def\Rextge0{\Rext_{\scriptscriptstyle \ge 0}} %
    \def\Rge0{\mathbb{R}_{\scriptscriptstyle \ge 0}} %
    \newcommand{\mat}[1]{\mathbf{#1}} %

    \DeclarePairedDelimiter{\pqty}{\lparen}{\rparen}

	\makeatletter
	\newcommand{\subalign}[1]{%
	  \vcenter{%
	    \Let@ \restore@math@cr \default@tag
	    \baselineskip\fontdimen10 \scriptfont\tw@
	    \advance\baselineskip\fontdimen12 \scriptfont\tw@
	    \lineskip\thr@@\fontdimen8 \scriptfont\thr@@
	    \lineskiplimit\lineskip
	    \ialign{\hfil$\m@th\scriptstyle##$&$\m@th\scriptstyle{}##$\hfil\crcr
	      #1\crcr
	    }%
	  }%
    	}
	\makeatother
	\newcommand\numberthis{\addtocounter{equation}{1}\tag{\theequation}}

\usepackage{tikz}
	\usetikzlibrary{positioning,fit,calc, decorations, arrows, shapes, shapes.geometric}
	\usetikzlibrary{cd}

	\tikzset{AmpRep/.style={ampersand replacement=\&}}
	\tikzset{center base/.style={baseline={([yshift=-.8ex]current bounding box.center)}}}
	\tikzset{paperfig/.style={center base,scale=0.9, every node/.style={transform shape}}}

	\tikzset{dpadded/.style={rounded corners=2, inner sep=0.7em, draw, outer sep=0.3em, fill={black!50}, fill opacity=0.08, text opacity=1}}
	\tikzset{dpad0/.style={outer sep=0.05em, inner sep=0.3em, draw=gray!75, rounded corners=4, fill=black!08, fill opacity=1, align=center}}
	\tikzset{dpadinline/.style={outer sep=0.05em, inner sep=2.5pt, rounded corners=2.5pt, draw=gray!75, fill=black!08, fill opacity=1, align=center, font=\small}}

 	\tikzset{dpad/.style args={#1}{every matrix/.append style={nodes={dpadded, #1}}}}
	\tikzset{light pad/.style={outer sep=0.2em, inner sep=0.5em, draw=gray!50}}

	\tikzset{arr/.style={draw, ->, thick, shorten <=3pt, shorten >=3pt}}
	\tikzset{arr0/.style={draw, ->, thick, shorten <=0pt, shorten >=0pt}}
	\tikzset{arr1/.style={draw, ->, thick, shorten <=1pt, shorten >=1pt}}
	\tikzset{arr2/.style={draw, ->, thick, shorten <=2pt, shorten >=2pt}}

	\newcommand\lab[1]{(#1)(lab-#1)}
	\tikzset{alternative/.style args={#1|#2|#3}{name=#1, circle, fill, inner sep=1pt,label={[name={lab-#1},gray!30!black, inner sep=1pt]#3:\scriptsize #2}} }
	\tikzset{tpt/.style args={#1|#2}{alternative={#1|#2|below}} }
	\tikzset{Dom/.style args={#1[#2] (#3) around #4}{dpadded, name=#3, label={[name={lab-#3},align=center,label distance=-1.9em, shading=axis, top color=white, bottom color=black!04,inner sep=0pt, #2]150:#1}, fit={ #4 }, inner sep=0.5em}}

\relax %
    \newcommand{\nhphantom}[2]{\sbox0{\kern-2%
    \nulldelimiterspace$\left.\delimsize#1\vphantom{#2}\right.$}\hspace{-.97\wd0}}
    \makeatletter
    \newsavebox{\abcmycontentbox}
    \newcommand\DeclareDoubleDelim[5]{
    \DeclarePairedDelimiterXPP{#1}[1]%
        {%
            \sbox{\abcmycontentbox}{\ensuremath{##1}}%
        }{#2}{#5}{}%
        {%
            \nhphantom{#3}{\usebox\abcmycontentbox}%
            \hspace{1.2pt} \delimsize#3%
            \mathopen{}\usebox{\abcmycontentbox}\mathclose{}%
            \delimsize#4\hspace{1.2pt}%
            \nhphantom{#4}{\usebox\abcmycontentbox}%
        }%
    }
    \makeatother

\relax %
	\newcommand{\ssub}[1]{_{\!_{#1}\!}}
	
    \newcommand{\pdgunit}{\mathrlap{\mathit 1} \mspace{2.3mu}\mathit 1}
	
	\newcommand{\X}{\mathcal X}
	\newcommand{\V}{\mathcal V}
	
	\newcommand{\Ed}{\mathcal E}
	\newcommand{\Ar}{\mathcal A}

    \newcommand{\balpha}{\boldsymbol\alpha}
    \newcommand{\bbeta}{\boldsymbol\beta}

	\def\p_#1{\mathbb P_{\!#1\mskip-2mu}}	

	\newif\ifsuba \subatrue
	\newcommand\Src[1]{\ifsuba S\mskip-2mu\vphantom{|}_{{#1}} \else S \fi}
	\newcommand\Tgt[1]{\ifsuba T\mskip-3mu\vphantom{|}_{{#1}} \else T \fi}

	\DeclareMathAlphabet{\mathdcal}{U}{dutchcal}{m}{n}
	\DeclareMathAlphabet{\mathbdcal}{U}{dutchcal}{b}{n}
	\newcommand{\dg}[1]{\mathbdcal{#1}}

	\newcommand{\IDef}[1]{\mathit{IDef}_{\!#1}}
	\newcommand\OInc{\mathit{O\mskip-2.5muI\mskip-3.5mun\mskip-1.7muc}} %
	\newcommand{\SInc}{\mathit{S\mskip-2.5muI\mskip-3.5mun\mskip-1.7muc}} %
	\newcommand{\ed}[3]{#2%
	 {\overset{\smash{\mskip-5mu\raisebox{-1pt}{$\scriptscriptstyle
	        #1$}}}{\rightarrow}} #3}

	\newcommand{\bundle}{\mathbin{+}}
	\DeclareDoubleDelim
		\SD\{\{\}\}
	\DeclareDoubleDelim
		\bbr[[]]
	\makeatletter
	\newsavebox{\aar@content}
	\newcommand\aar{\@ifstar\aar@one@star\aar@plain}
	\newcommand\aar@one@star{\@ifstar\aar@resize{\aar@plain*}}
	\newcommand\aar@resize[1]{\sbox{\aar@content}{#1}\scaleleftright[3.8ex]
		{\Biggl\langle\!\!\!\!\Biggl\langle}{\usebox{\aar@content}}
		{\Biggr\rangle\!\!\!\!\Biggr\rangle}}
	\DeclareDoubleDelim
		\aar@plain\langle\langle\rangle\rangle
	\makeatother

\relax %
    \newcommand\bmu{\boldsymbol\mu}
    \newcommand\C{\mathdcal C}
    \newcommand\minimize{\mathop{\scalebox{0.98}{$\mathsf{minimize}$}}\limits}
    \newcommand\maximize{\mathop{\scalebox{0.98}{$\mathsf{maximize}$}}\limits}
    \newcommand\subjto{{\scalebox{0.98}{$\mathsf{subject~to}$}}}
    
    \newcommand\CI{\mathbin{\bot\!\!\!\bot}}
    
    \newcommand\actree{tree marginal}
    \newcommand\AcTree{Tree Marginal}
    \newcommand\cactree{calibrated \actree}

\usepackage{amsthm,thmtools} %
	\usepackage[noabbrev,nameinlink,capitalize]{cleveref}
    \theoremstyle{plain}
    \newtheorem{theorem}{Theorem}
	\newtheorem{coro}{Corollary}[theorem]
    \newtheorem{prop}[theorem]{Proposition}

	\declaretheorem[numberwithin=theorem,name=Claim]{iclaim}
    
    \newtheorem{lemma}[theorem]{Lemma}
    \theoremstyle{definition}
    \declaretheorem[name=Definition, qed=$\square$]{defn}

	\crefname{defn}{Definition}{Definitions}
	\crefname{prop}{Proposition}{Propositions}
    \crefname{issue}{Issue}{Issues}
    \crefname{claim}{Claim}{Claims}

\relax %
	\usepackage{xpatch}
	\makeatletter
	\xpatchcmd{\thmt@restatable}%
	   {\csname #2\@xa\endcsname\ifx\@nx#1\@nx\else[{#1}]\fi}%
	   {\ifthmt@thisistheone\csname #2\@xa\endcsname\ifx\@nx#1\@nx\else[{#1}]\fi
	   \else\fi}
	   {}{\typeout{FIRST PATCH TO THM RESTATE FAILED}} %
	\xpatchcmd{\thmt@restatable}%
	   {\csname end#2\endcsname}
	   {\ifthmt@thisistheone\csname end#2\endcsname\else\fi}
	   {}{\typeout{FAILED SECOND THMT RESTATE PATCH}}
       
    \newcommand\onlyfirsttime[1]{\ifthmt@thisistheone#1\else\fi}

	\newcommand{\recall}[1]{\medskip\par\noindent{\bf \Cref{thmt@@#1}.} \begingroup\em \noindent
	   \expandafter\csname#1\endcsname* \endgroup\par\smallskip%
       }

	\newenvironment{linked}[3][]{%
		\def\linkedproof{#3}%
		\def\linkedtype{#2}%
		\ifmarginprooflinks
		\marginpar{%
			\vspace{1.5em}
			\centering%
			\hyperref[proof:\linkedproof]{%
			\color{blue!30!white}%
			\scaleleftright{$\Big[$}{\,\mbox{\tiny\centering\tt\begin{tabular}{@{}c@{}}
				link to\\[-0.15em]
				proof
			\end{tabular}}\,}{$\Big]$}}~
			}%
		\fi
        \restatable[#1]{#2}{#2:#3}\label{#2:#3}%
        }%
		{\endrestatable%
		}
	\makeatother
		\newcounter{proofcntr}
		\newenvironment{lproof}{\begin{proof}\refstepcounter{proofcntr}}{\end{proof}}

        \usepackage{cancel} %
        \newcommand{\Cancel}[2][red]{{\color{#1}\cancel{\color{black}#2}}}

		\usepackage{tcolorbox}
		\tcbuselibrary{most}
        \colorlet{proofleft}{blue!20!black!40!white}
        \colorlet{proofmatt}{blue!20!black!05!white}
		\tcolorboxenvironment{lproof}{
			enhanced,
			parbox=false,
			boxrule=0pt,
			frame hidden,
			borderline west={4pt}{0pt}{proofleft},
			colback={proofmatt},
			sharp corners,
			breakable,
		}
        \tikzset{proofmatDom/.style args={#1[#2] (#3) around #4}{dpadded, name=#3, label={[name={lab-#3},align=center,label distance=-1.9em, shading=axis, top color=proofmatt, bottom color=black!04!proofmatt,inner sep=0pt, #2]150:#1}, fit={ #4 }, inner sep=0.5em}}

\relax %
    \newcommand{\TODO}[1][INCOMPLETE]{{\color{red}\raggedright\hangindent=0.5cm\rightskip=0.4cm$\smash{\Big\langle}$~\texttt{#1}~\raisebox{-0.3ex}{${\Big\rangle}$}\hspace{-1.5cm}\par}}

    \newcommand*{\daggerfootnote}[1]{%
        \renewcommand*{\thefootnote}{\fnsymbol{footnote}}%
        \footnote[2]{#1}%
        \renewcommand*{\thefootnote}{\arabic{footnote}}%
        }
    \makeatletter
    \newcommand\footnoteref[1]{\protected@xdef\@thefnmark{\ref{#1}}\@footnotemark}
    \makeatother
    \crefformat{footnote}{#2\footnotemark[#1]#3}

\relax %

	\usepackage[american]{babel}
	\usepackage{csquotes}

    \ifbiblatex
	\usepackage[backend=biber, style=authoryear]{biblatex}
	\DeclareLanguageMapping{american}{american-apa}
	\DeclareFieldFormat{citehyperref}{%
	  \DeclareFieldAlias{bibhyperref}{noformat}%
	  \bibhyperref{#1}}

	\DeclareFieldFormat{textcitehyperref}{%
	  \DeclareFieldAlias{bibhyperref}{noformat}%
	  \bibhyperref{%
	    #1%
	    \ifbool{cbx:parens}
	      {\bibcloseparen\global\boolfalse{cbx:parens}}
	      {}}}

	\savebibmacro{cite}
	\savebibmacro{textcite}

	\renewbibmacro*{cite}{%
	  \printtext[citehyperref]{%
	    \restorebibmacro{cite}%
	    \usebibmacro{cite}}}

	\renewbibmacro*{textcite}{%
	  \ifboolexpr{
	    ( not test {\iffieldundef{prenote}} and
	      test {\ifnumequal{\value{citecount}}{1}} )
	    or
	    ( not test {\iffieldundef{postnote}} and
	      test {\ifnumequal{\value{citecount}}{\value{citetotal}}} )
	  }
	    {\DeclareFieldAlias{textcitehyperref}{noformat}}
	    {}%
	  \printtext[textcitehyperref]{%
	    \restorebibmacro{textcite}%
	    \usebibmacro{textcite}}}

	\DeclareCiteCommand{\brakcite}
	  {\usebibmacro{prenote}}
	  {\usebibmacro{citeindex}%
	   \printtext[bibhyperref]{[\usebibmacro{cite}]}}
	  {\multicitedelim}
	  {\usebibmacro{postnote}}
      
     \else
     \gdef\parencite{\citep}
     \gdef\textcite{\citet}
     \fi

\newcommand\discard[1]{}

\newcommand\zogamma{{\mathrlap{\raisebox{-0.1ex}{$\hat{\phantom{x}}$}}\gamma}}

\colorlet{mayyybe}{blue!50!red!20!white}
\colorlet{rewrite}{purple!80!black}

\colorlet{olicolor}{blue!50!red!70!black}
\colorlet{joecolor}{green!50!blue!70!black}

\newcommand\vjoe[1]{{\color{joecolor}\textbf{$\boldsymbol\{$Joe: }#1 \textbf{$\boldsymbol\}$}}} 
\newcommand\voli[1]{{\color{olicolor}\textbf{$\boldsymbol\{$Oli: }#1 \textbf{$\boldsymbol\}$}}}

\title{Inference for Probabilistic Dependency Graphs}

\author[1]{Oliver~E.~Richardson}
\author[1]{Joseph~Y.~Halpern}
\author[1]{Christopher De Sa}
\affil[1]{%
    Deparment of Computer Science\\
    ~Cornell University\\
    ~Ithaca NY 14853
}
  
\begin{document}
\maketitle

\begin{abstract}
    Probabilistic dependency graphs (PDGs)
    are a flexible class of probabilistic graphical models,
    subsuming Bayesian Networks and Factor Graphs.
    They can also capture inconsistent beliefs, and provide a way of measuring the degree of this inconsistency.
    We present the first tractable inference algorithm for
    PDGs with discrete variables,
    making the asymptotic complexity of PDG inference similar
    that of the graphical models they generalize. 
    The key components are:
    (1) 
    the observation that, in many cases, the distribution a PDG specifies 
    can be formulated as a convex optimization problem 
        (with exponential cone constraints), 
    (2) a construction that allows us to express these problems compactly for PDGs of boundeed treewidth, 
    (3) contributions to the theory of PDGs that justify the construction, 
    and
    (4) an appeal to interior point methods that can solve such problems in polynomial time.
    We verify the correctness and complexity of our approach, 
    and provide an implementation of it.
    We then evaluate our implementation, and demonstrate that 
    it outperforms
    baseline approaches.
    Our code is available at {\small\url{github.com/orichardson/pdg-infer-uai}}.
\end{abstract}

\section{Introduction}

\emph{Probabilistic dependency graphs (PDGs)} \parencite{pdg-aaai},
form a very general class of probabilistic graphical models,
that includes not only
Bayesian Networks (BNs) and
Factor Graphs (FGs),
but also more recent statistical models built out of neural networks,
such as Variational Autoencoders (VAEs) \parencite{kingma2013autoencoding}.
\expandafter\discard\vjoe{
 They can also capture inconsistent beliefs; moreover, we can measure how
 inconsistent a probability distribution is with a PDG.  This allows us
 to define the dgree of inconsistent of a PDG to essentially be that of the
 probability measure that is least inconsistent with it.}%
\expandafter\discard\voli{They can also capture inconsistent beliefs, and moreover provide a natural measurement of the degree of this inconsistency. }%
\expandafter\discard\vjoe{They can also capture inconsistent beliefs; moreover, we can measure how inconsistent a probability distribution is with a PDG. This allows us to define the degree of inconsistenty of a PDG to essentially be that of the probability measure that is least inconsistent with it.}%
\expandafter\discard\voli{%
    They can also contain inconsistent beliefs, such as two different probabilities over the same variable. Moreover, there is a natural way to measure the degree of this inconsistency: starting with a measure of how incompatible a joint distribution is with a PDG, the inconsistency of the PDG is the smallest incompatibility with any joint distribution.  }%
\discard{%
    They can also capture inconsistent beliefs,
    Moreover, there is a useful way to quantify the degree of this inconsistency:
    the discrepancy between the PDG, and the 
        probability measure that 
        best represents it.
     }%
PDGs can also capture inconsistent 
beliefs, and provide a useful way to measure the degree of this inconsistency;
for a VAE, this is the loss function used in training
    \parencite{one-true-loss}.
    PDGs have some
significant advantages over other representations of probabilistic information. 
    Their flexibility allows them to model beliefs that BNs cannot, 
such as information from independent studies of the same variable
(perhaps with different controls, yielding probabilistic observations $p(Y|X)$ and $q(Y|Z)$).
PDGs
can deal gracefully with conflicting information from multiple sources.
Every subcomponent of a PDG has probabilistic
    meaning, independent of the other components;
compared to FGs, this makes PDGs more interpretable.
But up to now, there has been no practical way to do inference for
PDGs---that is,
to answer questions of the form
``what is the probability of $Y$ given $X$?''
This paper presents the first algorithm to do so.

Before discussing our algorithm, 
we must discuss what it even means to do inference for a PDG.  
A BN or FG represents 
a unique joint distribution. 
Thus, for example, when we ask ``what is the probability of $Y$ given that $X{=}x$?''
\def\BNPr{\mu}
in a BN, we mean ``what is $\BNPr(Y | X{=}x)$?'' for the probability 
measure $\BNPr$ that the BN represents.
But a PDG might, in general, represent more than just one distribution.

Like a BN, a PDG encodes
two types of information: ``structural'' 
information about the independence of causal mechanisms,
and ``observational'' information
about conditional probabilities.
Unlike in a BN, the two can conflict in a PDG.
Corresponding to these two types of information,
a PDG has two
loss functions,
which quantify how far a distribution $\mu$ is from
modeling the information of each type.
Given a number $\zogamma
\in [0,1]
$
indicating the importance of structure relative to observation,
we take the \emph{$\zogamma$-semantics} of a PDG to be the
set of distributions that minimize 
the appropriate convex combination of
losses.
We also consider the \emph{$0^+$\!-semantics}: the limiting case that
arises as $\zogamma$ goes to zero
(which focuses
 on observation, using structure only to break ties).
This set
can be shown to contain precisely one distribution
for PDGs satisfying a mild regularity condition 
(required by definition by \citeauthor{pdg-aaai});
we call such PDGs \emph{proper}.
Thus, we have
 a parameterized family of inference notions:
to do $\zogamma$-inference, for $\zogamma \in [0,1] \cup \{0^+\}$,
is to answer queries in a way that is true of all distributions in the $\zogamma$-semantics.

If there are distributions
fully
consistent with
both the observational and the structural information
in a PDG $\dg M$, 
then for $\zogamma \in (0,1) \cup \{0^+\}$, all
notions of $\zogamma$-inference 
coincide.
\discard{
    For PDGs satisfying a mild condition 
    (required by definition in \citeauthor{pdg-aaai})
    which for now we call \emph{proper},
    selecting a small enough positive value of $\zogamma$
    suffices to ensure that
    the $\zogamma$-semantics consists of only a single distribution.
    The $0^+$ semantics of a proper PDG is also a unique distribution,
    and one that does not depend on the choice of a small positive number. }%
\def\PrM{\mu_{\dg M}}%
If $\dg M$ is also proper,
    this means there is
    a single distribution $\PrM$
    that minimizes both loss functions, 
    in which case we want to answer queries with respect to $\PrM$
    no matter how we weight observational and structural information. 
Moreover, if $\dg M$ represents a BN,
then $\PrM$ is the distribution represented by the BN.  
However, if there is no distribution that is consistent with both types of information, then the choice of $\zogamma$ matters.  

Since PDGs subsume BNs, and inference for BNs is already NP-hard, the same must be true of PDGs.
At a high level, the best we could hope for would be tractability on the restricted
class of models on which inference has traditionally been tractable---that is, a polynomial algorithm for models whose
underlying structure has \emph{bounded treewidth} (see
\Cref{sec:tw} for formal definitions).
That is indeed what we have.  
More precisely, we show that
$0^+$\!-inference
and $\zogamma$-inference for small $\zogamma$ 
can be done 
for discrete PDGs of bounded treewidth containing $N$ variables in
$\tilde O(N^4)$ 
time.
 
Our algorithm
is based on a line of recent work in 
convex programming
that establishes
polynomial-time
for a class of optimization problems called \emph{exponential conic programs}
\parencite{badenbroek2021algorithm,skajaa2015homogeneous,nesterov1996infeasible}.
Our contribution is to show that the problem of inference in a PDG
of bounded treewidth
can be efficiently converted to a (sequence of) exponential conic program(s), at which point it can be solved with a commercial solver
(e.g., \textcite{mosek}) in polynomial time. 
The direct appeal to a solver allows us
to benefit from the speed and reliability of such highly optimized solvers, and also from future improvements in exponential conic optimization.
Thus, our result is not only a theoretical one, but practical as well.

Beyond its role as a probabilistic model,
a PDG is also of interest for its degree of inconsistency---%
that is, the minimum value of its loss function. 
As shown by 
\textcite{one-true-loss},
many loss functions and statistical divergences
can be viewed as measuring 
the inconsistency
of a PDG that models the context appropriately.
This makes calculating this minimum value of interest%
---but up to now, there has been no way to do so.
There is a deep connection between this problem and PDG inference;
for now, we remark that this number is a byproduct of our techniques.

\expandafter\discard\voli{
Beyond its role as a probabilsitic model,
a PDG is also of interest for its \emph{degree of inconsistency};
As shown by \textcite{one-true-loss},
    many loss functions and statistical divergences
    can be viewed as measuring the inconsistency
    of a PDG that models the appropriate context.
It follows that the training process in machine learning can
    be conceptualized as
    adjusting parameters of cpds so as to minimize the inconsistency of a PDG.
        But how {does} one \emph{calculate} this degree of inconsistency?
This problem turns out to be closely related to that of
    inference in PDGs, and our approach addresses both.
}

\textbf{Contributions.}
We provide the first algorithm for inference in a PDG;
in addition, it calculates a PDG's degree of inconsistency. 
We prove that 
our algorithm
is correct, and also
fixed-parameter tractable: for PDGs of bounded treewidth,
it runs in polynomial time.
We also prove that PDG inference and inconsistency 
    calculation are equivalent problems.
Our algorithm reduces inference in PDGs to exponential conic programming
in a way that can be offloaded to powerful existing solvers.
We provide an implementation of this reduction in a
standard convex optimization framework, giving users an
interface between such solvers and the standard PDG Python library.
Finally, we evaluate our implementation. The
    results suggest our method is faster and 
    significantly more reliable than simple baseline approaches.

\section{Preliminaries \& Related Work}

\textbf{Vector Notation.}
For us, a vector is a map from a finite set $S$, called its 
\emph{shape},
to the extended reals $\Rext := \mathbb R \cup \{\infty\}$.
We write $\mat u := [u_i]_{i \in S}$ to define a vector $\mat u$ by its components.
\discard{\color{gray!30!white}
    We will sometimes use superscripts as well, especially when indices depend on one another. For example, if $\dg S$ is a finite set of finite sets, then
    $[u^S_s]^{S \in \dg S}_{s \in S}$ denotes a vector which has an element
    for each pair $(S,s)$, satisfying $s \in S \in \dg S$.
    By supplying just the upper index of such a vector, as in $\mat u^{S_0}$,
    we mean $[u^{S_0}_s]_{s \in {S_0}}$, the projection of $\mat u$ onto the subspace whose upper index is $S_0$.
}
Vectors of the same shape
can be added (+), partially ordered ($\le$), or multiplied ($\odot$) pointwise as usual.
$\mat 1$ denotes an all-ones vector, of a shape implied by context.
\discard{\vfull{
$\mat u^{\sf T}$ denotes the transpose of $\mat u$, and is used to
express the inner product $\mat u^{\sf T} \mat v$ of vectors $\mat u$ and
$\mat v$ of the same shape.}}

\discard{
    \color{gray!30!white} If $\mat u = [u_a]_{a \in A}$ is a vector over $A$ and $\mat v = [v_b]_{b \in B}$ is a vector over $B$, then $\mat u \mathbin{\otimes} \mat v := [ u_a \cdot v_b ]_{a \in A, b \in B}$ is a vector over $A \times B$. }

    \textbf{Probabilities.}
We write $\Delta S$ to denote the set of probability distributions over a finite set $S$.
Every variable $X$ can take on values from a finite set
$\V\mskip-1.5mu X$
of possible values.
We can regard sets of variables $\mat X$ as variables themselves, with
$\V \mat X = \Pi_{X \in \mat X} \V X$.
A conditional probability distribution (cpd) $p(Y|X)$ is a map
\ifvfull
$p : \V\mskip-1.5mu  X \to \Delta \V Y$ assigning to each $x \in \V\mskip-1.5mu X$ a
probability distribution 
$p(Y|x) \in \Delta \V Y$, which is shorthand for $p(Y|X{=}x)$.
\else
$p$ that assigns each $x \in \V\mskip-1.5mu X$ a probability distribution
$p(Y|X{=}x) \in \Delta \V Y$.
\fi
\ifvfull
Given a distribution $\mu$ over (the values of) a set of variables including $X$ and $Y$,
\else
Given a joint distribution $\mu$,
\fi
we write $\mu(X)$ for its marginal 
on $X$,
and $\mu(Y|X)$ for the cpd obtained by 
\ifvfull
conditioning on $X$ and marginalizing to $Y$.
We also refer to $\mu$'s entropy
    $\H(\mu) := \Ex_{\mu} [\log \frac1\mu]$ and 
    conditional entropy $\H_\mu(Y|X) := \Ex_\mu[\log\nicefrac1{\mu(Y|X)}]$
    of $Y$ given $X$.
\else
conditioning on $X$ and marginalizing to $Y$.
We also refer to $\mu$'s entropy $\H(\mu) := \Ex_{\mu} [\log \frac1\mu]$ and conditional entropy $\H_\mu(Y|X) := \Ex_\mu[\log\nicefrac1{\mu(Y|X)}]$.
\fi

\textbf{Hypergraphs and Treewidth.} \label{sec:tw}
A hypergraph 
$
(V, \Ed)$ is a set $V$ of vertices and a
set
 $\Ed$ of \emph{hyperedges}, which correspond to subsets of $V$.
An ordinary graph may be viewed as the special case in which every hyperedge contains  two vertices.

\begin{defn}
    A \emph{directed hypergraph}
    $(N, \mathcal A)$ is a set of nodes $N$, and
    a set
    of \emph{(hyper)arcs} $\mathcal A$,
    each $a \in \mathcal A$ of which
    is associated with 
    a set of source nodes $\Src a \subseteq N$,
    and target nodes $\Tgt a \subseteq N$.
    We also write $\ed {a}{S}{T} \in \Ar$ to specify an
    arc $a$ together with its sources $S = \Src a$ and targets $T = \Tgt a$.
\end{defn}

A directed hypergraph 
can be viewed as
    a hypergraph
by joining
each source and target set,
thereby ``forgetting'' the direction of the arrow.
\ifvfull
Thus, notions defined for undirected hypergraphs (like that of
treewidth, which we now review), can be 
    applied to directed hypergraphs as well.
\fi

Many problems that are intractable for general graphs
are tractable for trees, and
some graphs are closer to being trees than others.
A tree decomposition of a (hyper)graph $G = (V, \Ed)$ is a tree $(\C, \mathcal T)$ whose vertices $C \in \C$, called
\emph{clusters}, are subsets of $V$ such that:

\begin{enumerate}[nosep]
    \item every vertex $v \in V$ and every hyperedge $E \in \Ed$ is contained in at least one cluster, and
        \item every cluster $D$ along the unique path from $C_1$ to $C_2$ in $\cal T$,
         contains $C_1 \cap C_2$.
\end{enumerate}

The \emph{width} of a tree decomposition is one less than the size of its largest cluster,
and the \emph{treewidth} of a (hyper)graph $G$ is the smallest possible width of any tree decomposition of $G$.
It is NP-hard to determine the tree-width of a graph, but
if the tree-width is known to be bounded above, a tree decomposition may be constructed in linear time \parencite{bodlaender1993linear}.
For graphs of bounded tree-width, many problems 
(indeed, any problem expressible in a certain second-order logic \parencite{courcelle1990})
can be solved in
linear time.
This is also true of inference in 
standard graphical models.

\textbf{Graphical Models and Inference.}
A \emph{graphical model structure}
is a (directed) (hyper)graph whose vertices $\X$ are variables, and whose (hyper)edges 
somehow
indicate local influences between variables.
A \emph{probabilistic graphical model},
or simply  ``graphical model'',
is a
graphical model structure
together with quantitative information about these local influences.
Semantically,
a graphical model $\cal M$
typically
represents a joint probability distribution $\Pr_{\!\cal M}
 \in \Delta \V\!\X$ over its variables.
\discard{
    Although there is often more to the story,
    it can typically be
    expressed as a product
    $\Pr_{\!\cal M}(\X) \propto \prod_{E \in \Ed} \phi_{E}(E)$
    of factors $\boldsymbol\phi = 
    \{ \phi_E : \V E \to \mathbb R_{\ge 0} \}_{E \in \Ed}$
    over a hypergraph $(\X, \Ed)$ closely related to the structure of $\cal M$.
    For this reason, some authors use the term ``graphical model'' to refer to a tuple $(\X ,\Ed, \boldsymbol\phi)$,
    i.e., a factor graph.
    PDGs, however, do not represent probabilities this way.}
Inference for $\cal M$ is then the ability to calculate cpds $\Pr_{\!\cal M}(Y|X{=}x)$,
where $X,Y \subset \X$ and $x \in \V\! X$. 

\discard{
    To do inference in probabilistic model $\cal M$ is to answer probabilistic queries, of the form
    \textit{``what is the distribution of variables $Y$, given that $X\!=\!x$?''}
    If $\cal M$ represents the joint distribution $\Pr_{\cal M}$, then the
    appropriate answer to this question is $\Pr_{\cal M}(Y \mid X\!=\!x)$.}

Many inference algorithms (such as belief propagation),
when applied to tree-like graphical models,
run in linear time and are provably correct.
If the same algorithms are na{\"i}vely applied to graphs with cycles (as in loopy belief propagation),
then they may not converge, and even if they do,
may give incorrect (or even inconsistent) answers
\parencite{wainwright2008graphical}.
Nearly all exact inference algorithms
(including variable elimination  \parencite{bertele1972nonserial},
 message-passing with \parencite{lauritzen1988local}
    and without division \parencite{shafer1990probability},
    among others \parencite{wainwright2003tree})
effectively construct a tree decomposition, and can be
viewed as running on a tree \parencite[\S9-11]{koller2009probabilistic}.
Indeed,
under widely believed assumptions,
every class of graphical models
for which (exact) inference is \emph{not} NP-hard
has bounded treewidth
\parencite{chandrasekaran2012complexity}.

Given a tree decomposition $(\C, \mathcal T)$ of the 
underlying model structure,
many of these algorithms
use a 
standard
data structure
that we will call a \emph{\actree},
which
is a collection
$\bmu = \{\mu_C(C)\}_{ C \in \C}$ 
of probabilities over the clusters \parencite[\S10]{koller2009probabilistic}.
A \actree\
$\bmu$ 
is said to be \emph{calibrated} if neighboring clusters' distributions agree on the variables 
they share.
In this case, 
$\bmu$ determines a joint distribution by
\begin{equation}
    \Pr\nolimits_{\bmu}
    (\X)
        = \quad 
        {\prod_{\mathclap{C \in \C}} \mu_C(C)~\,}\Big/
        {~~\prod_{\mathclap{(C{-}D) \mathrlap{\in \cal T}}} \mu_{C}(C \cap D)\,,}
    \label{eq:cliquedist}
\end{equation}
which has the property that $\Pr_{\bmu}(C) = \mu_C$ for 
all
$C \in \C$.
\discard{%
\vfull{A \cactree\ $\bmu$ summarizes the answers to 
many queries about $\Pr_{\bmu}$ at once. 
To see why
in a simple case, note that 
for an unconditional query about $Y$ contained within a single cluster $C$, we have
$\Pr_{\bmu}(Y) = \mu_C(Y)$.
\discard{%
    Note also that if $\C = \{ \X \}$ just contains one big cluster, 
    a \actree\ is an explicit joint distribution, reducing queries to summation.}%
With some care, the general idea can be extended to arbitrary queries 
    \parencite[see][\S 10.3.3]{koller2009probabilistic};
    those conditional on evidence $X{=}x$ can be handled
    by conditioning the clusters that contain $X$,
    and then recalibrating $\bmu$ with 
    a standard algorithm like belief propagation.
    }
}
A \cactree\ summarizes the answers to
queries about $\Pr_{\bmu}$
\parencite[see][\S 10.3.3]{koller2009probabilistic}.
Therefore, to answer probabilistic queries with respect to a distribution $\mu$, it suffices to find a \cactree\ $\bmu$ that represents $\mu$, and appeal to standard algorithms.

\textbf{Probabilistic Dependency Graphs.}
Our presentation of PDGs is slightly
different from (but equivalent to) that of
\textcite{pdg-aaai}, which
the reader is encouraged to consult for more details and intuition.
At a high level, a PDG
is
just
a collection of cpds and causal assertions,
    weighted by confidence. More precisely:

\begin{defn}
    A PDG $\dg M \!=\! (\X\mskip-2mu, \Ar,
        \mathbb P, 
        \balpha, \bbeta )
    $
    is     
    a directed hypergraph 
    $(\X\mskip-2mu, \Ar)$ 
    whose nodes 
    are
    variables,
    together with 
    probabilities $\mathbb P$
    and
    confidence vectors
    $\balpha \!=\! [\alpha_a]_{a \in \Ar}$ 
    and $\bbeta \!=\! [\beta_a]_{a \in \Ar}$,
    so that
    each $\ed aST \! \in\! \Ar$ is associated with:
    
    \begin{itemize}[nosep,itemsep=2pt]
    \item
    a conditional probability distribution
    {\subafalse $\p_a(\Tgt a | \Src a)$}
    on the target variables given 
    values of
    the source variables,
    \item a weight 
    $\beta_a \in \smash{\Rext}$ 
    indicating
    the modeler's confidence in 
    the cpd {\subafalse $\p_a(\Tgt a | \Src a)$},
    \discard{(as measured by the number of independent observations that support $\p_a$), }
    and
    \item 
    a weight $\smash{\alpha_a \in \mathbb R}$
    indicating
    the modeler's confidence in the functional dependence of 
    {\subafalse$\Tgt a\mskip-2mu$ on $\Src a\mskip-2mu$}
    expressed by
    $a$.
    \discard{
    (as measured by the expected number of independent causal mechanisms corresponding to $a$,
    that determine $\Tgt a$ given $\Src a$).%
    }
    \end{itemize}
\expandafter\discard\voli{In aggregate, $\balpha = [\alpha_a]_{a \in \Ar} \in \Rext^\Ar$ and $\bbeta = [\beta_a]_{a \in \Ar} \in \Rext^\Ar$ are the vector forms of the weights, and
    $\mathcal P$ is the set of cpds indexed by $\Ar$.}
If $\bbeta \ge \mat 0$ and $\alpha_a\! > 0$ implies $\beta_a\! > 0$, we
write $\bbeta \gg \balpha$ and
call $\dg M$ \emph{proper}.
Note that  $\bbeta \gg \balpha$ if $\bbeta > \mat 0$.
\end{defn}

One significant advantage of PDGs is their modularity:
we can combine the information in $\dg M_1$ and $\dg M_2$ 
by taking the union of their variables and the disjoint union of their arcs (and associated data) to get a new PDG, denoted $\dg M_1 + \dg M_2$.

\expandafter\discard\voli{%
    Like other graphical models,
    PDGs encode two types of information: ``structural'' information 
    through the graphical structure $\Ar$ and weights $\balpha$,
    as well as ``observational'' information, 
    through the conditional probability distributions
    $\mathcal P$ and weights $\bbeta$. 
    Corresponding to these two types of information, 
    PDG semantics are based on two scoring functions 
    which quantify the discrepancy between 
    a joint distribution $\mu(\X)$ over all variables,
    and each of the two types of information.
}%
\expandafter\discard\vjoe{%
    As we mentioned in the introduction,
    PDGs encode two types of information: ``structural'' information 
    through the graphical structure $\Ar$, and
    ``observational'' information, 
    through the conditional probability distributions.
    The weight vectors $\alpha$ and $\beta$ encode our confidence in these
    two types of information.
    The semantics of PDGs are based on two scoring functions
    that quantify the discrepancy between 
    a joint distribution $\mu$ over (the possible values of) the variables
    in $\X$, and each of the two types of information. }%
As mentioned in the introduction,
a PDG contains two types of information:
``structural'' information, in the hypergraph $\Ar$ and
weights $\balpha$, and ``observational'' data,
in the cpds  $\mathbb P$ and weights $\bbeta$.
PDG semantics are based on two scoring functions
that quantify discrepancy between 
each type of information and a distribution
$\mu \in \Delta \V \!\X$ over its variables.

The \emph{observational incompatibility} of $\mu$ with $\dg M$, which
can be thought of as a ``distance''  between $\mu$ and the cpds of $\dg M$,
is given by the weighted sum of relative entropies:
\begin{align*}
    \OInc_{\dg M}(\mu) :=
        \sum_{\ed aST \mathrlap{\,\in \Ar}} \subafalse
        \beta_a\, \kldiv[\Big]{\mu(\Tgt a,\Src a)}{\p_a(\Tgt a | \Src a) \mu(\Src a)}.
\end{align*}
Under a standard interpretation of the relative entropy $\kldiv{\mu}{p} \mskip-1.5mu=\mskip-1.5mu \Ex_{\mu}[\log \frac\mu p]$,
$\OInc_{\dg M}$ measures the excess cost of using codes optimized for the cpds of $\dg M$ 
(weighted by their confidences),
when reality is
distributed according to $\mu$.

The second scoring function measures
the extent to which
$\mu$ 
is incompatible with 
a causal picture consisting of independent mechanisms 
along each hyperarc. 
This is captured by the
\emph{structural incompatibility}
(of $\mu$ with $\dg M$), 
and given by
\begin{equation*}
    \SInc_{\dg M}(\mu) := \,
        \pqty[\Big]{\; \sum_{\ed aST \mathrlap{\,\in \Ar}}\subafalse \alpha_a\, \H_\mu(\Tgt a | \Src a) } - \H(\mu).
\end{equation*}
Note that
$\SInc_{\dg M}$
does not depend on the cpds
of $\dg M$, nor the possible values of the 
variables; it
is defined purely in terms of
the weighted hypergraph structure $(\Ar,\balpha)$.

If the observational and structural information conflict, 
then the distribution(s) that best represent a PDG 
will depend
on the importance of structure relative to observation,
\discard{ This is captured by a trade-off 
    parameter $\gamma \ge 0$, which 
    can be used to define the scoring function
    $\bbr{\dg M}_\gamma: \Delta \V\!\X \to \Rext$, as follows:}%
as captured by a trade-off parameter $\zogamma \in [0,1]$
that controls the convex combination
$(1-\zogamma)\OInc + \zogamma \SInc$. 
So as to simplify the math
and match the notation in
previous work (\citeyear{pdg-aaai,one-true-loss}),
we mostly use a rescaled variant with a different parameterization.
Using
$\gamma := \nicefrac{\zogamma}{(1-\zogamma)} \in [0,\infty]$,
define the overall scoring function:
\begin{align*}
    \bbr{\dg M}_\gamma&(\mu) 
        := \OInc_{\dg M}(\mu) + \gamma \, \SInc_{\dg M}(\mu)
            \numberthis\label{eqn:scoring-fn} \\[-0.2ex]
        &\!\!\!= \scalebox{0.85}{$\displaystyle\frac{1}{1-\zogamma}$} \Big(\, (1-\zogamma) \OInc_{\dg M}(\mu) + \zogamma \, \SInc_{\dg M}(\mu)\, \Big) \\[-0.1ex]
        &= \Ex\nolimits_{\mu}\bigg[
            \,
            \sum_{\ed aST \mathrlap{\, \in \Ar}} \subafalse
            \log \frac
            {\mu(\Tgt a| \Src a)^{\beta_a - \gamma \alpha_a}}
            {\p_a(\Tgt a | \Src a)^{\beta_a}}
        \bigg] - \gamma \H(\mu)
        .
\end{align*}
Let $\bbr{\dg M}^*_\gamma := \argmin_\mu \bbr{\dg M}_\gamma(\mu)$ denote
the set of optimal distributions at a particular value $\gamma$.
One natural conception of inference in PDGs is then parameterized by
$\zogamma$:
to do $\zogamma$-inference
in $\dg M$ is to respond to probabilistic queries in a way that is sound with respect to every $\mu \in \bbr{\dg M}^*_\gamma$.
It is not too difficult to see that when $\bbeta \ge \gamma\balpha$, 
 \eqref{eqn:scoring-fn} is strictly convex, which ensures that
 $\bbr{\dg M}^*_\gamma$ is a singleton.
This paper demonstrates that
$\zogamma$-inference
is tractable for such PDGs.
\discard{%
    The former is a notational convenience,
    because for $\gamma \in (0, \infty)$,
    $\gamma$-inference is just $1$-inference for a
    slightly different PDG:
    $\bbr{
        \balpha, \bbeta}^*_\gamma = \bbr{
        \gamma \balpha, \bbeta}^*_1$.
    This paper demonstrates the tractability of
    1-inference for the specific case of PDGs satisfying $\bbeta \ge \balpha$, which is sufficient to ensure strict convexity of \eqref{eqn:scoring-fn}, and hence a unique optimal distribution.}%

The limiting behavior of the $\zogamma$-semantics as $\zogamma \to 0$,
which we denote $\bbr{\dg M}^*_{0^+}$ and call the \emph{$0^+$\!-semantics},
has some special properties.
If $\dg M$ is proper, then        
$\bbr{\dg M}^*_{0^+}$ contains precisely one distribution.
This distribution intuitively
reflects an extreme empirical
view: observational data trumps causal structure.
Note that in the absence of a causal picture
($\balpha = \mat0$), this
corresponds to the well-established 
practice of selecting
the maximum entropy distribution consistent with some observational constraints \parencite{jaynes1957information}.
One should be careful to distinguish 
$\bbr{\dg M}^*_{0^+}$
from $\bbr{\dg M}^*_0$,
the set of distributions that minimize
$\OInc_{\dg M}$; the latter  set
includes
 $\bbr{\dg M}^*_{0^+}$
\parencite[Prop 3.4]{pdg-aaai},
but may also contain other distributions.
\expandafter\discard\voli{
One of our main goals is to answer probabilistic queries with respect to $\bbr{\dg M}^*$, which we call \emph{$\epsilon$-inference}, 
since it coincides with $\gamma$-inference for $\gamma$ equal to an infinitessimal number $\epsilon$.}
This paper also shows how to efficiently answer queries with respect 
to the unique distribution in $\bbr{\dg M}^*_{0^+}$, which we call
\emph{$0^+$\!-inference}.

Given a PDG $\dg M$, the smallest possible value of its scoring function,
$
    \aar{\dg M}_\gamma := \inf_{\mu}\, \bbr{\dg M}_\gamma(\mu),
$
is known as its $\gamma$-inconsistency
and is interesting in its own right:
$\aar{\,\cdot\,}_\gamma$ 
is arguably a ``universal'' loss
function \parencite{one-true-loss}.

\textbf{Interior-Point Methods and Convex Optimization.}
Interior-point methods provide an iterative way of approximately solving linear programs in polynomial time \parencite{karmarkar1984new}.
With the theory of ``symmetric cones'', these methods were extended in the 1990s to handle second-order cone programs (SOCPs) and semidefinite programs (SDPs), which allow more expressive constraints.
But the constraints that these methods can handle are insufficient for
our purposes. We need what have been called \emph{exponential cone constraints}.
The \emph{exponential cone} is the convex set
\begin{align*}
        \begin{aligned}
        K_{\mskip-1mu\exp} :=
        \big\{ (x_1, x_2, x_3) &:
                x_1 \ge x_2 e^{x_3 / x_2}\!,\, x_2 > 0 \big\}
            \\[-0.6ex]\quad \mathbin{\cup}\,
        \big\{ (x_1, 0, x_3&) : x_1 \ge 0,\, x_3 \le 0 \big\}
        \qquad \subset \mathbb \Rext^3.
    \end{aligned}
\end{align*}
\discard{\voli{It is also sometimes called the ``relative entropy'' cone, because if $\mat m, \mat p \in \Delta^{n-1} \subset \mathbb R^n$ are points on a probability simplex, then $(-\mat u, \mat m, \mat p) \in K_{\exp}^n$ if and only if $\mat u$ is an upper bound on $\mat m \log (\nicefrac{\mat m}{\mat p})$, the pointwise contribution to relative entropy at each outcome.}}
Let $K \mskip-2mu := \mskip-2mu K_{\exp}^p \mskip-2mu \times \mskip-1mu
[0, \infty]^q 
\subset \smash\Rext^n$ be a product of $p$ exponential cones and $q = n - 3k$ non-negative orthants.
An \emph{exponential conic program} is then an optimization problem of the form
\begin{equation}
    \minimize
        _{\mat x}
        ~~ \mat c^{\sf T} \mat x
    \quad\subjto
    ~~ A \mat x = \mat b,~~\mat x \in 
        K,
        \label{eq:exp-conic-program}
\end{equation}
where $\mat c \in \Rext^{n}$ is some cost vector,
the function $\mat x \mapsto \mat c^{\sf T} \mat x$ is called the \emph{objective},
and $\mat b \in \Rext^m$, $\mat A \in \Rext^{m \times n}$ encode linear constraints.
\textcite*{nesterov1996infeasible} first established that such problems can be solved in polynomial time, but incur double the memory and eight times the time, compared to the symmetric counterparts. These drawbacks were eliminated in \cite{skajaa2015homogeneous}.
The algorithm that seems to display the best empirical performance \parencite{dahl2022primal}, however, was only recently shown to
run in polynomial time \parencite{badenbroek2021algorithm}.

Disciplined Convex Programming \parencite{dcp-thesis} is a
compositional approach to convex optimization that 
imposes certain restrictions on how problems can be specified.
A problem conforming to those rules is said to be \emph{dcp},
and can be efficiently compiled to a standard form
\parencite{agrawal2018rewriting},
which in our case is an exponential conic program.
Only two rules are relevant to us: a constraint of the form
$(x,y,z) \in K_{\exp}$ is
    dcp iff $x$, $y$, and $z$ are affine transformations of the
    optimization variables, 
and a linear program
augmented
with dcp 
constraints is dcp.
Because all the optimization problems that we give are
of this form,
we can easily compile them
to exponential conic programs even if they do not exactly conform to \eqref{eq:exp-conic-program}.

\discard{\color{gray!80!white}
    \section{AN EXPRESSIVE CLASS OF OPTIMIZATION PROBLMS}

    \begin{itemize}
        \item
        Many optimization problems can be effectively solved with
    \end{itemize}

    \begin{itemize}
        \item
        For optimization people: a new class of optimization problems,
    \end{itemize}

}

\section{Inference as a Convex Program}
    \label{sec:inf-as-cvx-program}

Here is an obvious, if inefficient,
way
of calculating
$\Pr_{\!\cal M}(Y|X{=}x)$ in a
probabilistic model $\cal M$. 
First compute an explicit representation of the joint distribution 
$\Pr_{\!\cal M} \in \Delta\mskip-2mu\V\!\X$, 
then marginalize to 
$\Pr_{\!\cal M}(X,Y)$ and condition on $X{=}x$.
For a factor graph or BN,
each step is
straightforward;
the problem is the exponential time and space required to represent $\Pr_{\!\cal M}(\X)$ explicitly.
A key feature of inference algorithms for BNs and FGs is that they
do not represent joint distributions in this way.
For PDGs, though, it is not
obvious that
we can calculate the $\zogamma$-semantics,
even if
we know it is unique, and
we ignore the space required to represent it (as we do in this section).
Note that $\zogamma$-inference is already an optimization problem by definition:
\[
    \minimize_\mu\quad
        \bbr{\dg M}_\gamma(\mu)
    \quad \subjto\quad \mu \in \Delta\mskip-2mu\V\mskip-2mu\X.
\]
For small enough 
 $\gamma$,
it is even convex.
But can we solve it efficiently?
With exponential cone constraints,
the answer is yes, as we show in \Cref{sec:small-gamma}.
Moreover, we can compute the $0^+$\!-semantics with a sequence of two exponential conic programs (\Cref{sec:empirical-limit}).
To give a flavor of our constructions and ease into
the more complicated ones, we begin by 
minimizing
 $\OInc$, the simpler of the two scoring functions.

\subsection{%
    Minimizing Incompatibilty
    (\texorpdfstring{$\boldsymbol\gamma\boldsymbol=\mat0$}{gamma=0})%
} \label{sec:minimize-inc}

When $\gamma = 0$, we want to find minimizers of $\OInc$,
which is a
weighted sum of conditional relative entropies.
There is a straightforward connection between the exponential cone and
relative entropy:
if $\mat m, \mat p \in
\Delta \{1, \ldots, n\}
\subset \mathbb R^n$ are points on
a probability simplex,
then $(-\mat u, \mat m, \mat p) \in K_{\exp}^n$ if and only if
$\mat u$ is an upper bound on $\mat m \log \frac{\mat m}{\mat p}$,
the pointwise contribution to relative entropy at each outcome.
Thus, perhaps unsurprisingly, we can use an exponential conic program to
find minimizers of $\OInc$.
If all beliefs are unconditional and over the same space,
the construction is standard;
we review it here, so that we can build upon it.

\textbf{Warm-up.}
\begingroup
Consider a PDG with
only one variable $X$
with
$\V\mskip-2mu X = \{1, \ldots, n\}$.
\discard{\vfull{
    Suppose further that for every arc $j \in \Ar = \{1, \ldots, k\}$, the cpd $\p_j(X)$ is an unconditional distribution over $X$.
    That is, $\Tgt j = \{ X \}$, and $\Src j = \emptyset$.
    Such unconditional probabilities may be identified with vectors $\mat p_j \in [0,1]^n$, and all $k$ of them may conjoined to form a
    matrix $\mat P = [\,p_{ij}] \in [0,1]^{n \times k}$.
    A candidate
    (joint)
    distribution $\mu(X)$
    may be represented as a vector $\mat m \in [0,1]^n$.}}
Suppose further that every arc $j \in \Ar = \{1, \ldots, k\}$
has $\Tgt j = \{ X \}$ and $\Src j = \emptyset$.
Then each $\p_j(X)$ can be identified with a vector $\mat p_j \in [0,1]^n$, and all $k$ of them can conjoined to form a matrix $\mat P = [\,p_{ij}] \in [0,1]^{n \times k}$.
Similarly, a candidate distribution $\mu$ can be identified with $\mat m \in [0,1]^n$. 
Now consider a matrix $\mat U = [u_{i,j}] \in \Rext^{n \times k}$
that, intuitively, gives an upper bound on
the contribution to $\OInc$ due to each edge and value of $X$.
Observe that
\begin{align*}
    &&(- \mat U,~
        [\mat m,\,...\,, \mat m]
        ,~ \mat P) &\in K_{\exp}^{n \times k} \\
    &\iff& \forall  i,j.~~
            u_{ij} &\ge m_i \log (\nicefrac{m_i}{p_{ij}}) \\
    &\implies& \forall j.~~
        {\textstyle\sum_i} u_{ij}  &\ge \kldiv{\mu}{p_j} \\
    &\implies& {\textstyle\sum_{i,j}} \beta_j u_{ij}  &\ge {\textstyle\sum_j} \beta_j \kldiv{\mu}{p_j} \\
    &\iff& \mat 1^{\sf T} \mat U \bbeta &\ge \OInc(\mu) .
        \numberthis\label{eqn:warmup-logic}
\end{align*}

So now, if $(\mat U, \mat m)$ is a solution to the convex program
\begin{align*}
    \minimize_{\mat m, \mat U}~~
        \mat 1^{\sf T} \mat U \bbeta
    \quad\subjto\quad &
        \mat 1^{\sf T} \mat m  = 1, \\[-2ex]
        (-\mat U,\;&
            [\mat m,\,...\,, \mat m]
            ,\; \mat P
        )
            \in K_{\exp}^{n \times k},
\end{align*}
then (a) 
the
objective value $\mat 1^{\sf T} \mat U \bbeta$
equals
the inconsistency $\aar{\dg M}_0$, and (b) $\mu \in \bbr{\dg M}^*_0$,
meaning $\mu$ minimizes $\OInc_{\dg M}$.

\endgroup

\textbf{The General Case.}
We now show how the same construction can be used to find
 a distribution $\mu \in \bbr{\dg M}^*_0$
for an arbitrary PDG $\dg M = (\X, \Ar, \mathcal P, \balpha, \bbeta)$.
\discard{\color{gray}Now that we have a taste for how this works in terms of matrices,
    let's now move up a level,
    and identify distributions with their simplex representations.}
To further simplify the presentation,
for each arc $a \in \Ar$, let
$\V a := \V(\Src a, \Tgt a)$
denote all joint settings of $a$'s source and target variables, and
write
$
\V\!\Ar :%
    = \sqcup_{a \in A} \V a
    = \{ (a, s, t) : a \in \Ar,\, (s,t) \in \V(\Src a,\Tgt a) \}
$
for the set of all choices of an arc together with values of its source and target.
For each $a \in \Ar$, 
we can regard $\mu(\Tgt a, \Src a)$ and $\mu(\Src a)\p_a(\Tgt a | \Src a)$, both distributions over $\{\Src a,\Tgt a\}$, 
as vectors of shape $\V a$.
As before, we introduce an optimization variable $\mat u$ that packages together
    all of the relevant pointwise upper bounds.
To that end, consider a 
vector
$\mat u = [u_{a,s,t}] \in \Rext^{\V\!\Ar}$
in the optimization problem
\discard{\begin{align*}
    \minimize_{\mu, \mat u} \quad
        \sum_{\mathrlap{(a,s,t) \in \V\!\Ar}} \beta_a \, u_{a,s,t}
        \quad~~&
    \numberthis\label{prob:joint-inc}\\[-3ex]
    &\subjto\quad \mu \in \Delta\V\!\X, \\[-0.3ex]
    \forall a \in \Ar.~\big({-}{\mat u}_a,\, \mu( \Tgt a,\Src a),&\, \p_a(\Tgt a | \Src a)  \mu(\Src a) \big) \in K_{\exp}^{\V a}
\end{align*}}%
{\begin{align*}
    \minimize_{\mu, \mat u} &\quad
        \sum_{\mathrlap{(a,s,t) \in \V\!\Ar}} \beta_a \, u_{a,s,t}
    \numberthis\label{prob:joint-inc}\\
    \subjto&\quad \mu \in \Delta\V\!\X, \\[-0.4ex]
        \forall a \in \Ar.~&\big(-{\mat u}_a,\, \mu( \Tgt a,\Src a),\, \p_a(\Tgt a | \Src a)  \mu(\Src a) \big) \in K_{\exp}^{\V a}
        .
\end{align*}}%
where $\mat u_a = [u_{a,s,t}]_{(s,t) \in \V a}$ consists of those
components of $\mat u$ associated with arc $a$.
Note that
the marginals
 $\mu(\Src a, \Tgt a)$ and $\mu(\Src a)$ 
are affine transformations of $\mu$, so \eqref{prob:joint-inc} is dcp.
A straightforward generalization of the logic in \eqref{eqn:warmup-logic} gives us:

\begin{linked}{prop}{joint-inc-correct}
    If $(\mu, \mat u)$ is a solution to \eqref{prob:joint-inc}, then
    $\mu \in \bbr{\dg M}_0^*$,
    and
    $%
        \sum_{(a,s,t) \in \V\!\Ar} \beta_a u_{a,s,t} = \aar{\dg M}_0$.
\end{linked}

Thus, a solution to \eqref{prob:joint-inc}
encodes a distribution that minimizes $\OInc$, and
the (0-)inconsistency
of $\dg M$.
This is a start, but to do 
$0^+$\!-inference,
among the minimizers of $\OInc$
we must find the unique distribution in $\bbr{\dg M}^*_{0^+}$, 
while for $\zogamma$-inference ($\zogamma > 0$), we need to find the optimizers of
$\bbr{\dg M}^*_\gamma$.
Either way, we must consider 
$\SInc$
in addition to $\OInc$. 

\subsection{%
    \texorpdfstring{$\boldsymbol\gamma$}%
    {gamma}-Inference
    for small
    \texorpdfstring{$%
    \boldsymbol\gamma%
    \boldsymbol>\mat0$}{gamma}%
} \label{sec:small-gamma}

When $\gamma > 0$ is small enough,
the scoring function \eqref{eqn:scoring-fn} is not only convex,
but admits a straightforward representation as an exponential conic program.
To see this, note that \eqref{eqn:scoring-fn} can be rewritten \parencite[Prop 4.6]{pdg-aaai} as:
\begin{equation}
    \begin{aligned}
        \bbr{\dg M}_\gamma(\mu) = &-\gamma\H(\mu) -
            \sum_{a \in \Ar}
                \beta_a\, \Ex_\mu
                    \log {\p_a(\Tgt a | \Src a)}
                \\[-0.6ex]
            &~~+ \sum_{a \in \Ar}
            (\gamma \alpha_a - \beta_a)
                \H_\mu (\Tgt a | \Src a).
    \end{aligned}
    \label{eq:altscore}
\end{equation}
The first term,
$-\gamma\H(\mu)$,
is strictly convex and has a well-known
translation into an exponential cone constraint;
the second one linear in $\mu$.
If $0 < \gamma \le \min_{a} \frac{\beta_a}{\alpha_a}$, then
every summand of the last term is a negative conditional entropy, and 
can be captured by an exponential cone constraint.
The only wrinkle is that it is possible for a user to specify that some $\p_a(t\mid s) = 0$, in which case the linear term 
is undefined.
The result is a requirement that $\mu(s,t) = 0$ at such points,
which we can instead encode directly with linear constraints.
To do this formally,
divide $\V\!\Ar$ into two parts:
$\V\!\Ar^+ := \{ (a,s,t) \in\V\!\Ar : \p_a(t |s) > 0\}$ and
$\V\!\Ar^0 := \{ (a,s,t) \in\V\!\Ar : \p_a(t |s) = 0\}$.
Armed with this notation, consider upper bound vectors
$\mat u = [ u_{a,s,t}]_{(a,s,t) \in \V\!\Ar}$ and $\mat v = [v_w]_{w \in \V\!\X}$,
in the following optimization problem:
{%
\begin{align*}
\minimize_{\mu, \mat u, \mat v} & ~~
    \sum_{\mathrlap{\!\!\!(a,s,t) \in \V\!\Ar}}
        (\beta_a \!- \alpha_a \gamma) u_{a,s,t}
        \,+
        \gamma
        \sum_{\mathclap{w \in \V\!\X}} v_w
    \numberthis\label{prob:joint-small-gamma}
    \\[-0.2ex]
    &\qquad
    - \sum_{\mathrlap{\!\!\!(a,s,t) \in \smash{\V\!\Ar^+}}} 
        \alpha_a \gamma \, 
        \mu(\Src a{=}s,\Tgt a {=} t) \log \p_a (t|s)
\\[0.2ex]
\subjto&\quad \mu \in \Delta\V\!\X, 
        \quad ( -\mat v,  \mu,  \mat 1) \in K_{\exp}^{\V\!\X},
    \\[-0.4ex]
    \forall a \in \Ar.~
        &\big(-\mat u_a, \mu( \Tgt a,\Src a),\p_a(\Tgt a | \Src a)  \mu(\Src a) \big)
            \in K_{\exp}^{\V a}, \\[-0.1ex]
    \forall (a,s,t) &\in \V\!\Ar^0\!.~
    \mu(\Src a{=}\mskip2mus, \Tgt a{=}\mskip2mut) = 0.
\end{align*}}

This optimization problem may look complex, but it 
falls out of
\eqref{eq:altscore} 
fairly
directly,
and gives us what we wanted.

\begin{linked}{prop}{joint-small-gamma-correct}
    If $(\mu, \mat u, \mat v)$ is a solution to \eqref{prob:joint-small-gamma},
    and $\bbeta \ge \gamma \balpha$,
    then
    $\mu$ is the unique element of
    $\bbr{\dg M}^*_\gamma$, and $\aar{\dg M}_\gamma$
    equals the objective of \eqref{prob:joint-small-gamma} evaluated at $(\mu, \mat u, \mat v)$.
\end{linked}

\subsection{
    Calculating the \texorpdfstring{$\mat 0^{\boldsymbol+}$\!}{0+}-semantics
    (\texorpdfstring{$\boldsymbol\gamma\boldsymbol\to\mat 0$}{gamma->0})}
    \label{sec:empirical-limit}
\cref{sec:minimize-inc} shows how to find a distribution $\nu$ that minimizes
$\OInc$%
---but to do
$0^+$\!-inference,
we need to find the minimizer
that, uniquely among them, best minimizes
$\SInc$.
It turns out this can be done by
    using $\nu$ to construct a second optimization problem.
The justification requires two more results;
we start by characterizing the minimizers of $\OInc$.

\begin{linked}{prop}{marginonly}
    If $\dg M$ has arcs $\Ar$ and $\bbeta \ge 0$,
    the minimizers of $\OInc_{\dg M}$ all have the same conditional
        marginals along $\Ar$.
    That is, for all $\mu_1, \mu_2 \in \bbr{\dg M}_0^*$
    and all $\ed aST \in \Ar$ 
    with $\beta_a > 0$, we have
    {\subafalse
    $\mu_1(\Tgt a, \Src a)\mu_2(\Src a) = \mu_2(\Tgt a, \Src a) \mu_1(\Src a)$.%
    \onlyfirsttime{\footnotemark}
    }
\end{linked}
\footnotetext{Intuitively, this assserts 
$\mu_1(\Tgt a | \Src a) = \mu_2(\Tgt a | \Src a)$,
but also handles cases where some
$\mu_1(\Src a {=} s)$ or $\mu_2(\Src a {=} s)$ 
equals zero.}

As a result, once we find one minimizer $\nu$ of $\OInc_{\dg M}$
(e.g., via \eqref{prob:joint-inc}),
it suffices to optimize $\SInc$ among distributions that have the same
conditional marginals along $\Ar$ that $\nu$ does.
This presents another problem: $\SInc$
is typically not convex.
Fortunately, if we constrain to distributions that minimize $\OInc$, then it is.
Moreover, on this restricted domain, it can be represented 
with dcp exponential cone constraints.

\begin{linked}{prop}{idef-frozen}
If $\mu \in \bbr{\dg M}_0^*$\,,
then
\begin{equation}
    \underset{{\dg M}}{\SInc_{}}(\mskip-0.5mu\mu\mskip-1mu) \!=\!
        \sum_{\mathclap{ w \in \V\!\X } }
            \mskip-1mu
            \mu(\mskip-1.5mu w \mskip-1.5mu)
            \log \!  \bigg(\!
                \faktor{\mu(\mskip-1.5mu w\mskip-1.5mu )}{\,\prod_{\mathclap{a \in \Ar}} 
                \mskip-1mu
                \nu\big(\mskip-1mu\Tgt a \mskip-2.2mu(\mskip-2mu w \mskip-2mu) 
                    \mskip-0.5mu \big|  \Src a \mskip-2mu(\mskip-2mu w \mskip-2mu)\mskip-2mu\big)^
                {\!\alpha\ssub a}
                }\!\mskip-2mu
            \bigg)\mskip-2mu
        ,\!
        \label{eq:idef-alt-constr}
\end{equation}
where $\{ \nu(\Tgt a | \Src a ) \}_{a \in \Ar}$ are the
marginals along the arcs $\Ar$
shared by all distributions in $\bbr{\dg M}^*_0$\
(per \cref{prop:marginonly}),
and $\Src a \mskip-1mu(\mskip-1mu w \mskip-1mu), \Tgt a \mskip-1mu(\mskip-1mu w \mskip-1mu)$ are the values of variables $\Src a$ and $\Tgt a$ in $w$.
\end{linked}

If we already know a distribution $\nu \in \bbr{\dg M}_0^*$,
perhaps by solving \eqref{prob:joint-inc}, then
the denominator of \eqref{eq:idef-alt-constr} does not depend on $\mu$ 
and so is constant in our search for minimizers of
$\SInc_{}$.
For ease of exposition, aggregate these values into a vector
\begin{equation}
    \mat k :=
        \Big[
        ~\prod_{\smash{a \in \Ar}} \nu(\Tgt a (w) | \Src a (w))^{\alpha\ssub a}
        \Big]%
        _{w \in \V\!\X}~\raisebox{-1ex}.
        \label{eq:cm-product}
\end{equation}
We can now capture $\bbr{\dg M}^*_{0^+}$ with a convex program.

\begin{linked}{prop}{joint+idef-correct}
If $\nu \in \bbr{\dg M}_0^*$
and $(\mu, \mat u)$ 
solves the problem
\begin{align*}
    \minimize_{\mu, \mat u} & \quad
        \smash{\mat 1^{\sf T} \mat u}
        \numberthis\label{prob:joint+idef}\\[-0.5ex]
    \subjto &\quad
        (-\mat u,  \mu, \mat k ) \in K_{\exp}^{\V\!\X},~~\quad \mu \in \Delta\V\!\X, \\[-0.4ex]
            \forall& \ed aST \subafalse \in \Ar.~~\mu(\Src a, \Tgt a)\, \nu(\Src a) = \mu(\Src a)\, \nu(\Src a, \Tgt a),
\end{align*}
then $\bbr{\dg M}^*_{0^+} = \{ \mu \}$
and $\mat 1^{\sf T} \mat u = \SInc_{\dg M}(\mu)$.
\end{linked}
Running \eqref{prob:joint+idef} through a convex solver gives rise to the 
first algorithm
that can reliably find $\bbr{\dg M}^*_{0^+}$.

\section{Polynomial-Time Inference Under Bounded Treewidth}
    \label{sec:clique-tree-expcone}
We have now seen how $\zogamma$-inference
(for small $\zogamma$) can be reduced to convex optimization
over joint distributions $\mu$---%
but $\mu$ grows exponentially with the number of variables in the PDG,
so we do not yet have a tractable inference algorithm.
We now show how $\mu$ can be replaced with a \actree\ over the PDG's structure. 
What makes this possible is a key independence property of traditional graphical models,
which we now prove
holds for PDGs as well.

\begin{linked}[Markov Property for PDGs]{theorem}{markov-property}
  If\, $\dg M_1$ and $\dg M_2$ are PDGs
    over sets $\X_1$ and $\X_2$ of variables, respectively,
\discard{
    Then for all $\gamma > 0$, we have that
    \[  \bbr{\dg M_1 \bundle \dg M_2}^*_\gamma
			~\models~
		\X_1 \mathbin{\bot\!\!\!\bot} \X_2 \mid \X_1 \cap \X_2. \] 
    That is: for every distribution $\mu \in \bbr{\dg M_1 \bundle \dg M_2}^*_\gamma$,
    the variables of $\dg M_1$ and of $\dg M_2$ are conditionally independent given the variables they have in common.
}
then
$\X_1$
and $\X_2$
are conditionally independent given $\X_1 \cap \X_2$
in every
 $\mu \in \bbr{\dg M_1 \bundle \dg M_2}^*_\gamma$\,,
for all $\gamma > 0$ and
$\gamma=0^+$.
\end{linked}

For the remainder of this section, fix a PDG $\dg M$ and a tree decomposition $(\C, \mathcal T)$ of $\dg M$'s hypergraph.
One significant consequence of \cref{theorem:markov-property} is that, in the
search for optimizers of \eqref{eqn:scoring-fn}, we
need consider only distributions that satisfy those independencies,
all of which can be represented as a \actree\ 
$\bmu = \{\mu_C \in \Delta\V(C) \}_{C \in \C}$
over $(\C, \mathcal T)$.

\begin{linked}{coro}{can-use-cliquetree}
    If $\dg M$ is a PDG with arcs $\Ar$, 
    $(\C, \mathcal T)$ is a tree decomposition of $\Ar$,
    $\gamma > 0$, and
    $\mu \in \bbr{\dg M}^*_\gamma$, then there exists a \actree\ 
    $\bmu$ over $(\C, \mathcal T)$ such that $\Pr_{\bmu} = \mu$.
\end{linked}

For convenience, let
$\V\C := \{(C,c) : C \in \C, c \in \V(C)\}$ 
be
the set of all choices of a cluster together with a setting of its variables. 
Like before, we start by optimizing $\OInc$, this time
over \cactree s $\bmu$,
which we identify with vectors
$
 \bmu
    \cong [\mu_C(C{=}c)]_{(C,c)\in\V\C} 
$.
We need the conditional marginals $\Pr_{\bmu}(\Tgt a | \Src a)$ of $\bmu$ along every arc $a$ in order
to calculate $\OInc_{\dg M}(\Pr_{\bmu})$; fortunately, they are readily available.
Since $(\C, \mathcal T)$ is a tree decomposition,
we know $\Src a$ and $\Tgt a$ lie entirely within some cluster $C_{\!a} \in \C$,
and $\Pr_{\!\bmu}(\Tgt a | \Src a) = \mu_{C_{\!a}}\!(\Tgt a | \Src a)$ if $\bmu$ is calibrated.
For $\mat u \in \Rext^{\V\!\Ar}$, consider the problem
\begin{align*}
    \minimize_{\bmu, \mat u} &\quad
        \sum_{\mathrlap{(a,s,t) \in \V\!\Ar}}\beta_a \,  u_{a,s,t}
    \numberthis\label{prob:cluster-inc}\\
    \subjto&\quad
        \forall C \in \C.~\mu_C \in \Delta\V(C), \\[-0.3ex]
        \forall a \in \Ar.~
            \big(&\!- \! \mat u_a,\, \mu_{C\!_a}\!(\Src a,\mskip-2mu \Tgt a),\, \mu_{C\!_a}\!(\Src a) \p_a(\Tgt a | \Src a)\big) \in K_{\exp}^{\V a} \\[-0.2ex]
        \forall (C,D) &\in \mathcal T.~~ \mu_{C}(C \cap D) = \mu_{D}(C \cap D),
\end{align*}
where again $\mat u_a$ is the restriction of $\mat u$ to components associated with $a$.
Problem \eqref{prob:cluster-inc} is similar to \eqref{prob:joint-inc}, except
that it requires local marginal constraints to restrict our search to \cactree s.
It is analogous to problem 
\textsc{CTree-Optimize-KL}
of \textcite[pg. 384]{koller2009probabilistic}.

\begin{linked}{prop}{cluster-inc-correct}
    If $(\bmu, \mat u)$ is a solution to \eqref{prob:cluster-inc}, then
    \begin{enumerate}[label={(\alph*)},nosep]
    \item $\bmu$ is a calibrated, with $\Pr_{\bmu} \in \bbr{\dg M}^*_0$, and
    \item the objective of \eqref{prob:cluster-inc} evaluated at $\mat u$ equals $\aar{\dg M}_0$.
    \end{enumerate}
\end{linked}
We can now find a minimizer of $\OInc$ and
compute $\aar{\dg M}_0$ without storing a joint distribution.
\discard{
Note that
\eqref{prob:cluster-inc} is the result of modifying \eqref{prob:joint-inc} in the obvious way to deal with \actree s.
}
But to do anything else, we must deviate from the template laid out in \cref{sec:inf-as-cvx-program}.

\textbf{Dealing with Joint Entropy.}
In the construction of \eqref{prob:cluster-inc},
we rely heavily on the fact that each term of $\OInc_{\dg M}$
depends only on local marginal distributions $\mu_{C_{\mskip-2mu a}}\!(\Tgt a,  \Src a)$
and $\mu_{C_{\mskip-2mu a}}\!(\Src a)$.
The same is not true of $\SInc_{}$, which depends on the joint entropy $\H(\Pr_{\bmu})$ of the entire distribution.
At this point we should point out an important 
reason to restrict our focus to trees:
it allows the joint entropy to be expressed
in terms of the cluster marginals \parencite{wainwright2008graphical},
by
\begin{equation}\label{eq:bethe-entropy}
    -\H(\Pr\nolimits_{\bmu})
        = -\sum_{C \in \C} \H(\mu_C)
        ~+~ \sum_{\mathclap{(C,D) \in \mathcal T}} \H_{\bmu}(C \cap D).
\end{equation}
Even so,
it is not obvious that
\eqref{eq:bethe-entropy} can be
captured with dcp exponential cone constraints.
(Exponential conic programs can minimize negative entropy,
but not positive entropy, which is concave.)
We now describe how this can be done.

\def\Par#1{\mathrm{Par}(#1)}
\def\Pash{\mathit{V\mskip-5muC\mskip-3.5muP\!}}

Choose a root node $C_0$ of the tree decomposition, and orient each edge of $\mathcal T$ so that it points away from $C_0$.
Each cluster $C \in \cal C$, except for $C_0$, then has a parent cluster $\Par C$;
define $\Par{C_0} := \emptyset$ to be an empty cluster, since $C_0$ has no parent.
Finally, for each $C \in \C$, let $\Pash_C := C \mathbin{\cap} \Par C$ denote the
the set of $\mathbf v$ariables that cluster $C$ has in $\mathbf c$ommon with its $\mathbf p$arent cluster.
\unskip\footnotemark 
As $\cal T$ is now a directed tree, this definition allow us to express
\eqref{eq:bethe-entropy} in a more useful form:
\begin{align*}
    - \H(&\Pr\nolimits_{\bmu}) =
        - \H(\mu_{C_0}) - \!
        \sum_{(C \to D)\mathrlap{ \in \mathcal T}}
        \H_{\Pr_{\bmu}}(D \mid C)\\[-1ex]
    &= 
        \sum_{C \in \cal C} \sum_{c \in \V(C)}
        \mu_C(C{=}c)
        \log \frac
            { \mu_C(C{=}c)}
            { \mu_C(\Pash_C(c)) }
        ,
            \numberthis\label{eq:cluster-ent-decomp}
\end{align*}
where $\Pash_C(c)$ is the restriction of the 
joint value $c \in \V(C)$
to the variables $\Pash_C \subseteq C$
\unskip.
Crucially, the denominator of \eqref{eq:cluster-ent-decomp} is an affine transformation of $\mu_C$.
The upshot: we have rewritten the joint entropy
as a sum of functions of the clusters, each of which can be captured with a dcp exponential cone constraint.
This gives us analogues of the problems
in \cref{sec:small-gamma,%
sec:empirical-limit} that 
operate on \actree s.

\textbf{Finding \actree s for $\zogamma$-inference.}
The ability to decompose the joint entropy as in \eqref{eq:cluster-ent-decomp} allows us to adapt 
\eqref{prob:joint-small-gamma} 
to operate on \cactree s, rather than joint distributions. 
Beyond the changes already present in \eqref{prob:cluster-inc},
the key
is to replace 
the exponential cone constraint
$( -\mat v,  \mu,  \mat 1) \in K_{\exp}^{\V\!\X}$,
which captures the entropy
of $\mu$,
with
\[
\big(-\mat v,\;\bmu,\, [\,\mu_{C}(\Pash_C(c))\,]_{(C,c)\in\V\C}\big) \in K_{\exp}^{\V\C},
\]
which captures the entropy of $\bmu$, by
\eqref{eq:cluster-ent-decomp}.
\ifvfull %
Over vectors
$\mat v, \bmu \in \Rext^{\V\C}$ and
$\mat u \in \Rext^{\V\!\Ar}$,
the problem becomes: 
{\allowdisplaybreaks
\begin{align*}
    \minimize_{\bmu, \mat u, \mat v} & ~~
    \sum_{\mathrlap{\!\!\!(a,s,t) \in \V\!\Ar}} (\beta_a \!- \alpha_a \gamma) u_{a,s,t}
    + \gamma \sum_{\mathclap{(C,c) \in \V\C}}  v_{C,c}
    \numberthis\label{prob:cluster-small-gamma}
    \\[-0.2ex]
    - \sum_{\mathrlap{\!\!\!(a,s,t) \in \smash{\V\!\Ar^+}}}&
        \alpha_a\gamma\,
        \mu_{C\!_a}\!(\Src a{=}s,\Tgt a {=} t)
        \log \p_a (\Tgt a{=}t\mid s)
\\[0.2ex]
\subjto&\quad
    \forall C \in \C.~\mu_C \in \Delta\V(C), \\[-0.2ex]
    \forall a \in \Ar.~
        \big(&\!- \! \mat u_a,\, \mu_{C\!_a}\!(\Src a,\mskip-2mu \Tgt a),\, \mu_{C\!_a}\!(\Src a) \p_a(\Tgt a | \Src a)\big) \in K_{\exp}^{\V a}, \\
    \forall (a,s,t) &\in \V\!\Ar^0\!.~
    \mu_{C\!_a}\!(\Src a{=}\mskip2mus, \Tgt a{=}\mskip2mut) = 0, \\[-0.2ex]
    \forall (C,D) &\in \mathcal T.~~ \mu_{C}(C \cap D) = \mu_{D}(C \cap D),\\[-0.3ex]
    \big(-\mat v,\;&\bmu,\, [\,\mu_{C}(\Pash_C(c))\,]_{(C,c)\in\V\C}\big) \in K_{\exp}^{\V\C}
    .\\[-5ex]
\end{align*}%
}%
\else%
This gives rise to an optimization problem 
over
$\mat v, \bmu \in \Rext^{\V\C}$ and
$\mat u \in \Rext^{\V\!\Ar}$,
that we call \eqref{prob:cluster-small-gamma}.
The rest of the details are less instructive, so we defer
them to \cref{appendix:prob-details} for brevity. 
\fi%

\footnotetext{
    Different choices of $C_0$ 
    yield
    different definitions of $\Pash$, 
    and ultimately optimization problems of different sizes;
    the optimal choice
    can be found with Edmund's Algorithm \parencite{chu1965shortest},
    which computes a directed analogue of the minimum spanning tree.}

\begin{linked}{prop}{cluster-small-gamma-correct}
    If $(\bmu, \mat u, \mat v)$ is a solution to \eqref{prob:cluster-small-gamma}
    and $\bbeta \ge \gamma \balpha$, then
    $\Pr_{\bmu}$ is the unique element of $\bbr{\dg M}^*_\gamma$,
    and the objective of \eqref{prob:cluster-small-gamma} at $(\bmu, \mat u, \mat v)$ equals $\aar{\dg M}_\gamma$.
\end{linked}%
\discard{%
\textbf{A \actree\ for
    $\mat 0^{\boldsymbol+}$\!-inference.%
}}%
A related use of \eqref{eq:cluster-ent-decomp} is
to enable an analogue of
\ifvfull%
\eqref{prob:joint+idef} that searches over \actree s (rather than joint distributions),
to find
a compact representation of $\bbr{\dg M}^*_{0^+}$.
We begin with a straightforward adaptation of
    the relevant machinery in \Cref{sec:empirical-limit}.
Suppose that $\boldsymbol\nu {=} \{\nu_C : C \in \C\}$ is a \cactree\ over the tree decomposition $(\C, \mathcal T)$ representing a distribution $\Pr_{\boldsymbol\nu} \in \bbr{\dg M}^*_0$, say obtained by solving \eqref{prob:cluster-inc}.
For $C \in \C$, let $\Ar_C:= \{ a \in \Ar : C_a = C\}$ be the set of
arcs assigned to cluster $C$, and let
\[
    \mat k := \smash{\bigg[} \prod_{\smash{\mathrlap{a \in \Ar_C}}} \nu_C (\Tgt a (c) | \Src a (c))^{\alpha_a} \smash{\bigg]\mathclose{\vphantom{\Big|}}_{(C,c) \in \V \C}} ~~\in \Rext^{\V\C}
\]
be the analogue of \eqref{eq:cm-product} for a cluster tree.
Once again, consider
$\mat u := [ u_{(C,c)} ]_{(C,c) \in \V\C}$
in the optimization problem
{\allowdisplaybreaks%
\begin{align*}
\minimize_{\bmu, \mat u} & \quad
    \mat 1^{\sf T} \mat u
    \numberthis\label{prob:cluster+idef}\\
\subjto &\quad
    \forall C \in \C.~\mu_C \in \Delta\V(C), \\[-0.2ex]
     \big({-}\mat u,\,  \bmu,\,\, &
            \mat k \odot
            \big[\;\mu_C(\Pash_C(c))\;\big]_{(C,c) \in \V\C}
            \big) \in K_{\exp}^{\V\C}, \\[-0.2ex]
    \forall a \in \Ar.&~~\mu_{C_{\!a}}\!(\Src a, \Tgt a) \nu_{C_{\!a}}\!(\Src a) = \mu_{C_{\!a}}\!(\Src a) \nu_{C_{\!a}}\!(\Src a, \Tgt a)\\
    \forall (C,D) &\in \mathcal T.~~ \mu_{C}(C \cap D) = \mu_{D}(C \cap D).
\end{align*}}%
The biggest change is in the second constraint: 
the upper bounds $[u_{(C,c)}]_{c \in \V C}$ for cluster $C$ now account only
for the additional entropy not already modeled by 
$C$'s
ancestors.
\else %
\eqref{prob:joint+idef} for clique s,
resulting in an optimization problem that we call \eqref{prob:cluster+idef},
whose details we defer to \cref{appendix:prob-details} as well.
\fi

\begin{linked}{prop}{cluster-idef-correct}
    If $(\bmu, \mat u)$ is a solution to \eqref{prob:cluster+idef},
    then $\bmu$ is a \cactree\
    and $\bbr{\dg M}^*_{0^+} = \{ \Pr_{\!\bmu} \}$.
\end{linked}

At this point, standard algorithms can use $\bmu$
to answer probabilistic queries about $\Pr_{\!\bmu}$ in polynomial time \parencite[\S 10.3.3]{koller2009probabilistic}.
\discard{
    Concretely: marginal probabilities can essentially be read off of a cailbrated a \actree,
    and evidence $X{=}x$ may be incorporated by
    setting $\mu_C(c) := 0$ for every $C{=}c$ that conflicts with $X{=}x$
    and recalibrating the \actree\ (e.g., with belief propagation). }
From \cref{prop:cluster-idef-correct,prop:cluster-small-gamma-correct}, it follows that
$\zogamma$-inference
(for small $\zogamma$, and for $0^+$)
can be reduced to a (pair of) convex optimization problem(s) with
a polynomial number of variables and constraints.
All that remains 
is to show that such a problem can be solved in polynomial time.
For this, we turn to interior-point methods.
As
\eqref{prob:cluster-small-gamma} and \eqref{prob:cluster+idef} are dcp, they
can be transformed via established methods \parencite{agrawal2018rewriting} into
a standard form
that can be solved in polynomial time by commercial solvers \parencite{mosek,ECOS}.
Threading the details of our constructions through
the analyses of \textcite{dahl2022primal}
and \textcite{nesterov1996infeasible}
results in
our main theorem.

\discard{
\begin{linked}{theorem}{main}
We can do PDG inference to precision $\epsilon$ in 
    \[ 
    O\pqty[\Big]{  (N\!+\!A)^4 V^{4T}
        \pqty[\Big]{ T \log V + \log \frac{N\!+\! A}{\epsilon} 
    + \log
    \frac{\beta^{\max}}{\beta^{\min}}
        } }%
    \]
    time,
    \unskip\daggerfootnote{At the cost of substantial overhead and engineering effort, the exponent $4$ can be reduced to 2.872, by appeal to \textcite{skajaa2015homogeneous} and the 
    current best matrix multiplication algorithm \parencite[$O(n^{2.372})$]{duan2022faster}
    to invert 
    $n{\times} n$
    linear systems. }
    where:
    \begin{itemize}[nosep,%
            ]
        \item $N$ is the total number of variables,
        \item $V$ is the number of values per variable,
        \item $T$ is the tree-width of the PDG's structure, 
        \item $A$ is the number of hyperarcs,
        \item $\beta^{\max}$ is the maximum observational confidence,
        and
        \item $\beta^{\min}$ is either $\gamma$ (for $\zogamma$-inference)
         or the minimimum non-zero observational confidence (for $0^+$\!-inference).
        \item $\gamma^{\max} = \max_{a \in \Ar} \beta_a/\alpha_a$
    \end{itemize}
\end{linked}
}

\begin{linked}{theorem}{main}
Let $\dg M = (\X, \Ar, \mathbb P, \balpha, \bbeta)$
be a proper discrete PDG with $N = |\X|$ variables each taking at most $V$ values
and $A = |\Ar|$ arcs,
in which each component of 
$\bbeta \in 
    \mathbb R^{\Ar}$
and $\mathbb P \in
    \mathbb R^{\V\!\Ar}$
is specified in binary with
at most
$k$ bits.
Suppose that $\gamma \in \{0^+\}\cup (0,\,  \min_{a \in \Ar} \frac{ \beta_a}{\alpha_a}]$.
If $(\C, \mathcal T)$ is a tree decomposition of $(\X,\Ar)$ of width $T$
and $\bmu^* \in \mathbb R^{\V\C}$ 
is the unique \cactree\ over $(\C, \mathcal T)$ 
that represents the $\zogamma$-semantics of $\dg M$,
then
\begin{enumerate}[wide, label={\rm{(\alph*)}}]
\item 
Given $\dg M$, $\gamma$, and $\epsilon > 0$, 
we can find a \cactree\ $\epsilon$ close in $\ell_2$ norm to 
$\bmu^*$
in time
\onlyfirsttime{%
\unskip$^\text{\ref{note:<4possible}}$%
}
\begin{align*}
    O\pqty[\bigg]{&|\V\!\Ar + \V\C|^{4}
        \pqty[\Big]{ \log |\V\!\Ar + \V\C| + \log \frac1\epsilon} k^2 \log k 
    }
    \\
    &\subseteq
    \tilde O\pqty[\Big]{k^2 |\V\!\Ar + \V\C|^{4}
        \log \nf 1\epsilon 
    }
    \\
    &\subseteq
    \tilde O\pqty[\Big]{k^2 (N+A)^4\,V^{4(T+1)}
         \log \nf 1\epsilon }
.
\end{align*}
\item
The unique \actree\ closest to $\bmu^*$ 
in which every component is represented with a $k$-bit binary number,
can be calculated in time 
\unskip$^\text{\ref{note:<4possible}}$%
\[
    \tilde O\pqty[\Big]{k^2 |\V\!\Ar + \V\C|^{4}}
    \subseteq
    \tilde O\pqty[\Big]{k^2 (N\!+\!A)^{4}\, V^{4(T+1
    )}}.
\]
\end{enumerate}
\end{linked}

Observe that the dependence on the precision is $\log (\nf1\epsilon)$, which is optimal in the sense that, in general, it takes time $\Omega(\log \nf 1\epsilon)$ to write down the binary representation of any number within $\epsilon$ of a given value.
\unskip\footnote{
    More precisely: if a value $x$ is chosen uniformly from $[0,1]$, then
    with probability $1-\sqrt\epsilon$ the binary representation of every $y \in [x-\epsilon, x+\epsilon]$ 
    has at least $\lfloor \frac12 \log_2 \nf1\epsilon \rfloor -1$ bits.
    }
In practice, this procedure can be used as if it were an exact algorithm,
with no more overhead than that incurred by floating point arithmetic.

\section{APPROXIMATION, HARDNESS, AND A DEEP CONNECTION BETWEEN INCONSISTENCY AND INFERENCE}
\def\ApproxPDGInfer{\textsf{APPROX-PDG-INFER}}
\def\ApproxPDGInc{\textsf{APPROX-CALC-INC}}
\def\ApproxInferUniq{\textsf{APPROX-INFER-CVX}}

While \cref{theorem:main} gives us a way of doing inference to machine precision in polynomial time, which is the typical use case of an exact algorithm, it is not technically an exact inference algorithm.
Indeed, if we require binary representations of numbers, exact inference for PDGs is technically not possible in finite time: in a PDG, the exact answer to an inference query may be an irrational number (even if all components of $\mathbb P$,$\balpha$, and $\bbeta$ are rational).
This leads us to formulate approximate inference more precisely.

\begin{defn}[approximate PDG inference]
    An instance of problem \ApproxPDGInfer\ 
    is a tuple $(\dg M, \gamma, Q, \epsilon)$, where
    $\dg M$ is a PDG with variables $\X$, $\gamma \in \{0^+\} \cup [0, \infty]$
    is the relative importance of structural information,
    $Q$ is a conditional probability query of the form
    ``$\Pr(Y{=}y|X{=}x) = ?$'', where $X,Y \subseteq \X$ and $(x,y) \in \V(X,Y)$,
    and $\epsilon > 0$ is the precision desired for the answer.
    A solution to this problem instance is a pair of numbers
    $(r^-, r^+)$
    such that
    \begin{align*}
        r^- &\le \inf_{\mu \in \bbr{\scalebox{0.7}{$\dg M$}}^*_\gamma} \mu(Y{=}y | X{=}x) \le r^- + \epsilon \\
        \text{and} \qquad
        r^+ &\ge \sup_{\mu \in \bbr{\scalebox{0.7}{$\dg M$}}^*_\gamma} \mu(Y{=}y | X{=}x) \ge r^+ - \epsilon. 
        \qedhere
    \end{align*}
\end{defn}

The problem we solved in \cref{sec:inf-as-cvx-program,sec:clique-tree-expcone}
is the special case in which $\dg M$ is assumed to be proper and $\gamma \in \{0^+\} \cup (0, \min_a \frac{\beta_a}{\alpha_a})$.  This is enough to ensure there is a unique optimal distribution $\mu^* \in \bbr{\dg M}^*_\gamma$, with respect to which we must answer all queries.  In this case, the definition above 
essentially amounts to providing a single
number $p$ such that $p-\epsilon \le \mu^*(Y{=}y|X{=}x) \le p+\epsilon$.  
We call this easier subproblem \ApproxInferUniq.
We will also be interested in the unconditional variants of both inference
problems, in which  
    no additional evidence is supplied (i.e., $X = \emptyset$). 
We now define the analogous problem of approximately calculating a PDG's  degree of inconsistency. 

\begin{defn}[approximate inconsistency calculation]
    An instance of problem 
    \ApproxPDGInc\ 
    is a triple $(\dg M, \gamma, \epsilon)$, where
    $\dg M$ is a PDG, $\gamma \ge 0$, and $\epsilon > 0$ is the desired precision. 
    A solution to this problem instance is 
    a number $r$ such that $|\aar{\dg M}_\gamma - r | < \epsilon$.
\end{defn}

The interior point method behind \cref{theorem:main} solves \ApproxPDGInc\ 
    in the process of finding a \actree\ for inference. 
But, technically, it does not solve \ApproxPDGInfer. 
A solution to \ApproxPDGInfer\ is
a conditional probability, not a \cactree.
While a \cactree\ does allow us to compute conditional
probabilities, 
an $\epsilon$-close \actree\ does not give us $\epsilon$-close answers
    to probabilistic queries, especially those conditioned on improbabable events (i.e., finding $\Pr(Y{=}y|X{=}x)$ when $\Pr(X{=}x) \approx 0$).
Nevertheless, because precision is so cheap,
the interior point method behind \cref{theorem:main}
can still be used as a subroutine to solve \ApproxInferUniq. 

\begin{linked}{theorem}{approx-infer}
\ApproxInferUniq\ can be solved in
\def\mustar{\mu^{\mskip-2mu*\!}}
\begin{align*}
    \tilde O \pqty[\bigg]{ \! (N\!+\!A)^4 V^{4(T+1)}
    \log \frac1{\epsilon \mustar(x)\!}
    \Big[
          \log \frac{\beta^{\max}\!}{\beta^{\min}\!} 
           + \log \frac1{\epsilon  \mustar(x)} 
    \Big]\! }
\end{align*}
time,
\onlyfirsttime{\unskip\footnote{%
   \label{note:<4possible}%
   At the cost of substantial overhead and engineering effort, the exponent $4$ can be reduced to 2.872, by appeal to \textcite{skajaa2015homogeneous} and the current best matrix multiplication algorithm \parencite[$O(n^{2.372})$]{duan2022faster} to invert $n{\times} n$ linear systems. 
}}
where $\mustar(x)$ is the probability of the event $X{=}x$ in the optimal
distribution $\mu^*$, 
 $\beta^{\max} := \max_{{a \in \Ar}} \beta_a$ is the largest observational confidence, 
\[
\text{and}\qquad
\beta^{\min} := 
\begin{cases}
    ~\displaystyle\min_{a \in \Ar} \{ \beta_a : \beta_a > 0\}& \text{if}~ \gamma = 0^+\\[-0.5ex]
    ~\hfill\gamma~~ & \text{if}~ \gamma > 0~~\text{.}
\end{cases}
\]
\end{linked}

The factor of $\log(\nf1{\mu^*\!(x)})$ is unusual, 
but even exact inference algorithms typically must write down $\mu^*(X{=}x)$
on the way to calculating $\mu^*(Y{=}x|X{=}x)$, which implicitly incurs 
a cost of at least $\log(\nf1{\mu^*\!(x)})$. 
A Bayesian network with $N$ variables in which cpds are articulated to precision $k$
can have nonzero marginal probabilities as small as $2^{- N k}$,
    in which case the additional worst case overhead for small probabilities is linear. 
We conjecture that it is not possible to form smaller marginal probabilities with PDGs, although the question remains open. 
Algorithmically speaking,
\cref{theorem:approx-infer} extends \cref{theorem:main} 
in three key ways. 

\begin{enumerate}
\item We must request additional precision 
    to ensure that the marginal probabilities deviate at most $\epsilon$
    from the true ones.  This sense of of approximation effectively
    bounds the $\ell_1$ norm of $\bmu^* - \bmu$,
    while \cref{theorem:main} (a) bounds it $\ell_2$ norm and
    (b) bounds its $\ell_\infty$ norm.

\item We must introduce a loop to refine precision until we have a suitably precise estimate of $\Pr(X{=}x)$.

\item Rather than directly dividing our estimate of $\Pr(Y{=}y,X{=}x)$ by our estimate of $\Pr(X{=}x)$, we calculate something slightly more stable. 
\end{enumerate}

See the proof for details. 
One immediate corrolary of \cref{theorem:approx-infer} is that
\ApproxInferUniq\ $\subseteq \mathsf{EXP}$, without the assumption of bounded treewidth.

It may be worth noting that there is at least one instance in the literature where \emph{approximate} Bayesian Network (BN) inference is tractable for a subclass of models other than those of bounded treewidth: \textcite{Dagum-Luby-approximate} give a randomzied algorithm for the special class of Bayesian Networks that do not have extreme conditional probabilities. Specifically they show that, assuming a network with $N$ nodes, the inference problem is in $\mathsf{RP}(N, \nf1\epsilon)$. 
In addition to restricting to a differet class of models (bounded conditional probabilities, but not bounded treewidth), their approximation 
    algorithm in another significant respect: it is polynomial in $\nf1\epsilon$, rather than in $\log(\nf1\epsilon)$. 
Thus, the time it requires is exponential with respect to the number of requested digits,
    while our algorithm takes linear time.

Because PDGs generalize BNs, approximate inference
for PDGs is at least as hard as it is for BNs.

\begin{prop}[\parencite{roth-hardness-1996}]
        \label{bn-sharp-P-hard}
    \ApproxPDGInfer\ is \#P-hard. 
\end{prop}

Thus, the exponential time of \cref{theorem:approx-infer}
    is the best we could have hoped for, in the general case.  
The argument is due to \textcite{roth-hardness-1996}, 
although we have altered it somewhat.

\begin{figure*}
    \centering
        \includegraphics[width=0.67\linewidth]{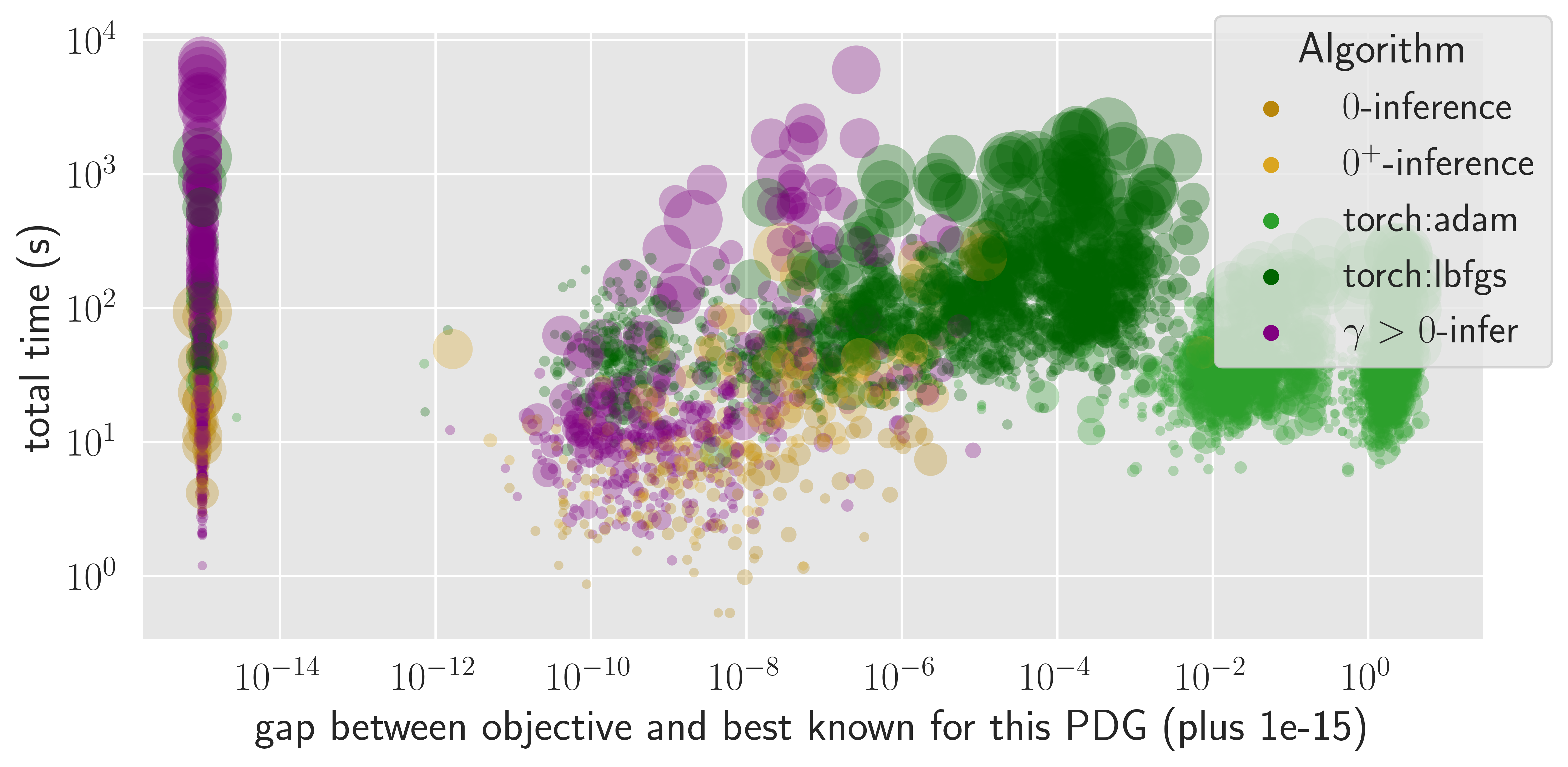}
        \includegraphics[width=0.32\linewidth]{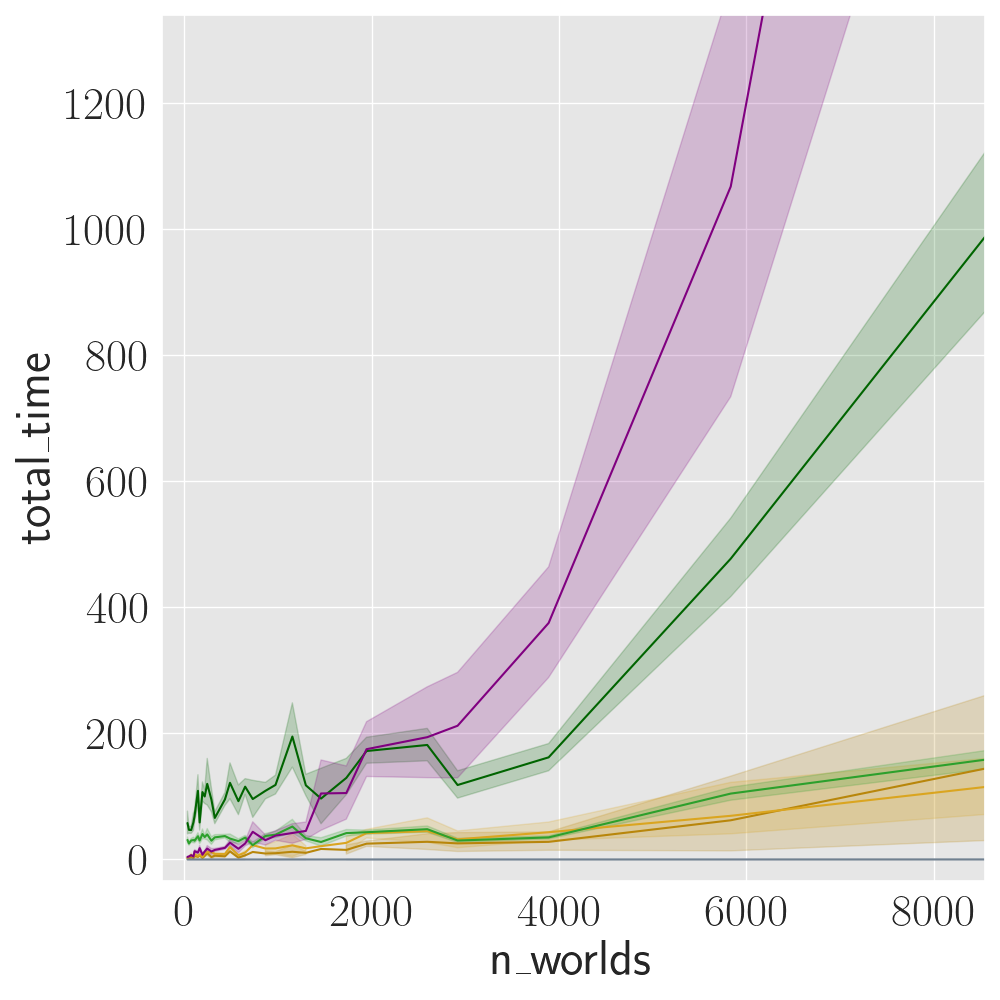}
    \caption{
        Accuracy and resource costs for the methods in \cref{sec:inf-as-cvx-program}.  
        Left: a scatter plot of several algorithms on random PDGs of $\approx 10$ variables. The x-axis is 
            the difference in scores
            $\bbr{\dg M}_\gamma(\mu) - \bbr{\dg M}_{\gamma}(\mu^*) + 10^{-15}$,
        where $\mu$ is the method's output,
        and $\mu^*$ achieves best (smallest) known value of $\bbr{\dg M}_{\gamma}$.
         (Thus, the best solutions lie on the far left.)
        The $y$ axis is the time required to compute $\mu$. 
        Our methods are in gold $(0^+$\!-inference) and violet ($\zogamma$-inference, for $\zogamma>0$); the baselines (black-box optimizers applied directly to \eqref{eqn:scoring-fn}) are in green.
        The area of each circle is proportional to the size of the optimization problem, as measured by
        {\small\texttt{n\_worlds}}$:=$
        $|\V\!\X|$.
        Right: how the same methods scale in run time, as $|\V\!\X|$ increases.
     }\label{fig:joint-gap-time}
\end{figure*}

Our approach to $\zogamma$-inference computes
$\aar{\dg M}_\gamma$ as a side effect.
But suppose that we were interested in calculating only this inconsistency.
Might there be a more direct, asymptotically easier way 
to do so? In general, the answer is no.

\begin{linked}{theorem}{consistent-NP-hard}
    \label{prop:sharp-p-hard}
    \begin{enumerate}[label={\rm{(\alph*)}}]
    \item Determining whether 
    there is a distribution
    that satisfies all cpds of a PDG
    is NP-hard.
    \item Calculating a PDG's degree of inconsistency (exactly) is \#P hard.
    \item \ApproxPDGInc\ is \#P hard,
        even for fixed $\gamma \ge 0$ and $\epsilon > 0$.
    \end{enumerate}
\end{linked}

\textcite{pdg-aaai}'s original approach to  
inferring the probability of 
$Y$ in a PDG $\dg M$ was to minimize their combined inconsistency. 
The idea is to add a hypothesis distribution $h(Y)$ to $\dg M$, and 
adjust $h$ to minimize the overall inconsistency $\aar{\dg M + h}_\gamma$.
Parts (b) and (c) of \cref{prop:sharp-p-hard} significantly undermine this approach, because even just calculating $\aar{\dg M + h}_\gamma$ is intractable. 
Typically minimizing a function is more difficult than evaluating it, so one might imagine the intractability of $\aar{\dg M + h}_\gamma$ to be merely the first of many difficulties---yet it turns out to be the only one. 
There is a strong sense in which being able to calculate inconsistency is enough to perform inference efficiently.
Specifically, 
with oracle access to the inconsistency $\aar{\dg M + h}_\gamma$\,, 
\citeauthor{pdg-aaai}'s approach gives right answer with the best possible asymptotic time complexity.
Thus, while it may not be a practical inference algorithm, it is a powerful reduction from inference to inconsistency calculation. 

\begin{linked}{theorem}{inf-via-inc-oracle}
\begin{enumerate}[label={\rm{(\alph*)}}]
    \item 
    There is an $O(\log \nf 1\epsilon)$-time reduction
    from unconditional
    \ApproxInferUniq\ to the problem of determining which of two PDGs is more inconsistent,
    using $O(\log \nf1\epsilon)$ subroutine calls.
\item 
    There is an 
    $\displaystyle
    O \left(
    \log \frac{\aar{\dg M}_\gamma}{\gamma\epsilon\, \mu^*\mskip-2mu(x)}
    \cdot
    \log \frac1{\epsilon\, \mu^*\mskip-2mu(x)}
    \right)
    \vphantom{\Bigg|}
    $
    time 
    reduction
    from \ApproxInferUniq\ 
    to \ApproxPDGInc\ 
    using $O(\log( \nf1\epsilon) \log \log \nicefrac1{\mu^*\mskip-2mu(x)})$ calls to the inconsistency subroutine. 
\item
    There is also an $O(|\V \C|)$ reduction
    from \ApproxPDGInc\ 
    to \ApproxInferUniq.
    With the additional assumption of bounded treewidth,
    this is linear in the number of variables in the PDG.
\end{enumerate}
\end{linked}

Recall that the runtime of $O(\log(\nf1\epsilon))$ achieved by part (a) is optimal, because it is the complexity of writing down an answer, which in general requires $\log(\nf 1\epsilon)$ bits. 
While it is a clean result, part (a) is unsatisfying as a complexity result because it relies heavily on being able to compare the two inconsistencies in constant time.
Part (b) fleshes out the algorithm of part (a) more precisely 
    by reducing to inconsistency approximation (which we now know is computable), 
    and also extends the procedure to handle to conditional probability queries. 
This leads to a significantly more complex analysis, and a more expensive reuction, although it is possible that much of the difference in the costs is due to loose bounds in our analysis. 
Part (c) is a straightforward observation in light of 
    the results in \cref{sec:clique-tree-expcone}.

To summarize: in the range of $\gamma$'s in which we have an (approximate) inference algorithm for PDGs, (approximately) calculating a PDG's degree of inconsistency is at least as difficult.
For PDGs of bounded treewidth, the two problems are equivalent, and can be solved in polynomial time.

\section{Experiments} \label{sec:expts}

We have given the first algorithm to provably do inference in polynomial
time, but that does not mean that it is the best way of answering queries in practice;
it also makes sense to use black-box optimization tools such as
    Adam \parencite{kingma2014adam} or L-BFGS \parencite{fletcher2013practical}
    to find minimizers of $\bbr{\dg M}_\gamma$.
Indeed, this scoring function has several properties
    that make it highly amenable to such methods: it is
    infinitely differentiable, $\gamma$-strongly convex, and its
    derivatives have simple closed-form expressions.
So it may seem surprising that $\bbr{\dg M}_\gamma$ poses
a challege to standard optimization tools---%
but it does, 
even when we optimize directly over joint distributions.

\textbf{Synthetic Experiment 1 ({\normalfont over joint distributions}).~~} 
Repeatedly do the following.
First, randomly generate a small PDG $\dg M$ containing 
at most 10 variables and 15 arcs. 
Then for various values of
$\gamma \in \{0, 0^+, 
    10^{-8}, 
     \ldots, \min_a \frac{\beta_a}{\alpha_a} \}$,
optimize $\bbr{\dg M}_\gamma(\mu)$ over joint distributions $\mu$, 
in one of two ways. 
\begin{enumerate}[wide,label=(\alph*),nosep,itemsep=0.2ex]
\item Use \verb|cvxpy| \parencite{diamond2016cvxpy}
to feed  
one of problems (\ref{prob:joint-inc},\ref{prob:joint-small-gamma},\ref{prob:joint+idef})
    to the MOSEK solver \parencite{mosek}, or
\item Choose a learning rate and a representation of $\mu$ in terms of optimization variables $\theta \in \mathbb R^n$.
    Then run a standard optimizer (Adam or L-BFGS) built into \verb|pytorch| \parencite{pytorch}
    to optimize $\theta$
    until $\mu_\theta$ converges to a minimizer of $\bbr{\dg M}_\gamma$ 
        (or a time limit is reached).
    Keep only the best result across all learning rates. 
\end{enumerate}

\begin{figure*}
    \centering
    \includegraphics[width=0.34\linewidth]{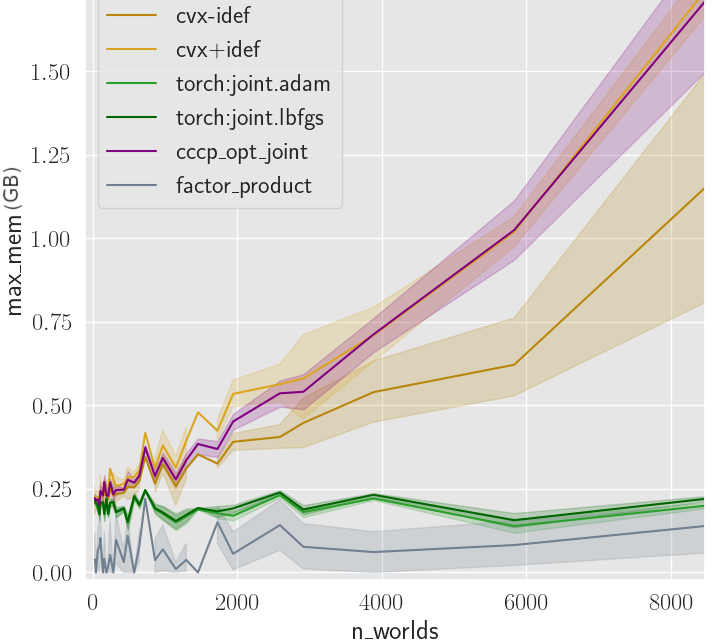}
    \includegraphics[width=0.65\linewidth]{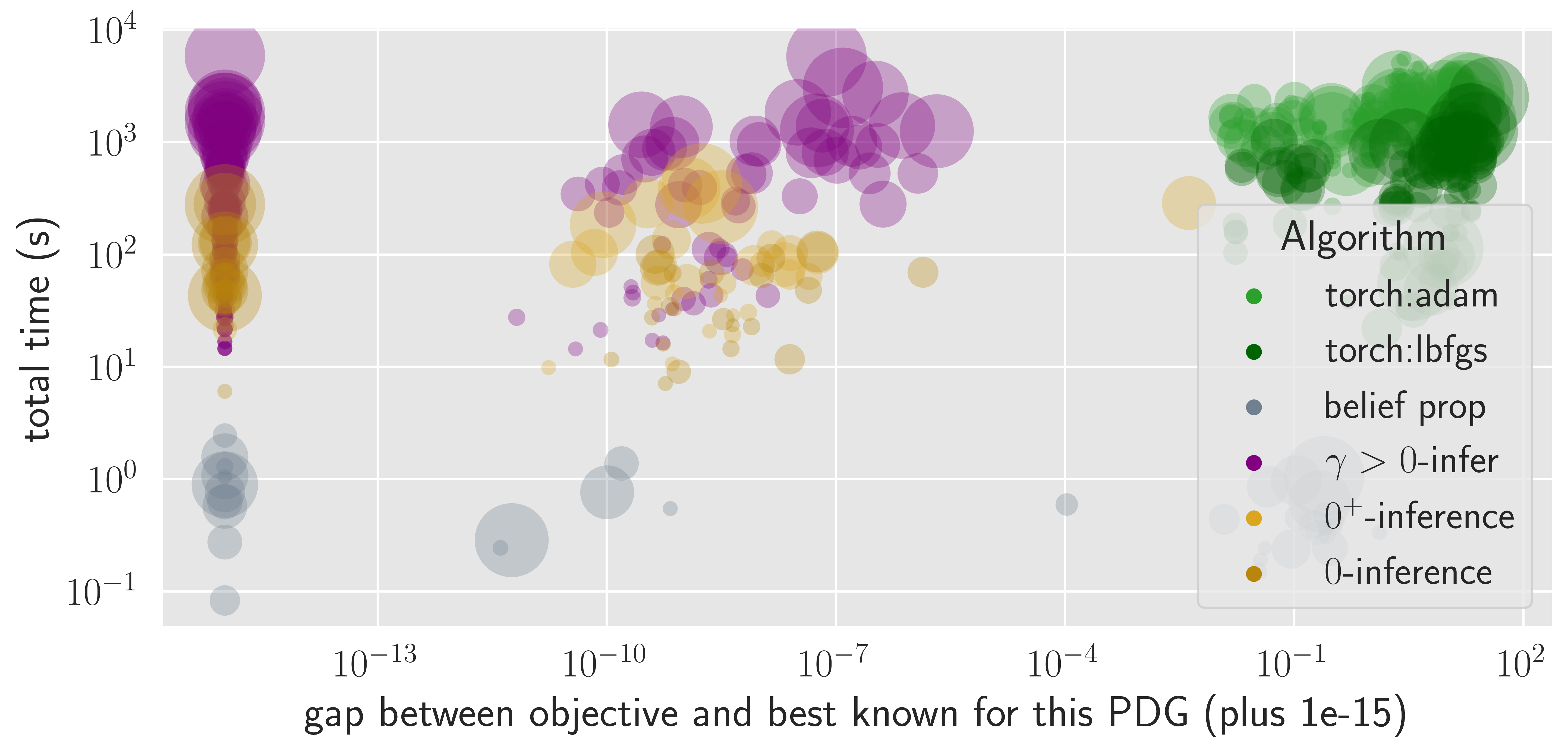}
    \caption{Left: Memory footprint.
    The convex solver (violet, gold)
     requires more memory than baselines (green).
    Right: Analogue of \cref{fig:joint-gap-time} for the cluster setting.
     Here there is even more separation between exponential conic optimization
         (gold, violet) and black-box optimization (greens).
     The grey points represent belief propagation, which is fastest and most accurate---%
         but only applies in the special case when $\bbeta=\gamma\balpha$.}
    \label{fig:joint-mem}
    \label{fig:clus-gap-vs-time}
\end{figure*}

The results are shown in \cref{fig:joint-gap-time}.
Observe that the convex solver (gold, violet) is significantly more accurate than the baselines,
and also much faster for small PDGs.
Our implementation of $0^+\!$-inference (gold) also appears to scale better than L-BFGS
    in this regime, although
    that of $\zogamma$-inference (purple) seems to scale much worse. 
We suspect that the difference comes from \verb|cvxpy|'s complilation process,
    because the two use similar amounts of memory (\cref{fig:joint-mem}),
    and so are problems of similar sizes.

\textbf{Synthetic Experiment 2 ({\normalfont over \actree s}).~~} 
For PDGs of bounded treewidth, \Cref{coro:can-use-cliquetree} allows us to express these optimzation problems compactly not just for the convex solver, but for the black-box baseline approaches as well.
We adapt the previous experiment for \actree s as follows.
First randomly sample a maximal graph $G$ of tree-width $k$, called a 
    $k$-tree \parencite{patil1986structure}; 
then generate a PDG $\dg M$ whose hyperarcs lie within cliques of $G$.
This ensures that the maximal cliques of $G$ form a tree-decomposition $(\C, \mathcal T)$ of $\dg M$'s underlying
    hypergraph.
We can now proceed as before: 
    either encode
    (\ref{prob:cluster-inc},\ref{prob:cluster-small-gamma},\ref{prob:cluster+idef})
    as disciplined convex programs in \verb|cvxpy|,
or use \verb|torch| to directly minimize $\bbr{\dg M}_\gamma(\Pr_{\!\bmu})$
    amongst \actree s $\bmu$ over $(\C, \mathcal T)$.

In the latter case, however, there is now an additional difficulty:
    it is not easy to strictly enforce the calibration constraints with the black-box methods.
Common practice is to instead add extra loss terms to ``encourage'' calibration---%
but
    it can still be worthwile for the optimizer to simply incur that loss 
    in order to violate the constraints.
Thus, for fairness, we must recalibrate the 
    the \actree s returned by all methods before evaluation.
The result is an even more significant advantage for the convex solver; 
see \cref{fig:clus-gap-vs-time}.

\textbf{Evaluation on BNs.~~}
We also applied the procedure of the Synthetic Experiment 2
to the smaller BNs in the
\href{https://www.bnlearn.com/bnrepository/}{\texttt{bnlearn}} repository,
and found similar results (but with fewer examples; see \cref{sec:bn-expt-details}). 
But for a PDG that happens to also be a BN, it is possible to use belief propagation, which is much faster and at least as accurate.

Explicit details about all of our experiments, 
and many more figures, can be found in \cref{sec:expt-setup}.

\section{Discussion and Conclusion}

In this paper, we have provided the first practical algorithm for
    inference in PDGs. 
In more detail, we have defined a parametric family of PDG inference notions, 
given a fixed-parameter tractable inference algorithm for a subset of these parameters,
    proven our algorithm correct, implemented it, and
    shown our code to empirically outperform baselines.
Yet many questions about PDG inference remain open.

Asymptotically, there may be a lot of room for improvement.
Our implementation runs in time $\tilde O(N^4)$, and our analysis suggests one of time $\tilde O(N^{2.872})$. 
But assuming bounded tree-width, most graph problems, including
inference inference for BNs and FGs, can be solved in time $O(N)$.

Furthermore, we have shown how to do inference for only a subset of
possible paramer values, specifically, 
when
either $\bbeta \ge \gamma \balpha$ or $\bbeta \gg \balpha$. 
The remaining cases are also of interest, and likely require different techniques. 
When $\bbeta = 0$ and $(\Ar, \balpha)$ encodes the structure of a BN,
    for instance,
    inference is about characterizing the BN's independencies.
While we do not know how to tackle 
the inference problem in the general setting, 
our methods can be augmented with the convex-concave procedure 
    \parencite{yuille2003concave} to obtain an inference
    algorithm that applies slightly more broadly; see \cref{sec:cccp}.
We imagine that this extension could also be useful for computing with PDGs 
    beyond the specific inference problem considered in this paper.

Given the long history of improvements to our
understanding of inference for 
    Bayesian networks,
we are optimistic that 
    faster and more general
    inference algorithms
    for PDGs
    are possible.
\discard{
Our analysis does not resove these problems, but it
    does shed light on some of them.  
The $0$-semantics, for instance, is 
characterized by \cref{prop:marginonly,prop:cluster-inc-correct}, 
Also, when $\bbr{\dg M}_\gamma$ is not convex, we can still find an optimal distribution with the concave-convex procedure \cite{yuille2003concave}, which we do in \cref{sec:cccp}---but this only suffices for inference if we already know there's a unique optimal distribution.}
\discard{
In some cases, this might actually allow us to do inference---say, if we happen to know for external reasons that $\bbr{\dg M}^*_\gamma$ is pseudo-convex (although we loose polynomial time guarantees and have no ability to automatically recognize such situations). In any case, we have implemented this, and describe it in \cref{sec:cccp}.}

\clearpage
\onecolumn
\appendix
\ifvfull\else
\section{Full Descriptions of \AcTree\ Convex Programs}
\label{appendix:prob-details}

\subsection{Finding a \AcTree\ for \texorpdfstring{$\zogamma$}{gamma}-Inference, when \texorpdfstring{$\gamma$}{gamma} is Small}

We now give the rest of the details for problem \eqref{prob:cluster-small-gamma},
which was described in \cref{sec:clique-tree-expcone}, and we claimed was correct in \cref{prop:cluster-small-gamma-correct}.
Over vectors
$\mat v, \bmu \in \Rext^{\V\C}$ and
$\mat u \in \Rext^{\V\!\Ar}$,
the problem is formally given by:

{\allowdisplaybreaks%
\begin{align*}
    \minimize_{\bmu, \mat u, \mat v}\quad&
    \sum_{\mathrlap{\!\!\!(a,s,t) \in \V\!\Ar}} (\beta_a \!- \alpha_a \gamma) u_{a,s,t}
    + \gamma \sum_{\mathclap{(C,c) \in \V\C}}  v_{C,c}
    \numberthis\label{prob:cluster-small-gamma}
    - \sum_{\mathrlap{\!\!\!(a,s,t) \in \smash{\V\!\Ar^+}}}
        \alpha_a\gamma\,
        \mu_{C\!_a}\!(\Src a{=}s,\Tgt a {=} t)
        \log \p_a (\Tgt a{=}t\mid s)
\\[0.2ex]
\subjto\quad&
    \forall C \in \C.~\mu_C \in \Delta\V(C), \\[-0.2ex]
    &\forall a \in \Ar.~
        \big(\!- \! \mat u_a,\, \mu_{C\!_a}\!(\Src a,\mskip-2mu \Tgt a),\, \mu_{C\!_a}\!(\Src a) \p_a(\Tgt a | \Src a)\big) \in K_{\exp}^{\V a}, \\
    &\forall (a,s,t) \in \V\!\Ar^0\!.~
    \mu_{C\!_a}\!(\Src a{=}\mskip2mus, \Tgt a{=}\mskip2mut) = 0, \\[-0.2ex]
    &\forall (C,D) \in \mathcal T.~~ \mu_{C}(C \cap D) = \mu_{D}(C \cap D),\\[-0.3ex]
    &\big(-\mat v,\,\bmu,\, [\,\mu_{C}(\Pash_C(c))\,]_{(C,c)\in\V\C}\big) \in K_{\exp}^{\V\C}
    .
\end{align*}}%
This is our most complex problem, but it is made up of building blocks
that we have discussed at length in the main paper: the objective comes from \eqref{eq:altscore}, the first and fourth lines of constraints restrict to \cactree s as in \eqref{prob:cluster-inc}, the second and third lines of constraints are essentially present in \eqref{prob:joint-small-gamma},
and the final constraint comes from \eqref{eq:cluster-ent-decomp}.

\subsection{Finding a \AcTree\ for $0^+$\!-Inference}
Next, we deal with the analogous construction of the \actree\ optimization for $0^+$\!-inference.
We begin with the straightforward adaption of the relevant prerequisites in \Cref{sec:empirical-limit}.
Suppose that $\boldsymbol\nu = \{\nu_C : C \in \C\}$ is a \cactree\ over the tree decomposition $(\C, \mathcal T)$ representing a distribution $\Pr_{\boldsymbol\nu} \in \bbr{\dg M}^*_0$, say obtained by solving \eqref{prob:cluster-inc}.
For $C \in \C$, let $\Ar_C:= \{ a \in \Ar : C_a = C\}$ be the set of
edges assigned to cluster $C$, and let
\[
    \mat k := \bigg[ \prod_{\mathrlap{a \in \Ar_C}} \nu_C (\Tgt a (c) | \Src a (c))^{\alpha_a} \bigg]_{(C,c) \in \V \C} \in \Rext^{\V\C}
\]
be the analogue of \eqref{eq:cm-product} for a cluster tree.
Once again, consider
$\mat u := [ u_{(C,c)} ]_{(C,c) \in \V\C}$,
in the optimization problem
\begin{align*}
\minimize_{\bmu, \mat u} \quad&
    \mat 1^{\sf T} \mat u
    \numberthis\label{prob:cluster+idef}\\
\subjto\quad&
    \forall C \in \C.~\mu_C \in \Delta\V(C), \\[-0.2ex]
     &\big({-}\mat u,\,  \bmu,\,\,
            \mat k \odot
            \big[\;\mu_C(\Pash_C(c))\;\big]_{(C,c) \in \V\C}
            \big) \in K_{\exp}^{\V\C}, \\[-0.2ex]
    &\forall a \in \Ar.~~\mu_{C_{\!a}}\!(\Src a, \Tgt a) \nu_{C_{\!a}}\!(\Src a) = \mu_{C_{\!a}}\!(\Src a) \nu_{C_{\!a}}\!(\Src a, \Tgt a)\\
    &\forall (C,D) \in \mathcal T.~~ \mu_{C}(C \cap D) = \mu_{D}(C \cap D).
\end{align*}
The biggest change is in the second constraint:
intuitively, the upper bounds at each cluster $C$ now only account
for the \emph{additional} entropy not already modeled by
$C$'s
ancestors.
By \cref{prop:cluster-idef-correct}, if given a proper PDG, the output $\bmu$ of this problem is the unique distribution $0^+$\!-semantics.

\clearpage
\fi

\section{Proofs}

Our results fall broadly into three categories:
\begin{enumerate}
    \item Foundational results about PDGs that we needed to prove to get an
        inference procedure, but which are likely to be generally useful
        for anyone working with PDGs
            (\cref{proofs:novel-pdg-results});
    \item Correctness and efficiency results, showing that the optimization
        problems we present in the main paper give the correct answers,
        and that they can be formulated and solved in polynomial time;
            (\cref{proofs:expcone-efficient-correct})
    \item Hardness results, i.e., \cref{theorem:inf-via-inc-oracle} and
    the constructions and lemmas needed to support it
        (\cref{proofs:hardness-results}).
\end{enumerate}
\subsection{Novel Results about PDGs}
    \label{proofs:novel-pdg-results}

\recall{prop:marginonly}
\begin{lproof}\label{proof:marginonly}
    For contradiction, suppose that $\mu_1, \mu_2 \in \bbr{\dg M}_0^*$, but
    there is some $(\hat a, \hat s, \hat t) \in \V\!\Ar$ such that $\beta_a > 0$ and
    \[
        \mu_1(\Tgt a{=}\hat t, \Src a{=}\hat s)\mu_2(\Src a{=}\hat s) \ne \mu_2(\Tgt a{=}\hat t, \Src a{=}\hat s) \mu_1(\Src a\hat s).
    \]
    For $t \in [0,1]$,
    let $\mu_t := (1-t) \mu_0 + t \, \mu_1$ as before.
    Then define
    \begin{align*}
        F(t) := \kldiv[\Big]{ \mu_t(\Src a, \Tgt a) }{  \mu_t(\Src a) \p_a(\Tgt a|\Src a) }.
    \end{align*}
    Since $\mu_0(\Src a, \Tgt a)$ and  $\mu_1(\Src a, \Tgt a)$ are joint distributions voer two variables, with different conditional marginals, as above, \cref{lem:seg-strictcvx} applies, and so $F(t)$ is strictly convex.

    Let
    \[ \OInc_{\dg M \setminus \hat a}
        := \sum_{a \ne \hat a} \beta_a \kldiv{\mu(\Tgt a, \Src a)}{\p_a(\Tgt a|\Src a) \mu(\Src a)}
    \]
    be the observational incompatibility loss, but without the term corresponding to edge $\hat a$.
    Since $\OInc_{\dg M \setminus \hat a}$ is convex in its argument, it is in particular convex along the segment from $\mu_0$ to $\mu_1$; that is, for $t \in [0,1]$, the function $t \mapsto \OInc_{\dg M \setminus \hat a}(\mu_t)$ is convex.
    Therefore, we know that the function
    \begin{align*}
        G(t) :=
        \OInc_{\dg M}(\mu_t)
        =
        \OInc_{\dg M \setminus a}( \mu_t ) + \beta_a\, F(t),
    \end{align*}
    is \emph{strictly} convex.
    But then this means $\mu_{\nf12}$ satisfies
    \[
        \OInc_{\dg M}( \mu_{\nicefrac12} ) < \OInc_{\dg M}( \mu_0 ),
    \]
    contradicting the premise that $\mu_0$ minimizes $\OInc_{\dg M}$ (i.e., $\mu_0 \in \bbr{\dg M}^*_0$).
    Therefore, it must be the case that all distributions in $\bbr{\dg M}_0^*$ have the same conditional marginals, as promised.
\end{lproof}

\clearpage
\recall{prop:idef-frozen}
\begin{lproof}\label{proof:idef-frozen}
    This is mostly a simple algebraic manipulation. By definition:
    \begin{align*}
        \SInc_{\dg M}(\mu) &= - \H(\mu) + \sum_{a \in \Ar} \alpha_a \H_\mu(\Tgt a | \Src a) \\
        &= \Ex_\mu \left[ - \log \frac{1}{\mu} + \sum_{a \in \Ar} \alpha_a \log \frac{1}{\mu(\Tgt a|\Src a)} \right] \\
        &= \sum_{w \in \V\!\X} \mu(w) \left[ \log \mu(w) + \sum_{a \in \Ar} \log \frac{1}{\mu(\Tgt a(w)|\Src a(w))^{\alpha_a}} \right] \\
        &= \sum_{w \in \V\!\X} \mu(w) \log \pqty[\bigg]{ \faktor{\mu(w)~}{~\prod_{a \in \Ar}\mu(\Tgt a(w)|\Src a(w))^{\alpha_a}}}
    \end{align*}
    But, by \cref{prop:marginonly}, if we restrict $\mu \in \bbr{\dg M}_0^*$, then the conditoinal marginals in the denominator do not depend on the particular choice of $\mu$; they're shared among all $\nu \in \bbr{\dg M}_0^*$.
\end{lproof}

\recall{theorem:markov-property}

\[
    \text{Or symbolically: }\qquad\quad
    \dg M_1 \bundle \dg M_2
        ~\models~
    \X_1 \mathbin{\bot\!\!\!\bot} \X_2 \mid \X_1 \cap \X_2. \]
\begin{lproof}\label{proof:markov-property}
    Note that,
    save for the joint entropy, every summand the scoring function $\bbr{\dg M_1 + \dg M_2}_\gamma : \Delta(\V\!\X_1 \times \V\!\X_2)$, is a function of the conditional marginal of $\mu$ along some edge.
    In particular, those terms that correspond to edges of $\dg M_1$ can be computed from the marginal $\mu(\X_1)$, while those that correspond to edges of $\dg M_2$ can be computed from the marginal $\mu(\X_2)$.
    Therefore, there are functions $f$ and $g$ such that:
    \[
        \bbr{\dg M_1 \bundle \dg M_2}_\gamma(\mu) = f(\mu(\X_1)) + g(\mu(\X_2)) - \gamma \H(\mu).
    \]

    To make this next step extra clear, let $\mat X := \X_1 \setminus \X_2$ and
    $\mat Z := \X_2 \setminus \X_1$, be the variables unique to each PDG, and $\mat S:= \X_1 \cap \X_2$ be the set of variables they have in common, so that $(\mat X, \mat S, \mat Z)$ is a partition of all variables $\mat X_1 \cup \mat X_2$.
    Now define a new distribution $\mu' \in \Delta(\V\!\X_1 \times \V\!\X_2)$ by
    \[
        \bf
        \mu'( X,  S,  Z)
            := \mu(S) \mu( Z \mid  S)\mu( X \mid  S)
            \qquad \Big(~
            = \mu( X,  S) \mu( Z \mid  S)
            = \mu( Z,  S) \mu( X \mid  S)~\Big).
    \]
    One can easily verify that $\mat X$ and $\mat Z$ are independent given $\mat S$ in $\mu'$ (by construction), and the alternate forms on the right make it easy to see that $\mu(\X_1) = \mu'(\X_1)$ and $\mu(\X_2) = \mu'(\X_2)$.
    Furthermore, for any $\nu'(\mat{X,S,Z})$, we can write
    \begin{align*}
        \H( \nu ) &=  \H_\nu(\mat{X,S,Z}) =
            \H_\nu(\mat{X,S}) + \H_\nu(\mat Z \mid \mat{X,S}) \\
            &= \H_\nu(\mat{X,S}) + \H_\nu(\mat Z \mid \mat{X,S}) - \H_\nu(\mat Z \mid \mat S) + \H_\nu(\mat Z \mid \mat S) \\
            &= \H_\nu(\mat X,\mat S) + \H_\nu(\mat Z \mid \mat S) - \I_\nu(\mat Z; \mat X| \mat S),
    \end{align*}
    where $\I_\nu(\bf X;Z|S)$, the conditional mutual information between $\mat X$ and $\mat Z$ given $\mat S$ (in $\nu$), is non-negative, and equal to zero if and only if $\mat X$ and $\mat Z$ are conditionally independent given $\mat S$ \parencite[see, for instance,][\S1]{mackay2003information}.
    So $\I_{\mu'}(\mat X; \mat Z| \mat S) = 0$, and
        $\H_{\mu'} = \H_{\mu'}(\mat X, \mat S) + \H_{\mu'}(\mat Z| \mat S)$.
    Because $\mu$ and $\mu'$ share marginals on $\X_1$ and $\X_2$, while the terms $\H(\mat X, \mat S)$ and $\H(\mat Z|\mat S)$ depend only on these marginals, respectively, we also know that $\H_{\mu}(\mat X, \mat S) = \H_{\mu'}(\mat X, \mat S)$ and $\H_{\mu}(\mat Z | \mat S) = \H_{\mu'}(\mat Z| \mat S)$; thus we have
    \begin{align*}
        \H(\mu) &= \H_\mu(\mat X,\mat S) + \H_\mu(\mat Z \mid \mat S) - \I_\mu(\mat Z; \mat X| \mat S) \\
            &= \H(\mu') - \I_\mu(\mat Z; \mat X| \mat S).
    \end{align*}
    Therefore,
    \begin{align*}
        \bbr{\dg M_1 \bundle \dg M_2}_\gamma(\mu)
         &= f(\mu(\X_1)) + g(\mu(\X_2)) - \gamma \H(\mu) \\
         &= f(\mu'(\X_1)) + g(\mu'(\X_2)) - \gamma \H(\mu') + \gamma \I_\mu(\mat Z; \mat X| \mat S) \\
         &= \bbr{\dg M_1 \bundle \dg M_2}_\gamma(\mu') + \gamma \I_\mu(\mat Z; \mat X| \mat S).
    \end{align*}
    But conditional mutual information is non-negative, and by assumption, $\bbr{\dg M \bundle \dg M_2}_\gamma(\mu)$ is minimal. Therefore, it must be the case that
    \[
        \I_\mu(\mat Z; \mat X| \mat S) = \I_\mu(\X_1; \X_2 \mid \X_1 \cap \X_2) = 0,
    \]
    showing that $\X_1$ and $\X_2$ are conditionally independent given the varaibles that they have in common. \\
    (The fact that $\I_\mu(\mat Z; \mat X| \mat S) = \I_\mu(\X_1; \X_2 \mid \X_1 \cap \X_2)$ is both easy to show and an instance of a well-known identity; see CIRV2 in Theorem 4.4.4 of \textcite{halpern2017reasoning}, for instance.)
\end{lproof}

\recall{coro:can-use-cliquetree}
\begin{lproof}\label{proof:can-use-cliquetree}

    The set of distributions that can be represented by a \cactree\ over $(\C,\cal T)$ is the same as the set of distributions that can represeted by a factor graph for which $(\C, \cal T)$ is a tree decomposition.
    One direction holds because any such product of factors ``calibrated'', via message passing algorithms such as belief propogation, to form a \actree.
    The other direction holds because $\Pr_{\bmu}$ itself is a product of factors that decomposes over $(\C, \cal T)$.

    Alternatively, this same set of distributions that satisfy the independencies of the Markov Network obtained by connecting every pair of variables that share a cluster.
    More formally, this network is the graph $G := (\X, E := \{ (X{-}Y) :  \exists C \in \C.~\{X,Y\} \subseteq C\})$.
    Also, $G$ happens to chordal as well, which we prove at the end.

    Using only the PDG Markov property (\cref{theorem:markov-property}), we now show that every independence described by $G$ also holds in every distribution $\mu \in \bbr{\dg M}^*_\gamma$.
    Suppose that,
    for sets of variables $\mat X, \mat Y, \mat Z \subseteq \X$,
    $\I(\mat X; \mat Y|\mat Z)$ is an independence
    described by $G$.
    This means \parencite[Defn 4.8]{koller2009probabilistic} that
    if $X \in \mat X$, $Y \in \mat Y$, and $\pi$ is a path in $G$ between them, then
    some node along $\pi$ lies in $\mat Z$.

    Let $\cal T'$ be the graph that results from removing each edge $(C{-}D) \in \cal T$ that satisfies $C \cap D \subseteq \mat Z$, which is a disjoint union  $\mathcal T' = \mathcal T_1 \sqcup \ldots \sqcup \mathcal T_n$ of subtrees that have no clusters in common.
    To parallel this notation, let $\C_1, \ldots, \C_n$ be their respective vertex sets.
    Note that for every edge $e=(C{-}D)\in \cal T'$, there must by definiton be some variable $U_e \in (C \cap D) \setminus \mat Z$.

    We claim that no subtree $\mathcal T_i$ can have both a cluster $D_X$ containing a variable $X \in \mat X \setminus \mat Z$ and also a cluster $D_Y$ containing a variable $Y \in \mat Y \setminus \mat Z$.
    Suppose that it did.
    Then the (unique) path in $\cal T$ between $D_X$ and $D_Y$, which we label
    \[
    \begin{tikzcd}[column sep=2em]
        D_X=
        &D_0 \ar[r,-,"e_1"]&
        D_1 \ar[r,-,"e_2"]&
          \cdots
        \ar[r,-,"e_{m-1}"]& D_{m-1}
        \ar[r,-,"e_m"] & D_m&
        =D_Y
    \end{tikzcd},
    \]
    would lie entirely within $\mathcal T_i \subseteq \mathcal T'$. This gives rise to
    a corresponding path in $G$:
    \[\begin{tikzcd}[column sep=1em,row sep=1.5ex]
        X \ar[r,-] \ar[d,sloped,phantom,"\in"]
        & U_{e_1}\ar[r,-] \ar[d,sloped,phantom,"\in"]
        & U_{e_2}\ar[r,-] \ar[d,sloped,phantom,"\in"]
           &\cdots\ar[r,-]
        & U_{e_{n{-}1}} \ar[r,-] \ar[d,sloped,phantom,"\in"]
        & U_{e_n}\ar[r,-] \ar[d,sloped,phantom,"\in"]
        & Y \ar[d,sloped,phantom,"\in"]
            \\
        D_0
        & D_0 \cap D_1
        & D_1 \cap D_2
        &
        & D_{n{-}2} \cap D_{n{-}1}
        & D_{n{-}1} \cap D_n
        & D_n
    \end{tikzcd}\quad,\]
    and moreover, this path is disjoint from $\mat Z$.
    This contradicts our assumption that every path in $G$ between a member of $\mat X$ and a member of $\mat Y$ must intersect with $\mat Z$, and so no subtree can have both a cluster containing a variable $X \in \mat X \setminus \mat Z$ and also one containing $Y \in \mat Y \setminus \mat Z$.

    \def\CX{\C_{\mat X}}
    \def\CNX{\C_{\mat{Y}}^+}
    We can now partition the clusters as $\C = \CX \sqcup \CNX$, where
    $\CX$ is the set of the clusters that belong to subtrees $\mathcal T_i$ with a cluster containing some $X \in \mat X \setminus \mat Z$, and
    its $\CNX$ is its complement, which in particular contains those subrees have some $Y \in \mat Y \setminus \mat Z$.
    Or, more formally, we define
    \[
        \CX :=~ \bigcup_{\mathclap{\substack{i \in \{1,\ldots,n\}\\ (\cup\C_i) \cap (\mat X\setminus\mat Z) \ne \emptyset }}}\, \C_i
        \quad\qquad \text{and}\qquad
        \CNX :=~ \bigcup_{\mathclap{\substack{i \in \{1,\ldots,n\}\\ (\cup\C_i) \cap (\mat X\setminus\mat Z) = \emptyset }}}\, \C_i
        \quad.
    \]
    \def\XX{\X_{\mat X}}
    \def\XNX{\X_{\mat Y}^+}
    Let $\XX := \cup \CX$ set of all variables appearing in the clusters $\CX$; symmetrically, define $\XNX := \cup \CNX$.

    We claim that $\XX \cap \XNX \subset \mat Z$.
    Choose any variable $U \in \XX \cap \XNX$.
    From the definitions of $\XX$ and $\XNX$, this means $U$ is a member of some cluster $C \in \CX$, and also a member of a cluster $D \in \CNX$.
    Recall that the clusters of each disjoint subtree $\mathcal T_i$ either fall entirely within $\CX$ or entirely within $\CNX$ by construction.
    This means that $C$ and $D$, which are on opposite sides of the partition, must have come from distinct subtrees.
    So, some edge $e = (C'{-}D') \in \mathcal T$ along the (unique) path from $C$ to $D$ must have been removed when forming $\mathcal T'$, which by the definition of $\mathcal T'$, means that $(C' \cap D') \subset Z$.
    But by the running intersection property (\actree\ property 2), every cluster along the path from $C$ to $D$ must contain $C \cap D$---in particular, this must be true of both $C'$ and $D'$.
    Therefore,
    \[
        U \in C \cap D \subset C' \cap D' \subset \mat Z.
    \]
    So $\XX \cap \XNX \subset \mat Z$, as promised.  We will rather use it in the equivalent form $(\XX \cap \XNX) \cup \mat Z = \mat Z$.

    Next, since $(\C, \cal T)$ is a tree decomposition of $\Ar$, each hyperarc $a \in \Ar$ can be assigned to some cluster $C_a$ that contains all of its variables; this allows us to lift the cluster partition $\C = \CX \sqcup \CNX$ to a partition $\Ar = \Ar_{\mat X} \sqcup \Ar_{\mat Y}^+$ of edges, and consequently, a partition of PDGs $\dg M = \dg M_{\mat X} \bundle \dg M_{\mat Y}^+$.
    Concretely: let $\dg M_{\mat X}$ be the sub-PDG of $\dg M$ induced by restricting to the variables $\XX \subseteq \X$ arcs $\Ar_{\mat X} = \{ a \in \Ar : C_a \in \CX \} \subseteq \Ar$; define $\dg M_{\mat Y}^+$ symmetrically. (To be explicit: the other data of $\dg M_{\mat X}$ and $\dg M_{\mat Y}^+$ are given by restricting each of $\{\mathbb P,\balpha,\bbeta\}$ to $\Ar_{\mat X}$ and $\Ar_{\mat Y}^+$, respectively.)

    This partition of $\dg M$ allows us to use the PDG Markov property.
    Suppose that for some $\gamma > 0$ that $\mu \in \bbr{\dg M}^*_\gamma = \bbr{\dg M_{\mat X} \bundle \dg M_2}^*_\gamma$.
    We can then apply \cref{theorem:markov-property}, to find that
    $\XX$ and $\XNX$ are independent given $\XX \cap \XNX$.
    We use standard standard properties of random variable independence
        \parencite[CIRV1-5 of][Theorem 4.4.4]{halpern2017reasoning} to find that $\mu$ must satisfy:
    \begin{align*}
        \XX  &\CI \XNX \mid \XX \cap \XNX\\
    \implies\qquad
        (\XX \setminus \mat Z) &\CI (\XNX\setminus \mat Z) \mid (\XX \cap \XNX) \cup \mat Z
            & \big[\,\text{CIRV3}\,\big] \\
    \implies\qquad
        (\mat X \setminus \mat Z) &\CI (\mat Y \setminus \mat Z) \mid (\XX \cap \XNX) \cup \mat Z
        & \big[\,\text{by CIRV2, as $\mat X \subseteq \XX$ and $\mat Y \subseteq \XNX$}\,\big] \\
    \implies\qquad
        (\mat X \setminus \mat Z) &\CI (\mat Y \setminus \mat Z) \mid \mat Z
        & \big[\,\text{since $(\XX \cap \XNX) \cup \mat Z = \mat Z$}\,\big] \\
    \iff\qquad
        \mat X &\CI \mat Y \mid \mat Z
        & \hspace{-4em}\big[\,\text{standard; e.g., Exercise 4.18 of \textcite{halpern2017reasoning}}\,\big] \\
    \end{align*}

    Using only the PDG Markov property, we have now shown that every independence
    modeled by the Markov Network $G$ also holds
    in every distribution $\mu \in \bbr{\dg M}^*_\gamma$. Moreover, $G$ is chordal (as we will prove momentarily),
    and is well-known that distributions that have the independencies of a chordal graph can be can be represented by \actree s \parencite[Theorem 4.12]{koller2009probabilistic}.
    Therefore, there is a \actree $\bmu$ representing every $\mu \in \bbr{\dg M}^*_\gamma$.

    \begin{iclaim}
        $G$ is chordal.
    \end{iclaim}
    \begin{proof}
        Suppose that $G$ contains a loop $X{-}Y{-}Z{-}W{-}X$.
        Suppose further, for contradiction, that neither $X$ and $Z$ nor $Y$ and $W$ share a cluster.
        Given a variable $V$, it is easy to see that property (2) of the tree decomposition ensures that the subtree $\mathcal T(V) \subseteq \mathcal T$ induced by the clusters $C \in \C$ that contain $V$, is connected.
        By assumption, ${\cal T}(Y)$ and ${\cal T}(W)$ must be disjoint.
        There is an edge between $Y$ and $Z$, so some cluster must contain both variables, meaning ${\cal T}(Y) \cap {\cal T}(Z)$ is non-empty.
        Similarly, ${\cal T}(Z) \cap {\cal T}(W)$ is non-empty because of the edge between $Z$ and $W$.
        This creates an (indirect) connection in $\cal T$ between ${\cal T}(Y)$ and ${\cal T}(W)$. Because $\cal T$ is a tree, and ${\cal T}(Y) \cap {\cal T}(W) = \emptyset$,
        every path from a cluster $C_1 \in {\cal T}(Y)$ to a cluster $C_2 \in {\cal T}(W)$ must pass through ${\cal T}(Z)$, which is not part of ${\cal T}(Y)$ or ${\cal T}(W)$.
        ${\cal T}(X)$ and ${\cal T}(Y)$ intersect as well, meaning that, for any $C \in {\cal T}(X)$, there is a (unique) path from $C$ to that point of intersection, then across edges of ${\cal T}(Y)$, then edges of ${\cal T}(Z)$, and finally connects to the clusters of ${\cal T}(W)$. And also, since $\cal T$ is a tree, that path must be unique.
        The problem is that there is also an edge between $X$ and $W$, so there's some cluster that contains $X$ and $W$; let's call it $C_0$.
        It's distinct from the cluster $D_0$ that contains $Z$ and $W$, since no cluster contains both $X$ and $Z$ by assumption.
        The unique path from $C_0$ to $D_0$
        intersects with ${\cal T}(Y)$.
        But now $W \in C_0 \cap D_0$, and by the running intersection property, every node along this unique path must contain $W$ as well.
        But this contradicts our assumption that $W$ is disjoint from $Y$! So $G$ is chordal.
    \end{proof}
\end{lproof}

\subsection{Correctness and Efficiency of Inference via Exponential Conic Programming}
    \label{proofs:expcone-efficient-correct}

\recall{prop:joint-inc-correct}

\begin{lproof}
    \label{proof:joint-inc-correct}
    Suppose that $(\mu, \mat u)$ is a solution to \eqref{prob:joint-inc}.
    The exponential cone constraints ensure that, for every $(a, s,t) \in \V\!\Ar$,
    $$
        u_{a,s,t} \ge \mu(s,t) \log \frac{\mu(s,t)}{\p_a(t|s)\mu(s)}
    $$
    where $\mu(s,t)$ and $\mu(s)$, as usual, are shorthand for $\mu(\Src a{=}s, \Tgt a{=}t)$ and $\mu(\Src a {=} s)$, respectively.

    Suppose, for contradiction, that one of these inequalities is strict at some an index $(a',s',t') \in \V\!\Ar$ for which $\beta_{a'} > 0$.
    Explicitly, this means
    $$
        u_{a',s',t'} > \mu(s_0,t_0) \log \frac{\mu(s',t')}{\p_{a'}(t'|s')\mu(s')}.
    $$
    In that case, we can define a vector $\mat u' = [u'_{a,s,t}]_{(a,s,t)\in\V\!\Ar}$ which is identical to $\mat u$, except that at $(a',s',t')$, it is halfway between the two quantities described as different above.  More precisely:
    $$
        u'_{a',s',t'} = \frac12 u_{a',s',t'} + \frac12 \log \mu(s',t') \log \frac{\mu(s',t')}{\p_a(t'|s')\mu(s')}.
    $$
    Note that $u'_{a',s',t'} < u_{a',s',t'}$,
    and also that, by construction, $(\mu, \mat u')$ also satisfies the constraints of \eqref{prob:joint-inc}.
    In more detail: for $(a', s', t')$ it doesn't violate the associated exponential cone constraint, as
    $$
        \left( \text{formally:} \quad
        u'_{a',s',t'} = \frac12 u_{a',s',t'} + \frac12 \log \mu(s',t')\log \frac{\mu(s',t')}{\p_{a'}(t'|s')\mu(s')}
        >
        \mu(s',t') \log \frac{\mu(s',t')}{\p_{a'}(t'|s')\mu(s')}
        \right),
    $$
    and $\mat u'$ remains unchanged at the other indices, and so satisfies the constraints at those indices, becasuse $\mat u$ does.
    But now, because $u'_{a', s', t'} < u_{a',s',t'}$, and $\beta_{a'} >0$, we also have
    \[
        \sum_{(a,s,t) \in \V\!\Ar} \beta_a u'_{a,s,t}
            > \sum_{(a,s,t) \in \V\!\Ar} \beta_a u'_{a,s,t}.
    \]
    Thus the objective value at $(\mu, \mat u')$ is strictly
    smaller than the one at $(\mu, \mat u)$, both of which are feasible points.
    This contradicts the assumption that $(\mu, \mat u)$ is optimal.
    We therefore conclude that none of these inequalities can be strict at points where $\beta_{a} > 0$.
    This can be compactly written as:
    \begin{align*}
        \forall (a,s,t) \in \V\!\Ar.\quad
        \beta_a u_{a,s,t} &= \beta_a \mu(s,t) \log \frac{\mu(s,t)}{\p_a(t|s)\mu(s)} \\
        \implies\qquad
        \sum_{(a,s,t) \in \V\!\Ar}\beta_a u_{a,s,t}
            &= \sum_{(a,s,t) \in \V\!\Ar} \beta_a \mu(s,t) \log \frac{\mu(s,t)}{\p_a(t|s)\mu(s)}
            = \OInc_{\dg M}(\mu).
    \end{align*}
    In other words, the objective of problem \eqref{prob:joint-inc} at
    $(\mu, \mat u)$ is equal to the observational incompatibility $\OInc_{\dg M}(\mu)$ of $\mu$ with $\dg M$.
    And, because $(\mu, \mat u)$ minimizes this value among all joint distributions, $\mu$ must be a minimum of $\OInc_{\dg M}$.

    More formally: assume for contradiciton that $\mu$ is not a minimizer of $\OInc_{\dg M}$. Then there would be some other distribution $\mu'$ for which $\OInc_{\dg M}(\mu') < \OInc_{\dg M}(\mu)$.
    Let $\mat u'' := [ \mu'(s,t) \log \frac{\mu'(s,t)}{\p_a(t|s) \mu'(s)} ]_{(a,s,t) \in \V\!\Ar}$. Clearly $(\mu', \mat u'')$ satisfies the constraints of the problem, and moreover,
    \[
        \sum_{(a,s,t)\in \V\!\Ar} \beta_a u_{a,s,t} =
        \OInc_{\dg M}(\mu) >
        \OInc_{\dg M}(\mu') =
        \sum_{(a,s,t)\in \V\!\Ar} \beta_a u'_{a,s,t},
    \]
    contradicting the assumption that the $(\mu, \mat u)$ is optimal for problem \eqref{prob:joint-inc}. Thus, $\mu$ is a minimizer of $\OInc_{\dg M}$, and the objective value is $\inf_{\mu} \OInc_{\dg M}(\mu) = \aar{\dg M}_0$, as desired.
\end{lproof}

\recall{prop:joint-small-gamma-correct}
For convenience, we repeat problem \eqref{prob:joint-small-gamma}
(left) and an equivalent variant of it that we implement (right) below.

\begin{minipage}{0.49\linewidth}
\begin{align*}
\minimize_{\mu, \mat u, \mat v} & ~~
    \sum_{\mathrlap{\!\!\!(a,s,t) \in \V\!\Ar}}
        (\beta_a \!- \alpha_a \gamma) u_{a,s,t}
        \,+
        \gamma
        \sum_{\mathclap{w \in \V\!\X}} v_w
    \tag{\ref{prob:joint-small-gamma}}
\\[-0.2ex]
    &\qquad
    - \sum_{\mathrlap{\!\!\!(a,s,t) \in \smash{\V\!\Ar^+}}}
        \alpha_a \gamma \,
        \mu(\Src a{=}s,\Tgt a {=} t) \log \p_a (t|s)
\\[0.2ex]
\subjto&\quad \mu \in \Delta\V\!\X,
        \quad ( -\mat v,  \mu,  \mat 1) \in K_{\exp}^{\V\!\X},
\\[-0.4ex]
    \forall a \in \Ar.~
        &\big(-\mat u_a, \mu( \Tgt a,\Src a),\p_a(\Tgt a | \Src a)  \mu(\Src a) \big)
            \in K_{\exp}^{\V a},
\\[-0.2ex]
    \forall (a,s,t) &\in \V\!\Ar^0\!.~
    \mu(\Src a{=}\mskip2mus, \Tgt a{=}\mskip2mut) = 0;
\end{align*}
\end{minipage}
~~\vrule~~
\begin{minipage}{0.49\linewidth}
\begin{align*}
\minimize_{\mu, \mat u, \mat v} & ~~
    \sum_{\mathrlap{\!\!\!(a,s,t) \in \V\!\Ar}}
        (\beta_a \!- \alpha_a \gamma) u_{a,s,t}
        \,+
        \gamma
        \sum_{\mathclap{w \in \V\!\X}} v_w
        \tag{\ref*{prob:joint-small-gamma}b}\label{prob:joint-small-gamma-b}
\\[-0.2ex]
    &\qquad
    - \sum_{\mathrlap{\!\!\!(a,s,t) \in \smash{\V\!\Ar^+}}}
        \beta_a \,
        \mu(\Src a{=}s,\Tgt a {=} t) \log \p_a (t|s)
\\[0.2ex]
\subjto&\quad \mu \in \Delta\V\!\X,
        \quad ( -\mat v,  \mu,  \mat 1) \in K_{\exp}^{\V\!\X},
\\[-0.4ex]
    \forall a \in \Ar.~
    &\big(-\mat u_a, \mu( \Tgt a,\Src a),
        \big[\,\mu(\Src a=s) \big]_{(s,t) \in \V a} \big)
        \in K_{\exp}^{\V a},
\\[-0.2ex]
    \forall (a,s,t) &\in \V\!\Ar^0\!.~
    \mu(\Src a{=}\mskip2mus, \Tgt a{=}\mskip2mut) = 0.
\end{align*}
\end{minipage}
\medskip

\begin{lproof}\label{proof:joint-small-gamma-correct}
    We start with the problem on the left, which is \eqref{prob:joint-small-gamma} from the main text.
    Suppose that $(\mu, \mat u, \mat v)$ is a solution to \eqref{prob:joint-small-gamma}.
    The exponential constraints ensure that
    \[
        \forall (a,s,t) \in \V\!\Ar.~
        u_{a,s,t} \ge \mu(s,t) \log \frac{\mu(t|s)}{\p_a(t|s)}
    \qquad\text{and}\qquad
        \forall w \in \V \X.~
        v_{w} \ge \mu(w) \log \mu(w).
    \]
    As in the previous proof, we claim that these must hold with equality (except possibly for $u_{a,s,t}$ at indices satisfying $\beta_a = \gamma \alpha_a$, when it doesn't matter).
    This is because otherwise one could reduce the value of a component of $u$ or $v$ while still satisfying all of the constraints, to obtain a strictly smaller objective, contradicing the assumption that $(\mu, \mat u, \mat v)$ minimizes it.

    Thus, $\mat v$ is a function of $\mu$, as is every value of $\mat u$ that affects the objective value of \eqref{prob:joint-small-gamma}, meaning that this objective value can be written as a function of $\mu$ alone:
    \begin{align*}
        &\sum_{\mathrlap{\!\!\!(a,s,t) \in \V\!\Ar}}
            (\beta_a \!- \alpha_a \gamma) u_{a,s,t}
        ~+ \gamma \sum_{\mathclap{w \in \V\!\X}} v_w
        ~- \sum_{\mathrlap{\!\!\!(a,s,t) \in \smash{\V\!\Ar^+}}}
            \alpha_a\gamma \, \mu(s,t) \log \p_a (t|s) \\
    &=
        \sum_{\mathrlap{\!\!\!(a,s,t) \in \V\!\Ar}}
            (\beta_a \!- \alpha_a \gamma) \pqty*{\mu(s,t) \log \frac{\mu(t|s)}{\p_a(t|s)}}
        ~+~ \gamma \sum_{\mathclap{w \in \V\!\X}} \mu(w) \log \mu(w)
        ~-~ \sum_{\mathrlap{\!\!\!(a,s,t) \in \smash{\V\!\Ar^+}}}
            \alpha_a\gamma \, \mu(s,t) \log \p_a (t|s) \\
    &=
        \sum_{a \in \Ar} (\beta_a \!- \alpha_a \gamma) \sum_{(s,t) \in \V a}
             \pqty*{\mu(s,t) \log \frac{\mu(t|s)}{\p_a(t|s)}}
        - \gamma \H(\mu)
        - \sum_{a \in \Ar} \alpha_a\gamma \, \sum_{(s,t) \in \V\!\Ar}
             \mu(s,t) \log \p_a (t|s) \\
     &=
         \sum_{a \in \Ar} (\beta_a \!- \alpha_a \gamma)
          \sum_{(s,t) \in \V a}
             \mu(s,t) \pqty*{\log \frac{1}{\p_a(t|s)} - \log \frac{1}{\mu(t|s)}}
         - \gamma \H(\mu)
         - \sum_{a \in \Ar} \alpha_a\gamma \, \Ex_{\mu} [ \log \p_a (\Tgt a|\Src a) ] \\
    &=
        \sum_{a \in \Ar} (\beta_a \!-\! \alpha_a \gamma)
           \Ex_{\mu}[ - \log \p_a(\Tgt a | \Src a)]
        - \sum_{a \in \Ar} (\beta_a \!-\! \alpha_a \gamma)
           \H_{\mu}(\Tgt a | \Src a)
        - \gamma \H(\mu)
        - \sum_{a \in \Ar} \alpha_a\gamma \, \Ex_{\mu} [ \log \p_a (\Tgt a|\Src a) ] \\
    &=
        \sum_{a \in \Ar} \Big( - \alpha_a\gamma - (\beta_a \!-\! \alpha_a \gamma) \Big)
           \Ex_{\mu}[ \log \p_a(\Tgt a | \Src a)]
        + \sum_{a \in \Ar} (\alpha_a \gamma \!-\! \beta_a)
           \H_{\mu}(\Tgt a | \Src a)
        - \gamma \H(\mu) \\
    &=
        -\sum_{a \in \Ar} \beta_a
           \Ex_{\mu}[ \log \p_a(\Tgt a | \Src a)]
        + \sum_{a \in \Ar} (\alpha_a \gamma \!-\! \beta_a)
           \H_{\mu}(\Tgt a | \Src a)
        - \gamma \H(\mu).
    \end{align*}
    ( In the third step, we were able to convert $\V\!\Ar^+$ to $\V\!\Ar$ because, as usual in when dealing with information-therotic quantities, we interpret $0 \log \frac{1}0$ as equal to zero, which is its limit. )

    The algebra, for the right side variant
    \eqref{prob:joint-small-gamma-b}
    is slightly simpler. In this case the middle conic constraint is almost the same, except for that $\p_a(t|s)$ has been replaced with $1$, and so it ensures that $u_{a,s,t} = \mu(s,t) \log \mu(t\mid s)$ (i.e., the same as before, but without the probability in the denomiator). So,
    \begin{align*}
        &\sum_{\mathrlap{\!\!\!(a,s,t) \in \V\!\Ar}}
            (\beta_a \!- \alpha_a \gamma) u_{a,s,t}
        ~+ \gamma \sum_{\mathclap{w \in \V\!\X}} v_w
        ~- \sum_{\mathrlap{\!\!\!(a,s,t) \in \smash{\V\!\Ar^+}}}
            \beta_a \, \mu(s,t) \log \p_a (t|s) \\
    &=
        \sum_{\mathrlap{\!\!\!(a,s,t) \in \V\!\Ar}}
            (\beta_a \!- \alpha_a \gamma) \mu(s,t) \log \mu(t|s)
        ~+~ \gamma \sum_{\mathclap{w \in \V\!\X}} \mu(w) \log \mu(w)
        ~-~ \sum_{\mathrlap{\!\!\!(a,s,t) \in \smash{\V\!\Ar^+}}}
            \beta_a \, \mu(s,t) \log \p_a (t|s) \\
    &=
        \sum_{a \in \Ar} (\beta_a \!- \alpha_a \gamma) \sum_{(s,t) \mathrlap{\in \V a}}
             \mu(s,t) \log \mu(t|s)
        - \gamma \H(\mu)
        - \sum_{a \in \Ar} \beta_a \, \sum_{(s,t) \mathrlap{\in \V\!\Ar}}
             \mu(s,t) \log \p_a (t|s) \\
        &=
        \sum_{a \in \Ar}
         (\alpha_a \gamma \!-\! \beta_a)
           \H_{\mu}(\Tgt a | \Src a)
        - \gamma \H(\mu)
        -\sum_{a \in \Ar} \beta_a
           \Ex_{\mu}[ \log \p_a(\Tgt a | \Src a)].
    \end{align*}

    In either case, the objective value is equal to $\bbr{\dg M}_\gamma(\mu)$, by \eqref{eq:altscore}.
    Because $(\mu, \mat u, \mat v)$ is optimal for this problem, we know that $\mu$ is a minimizer of $\bbr{\dg M}_{\gamma}(\mu)$, and that the objective value equals $\aar{\dg M}_\gamma$.
\end{lproof}

\begin{lemma}\label{lem:hess-relent}
    The gradient and Hessian of the conditional relative entropy
    are given by
    \begin{align*}
        \Big[ \nabla_{\mu} \kldiv{\mu(X,Y)}{\mu(X) p(Y|X) } &\Big]_u
            = \log \frac{\mu(Y\! u | X\! u)}{  p(Y\! u | X\! u)} \\
        \Big[ \nabla^2_\mu \kldiv{\mu(X,Y)}{\mu(X)p(Y|X)}&\Big]_{u,v}
            = \frac{\mathbbm1[X\!u {=} X\!v \land Y\!u {=} Y\!v]}{\mu(Y\! u, X\! u)}
            - \frac{\mathbbm1[X\!v = X\!u]}{\mu(X\!u)}
        ,
    \end{align*}
    where $X\! u = X(u)$ it the value of the variable $X$ in the joint setting $u \in \V\!\X$ of all variables.
\end{lemma}
\begin{lproof} \label{proof:hess-relent}
    \allowdisplaybreaks

    \def\pd/d#1[#2]{\,\frac{\partial}{\partial #1}\!\! \left[\vphantom{\Big|}#2\right]}

    Represent $\mu$ as a vector $[\mu_{w}]_{w \in \V\X}$.
    We will make repeated use of the following facts:
    \begin{align*}
        \pd/d\mu_u [\mu(X{=}x)]=
        \pd/d\mu_u [\mu(x)]
         &= \sum_w \pd/d\mu_u[\mu_w] \! \mathbbm1[X\!w{=}x]
            ~=~  \mathbbm1[ X\! u {=} x] ; \quad\text{and}\\
        \pd/d\mu_u [\mu(y|x)] &=
            \pd/d\mu_u [ \frac{\mu(x,y)}{\mu(x)}] \\
        &= \mu(x,y) \pd/d\mu_u[ \frac{1}{\mu(x)} ]
            + \frac1{\mu(x)} \pd/d\mu_u[ \mu(x,y) ] \\
        &= - \mu(x,y)\frac{\mathbbm1[ X\!u = x]}{\mu(x)^2}
             + \frac{1}{\mu(x)} \mathbbm1[X\!u {=} x \land Y\!u {=} y] \\
        &= \frac{\mathbbm1[X\!u = x]}{\mu(x)}\Big( \mathbbm1[Y\!u=y] - \mu(y|x) \Big).
    \end{align*}

    We now apply this to the (conditional) relative entropy:
    \begin{align*}
        &\pd/d\mu_u [ \kldiv{\mu(X,Y)}{\mu(X) p(Y|X) }] \\
            &= \pd/d\mu_u [ \sum_{w} \mu_w \log \frac{\mu(Y\! w | X\! w)}{ p(Y\! w | X\! w)} ] \\
            &= \sum_{w} \mathbbm1[u{=}w]  \log \frac{\mu(Y\! w | X\! w)}{  p(Y\! w | X\! w)}
                + \sum_{w} \mu_w  \pd/d\mu_u [  \log \frac{\mu(Y\! w | X\! w)}{  p(Y\! w | X\! w)} ] \\
            &=  \log \frac{\mu(Y\! u | X\! u)}{  p(Y\! u | X\! u)}
                + \sum_{w} \mu_w
                \frac{  p(Y\! w | X\! w)}{\mu(Y\! w | X\! w)}
                \pd/d\mu_u [ \frac{\mu(Y\! w | X\! w)}{  p(Y\! w | X\! w)} ] \\
            &=  \log \frac{\mu(Y\! u | X\! u)}{  p(Y\! u | X\! u)}
                + \sum_{w} \mu_w
                \frac{1}{\mu(Y\! w | X\! w)}
                \pd/d\mu_u [\mu(Y\! w | X\! w) ] \\
            &=  \log \frac{\mu(Y\! u | X\! u)}{  p(Y\! u | X\! u)}
                + \sum_{w} \mu_w
                \frac{1}{\mu(Y\! w | X\! w)}
                 \frac{\mathbbm1[X\!u = X\!w]}{\mu(X\!w)}\Big( \mathbbm1[Y\!u=Y\!w] - \mu(Y\!w|X\!w) \Big)\\
            &=  \log \frac{\mu(Y\! u | X\! u)}{  p(Y\! u | X\! u)}
                + \sum_{w} \mu_w  \frac{\mathbbm1[X\!u{=}X\!w \land Y\!u{=}Y\!w]}
                    {\mu(X\!w, Y\!w)}
                - \sum_{w} \mu_w \frac{\mathbbm1[X\!u = X\!w]}{\mu(X\!w)}
                \\
            &=  \log \frac{\mu(Y\! u | X\! u)}{  p(Y\! u | X\! u)}
                + \frac{1}{\mu(X\!u, Y\!u)} \sum_w \mu_w  \mathbbm1[X\!u{=}X\!w \land Y\!u{=}Y\!w]
                - \frac{1}{\mu(X\!u)} \sum_w \mu_w \mathbbm1[X\!u = X\!w]
                \\
            &=  \log \frac{\mu(Y\! u | X\! u)}{  p(Y\! u | X\! u)}
                + \frac{\mu(X\!u, Y\!u)}{\mu(X\!u, Y\!u)}
                - \frac{\mu(X\!u)}{\mu(X\!u)}   \\
            &= \log \frac{\mu(Y\! u | X\! u)}{  p(Y\! u | X\! u)} + 1 - 1 \\
            &= \log \frac{\mu(Y\! u | X\! u)}{  p(Y\! u | X\! u)}
            .
    \end{align*}

    This allows us to compute the Hessian of the conditional relative entropy, whose  components are
    \begin{align*}
        \frac{\partial^2}{\partial \mu_u \partial \mu_v} \Big[ \kldiv{\mu(XY)}{\mu(X)p(Y|X)}\Big]
        &=
        \pd/d\mu_v[ \log \frac{\mu(Y\! u | X\! u)}{  p(Y\! u | X\! u)} ] \\
        &=
        \frac{ p(Y\! u | X\! u)}{\mu(Y\! u | X\! u)} \frac1{ p(Y\! u | X\! u)}
        \pd/d\mu_v[ \mu(Y\! u | X\! u) ] \\
        &= \frac{1}{\mu(Y\! u | X\! u)}
            \frac{\mathbbm1[X\!v{=}X\!u]}{\mu(X\!u)}\Big( \mathbbm1[Y\!v{=}Y\!u] - \mu(Y\!u|X\!u) \Big)\\
        &= \frac{\mathbbm1[X\!u {=} X\!v \land Y\!u {=} Y\!v]}{\mu(Y\! u, X\! u)}
            - \frac{\mathbbm1[X\!v = X\!u]}{\mu(X\!u)}
        .
    \end{align*}

\end{lproof}

\begin{lemma}
    Let $p(Y|X)$ be a cpd,
    and suppose that $\mu_0, \mu_1 \in \Delta \V(X,Y)$ are joint distributions that have different conditional marginals on $Y$ given $X$; that is, that
    there exist $(x,y) \in \V(X,Y)$ such that
    $
        \mu_0(x,y) \mu_1(x)  \ne \mu_1(x,y) \mu_0(x).
    $
    Then the conditional relative entropy
    $
        \kldiv[\Big]{ \mu(X,Y) }{ \mu(X) p(Y|X) }
    $
    is strictly convex in $\mu$ along the line segment from $\mu_0$ to $\mu_1$.
    More precisely, for $t \in [0,1]$, if we define
    $\mu_t := (1-t) \mu_0 + t\, \mu_1$, then
    the function
    \[
    t ~\mapsto~ \kldiv[\Big]{ \mu_t(X,Y) }
        {\mu_t(X) p(Y|X)}
        \qquad\text{is strictly convex. }
    \]
    \label{lem:seg-strictcvx}
\end{lemma}
\begin{lproof}
    The function of interest can fail to be strictly convex only if the direction $\delta$ along $\mu_1-\mu_0$ is in the null-space of the Hessian matrix $\mat H(\mu)$ of the (conditional) relative entropy.
    By \cref{lem:hess-relent},
    \[
        \mat H_{(xy),(x'y')}
         = \frac{\mathbbm1[x {=} x' \land y {=} y']}{\mu(x,y)}
             - \frac{\mathbbm1[x {=} x']}{\mu(x)}.
    \]

    \def\bdelta{{\boldsymbol\delta}}
    Consider a function $\delta : \V(X,Y) \to \mathbb R$ that is not identically zero, which can be viewed as a vector $\bdelta = [\delta(x,y)]_{(x,y) \in \V(X,Y)} \in \mathbb R^{\V(X,Y)}$.
    We can also view $\delta$ as a (signed) measure on $\V(X,Y)$, that has marginals in the usual sense. In particular, we use the analogous notation
    \[
        \delta(x) :=
            \sum_{y \in \V Y} \delta(x,y).
    \]
    We then compute
    \begin{align*}
        \big(\, \mat H(\mu)\, \bdelta\, \big)_{x,y}
        &= \sum_{x', y'} \delta(x',y') \left( \frac{\mathbbm1[x {=} x' \land y {=} y']}{\mu(x,y)} - \frac{\mathbbm1[x {=} x']}{\mu(x)} \right) \\
        &= \frac{\delta(x,y)}{\mu(x,y)} - \frac{\delta(x)}{\mu(x)}.
    \end{align*}

    and also
    \begin{align*}
        \bdelta^{\sf T} \mat H(\mu) \,\bdelta
            &= \sum_{x,y} \delta(x,y) (\, \mat H(\mu)\, \bdelta\, )_{x,y} \\
            &= \sum_{x,y} \delta(x,y) \left(
                \frac{\delta(x,y)}{\mu(x,y)} - \frac{\delta(x)}{\mu(x)} \right) \\
            &= \sum_{x,y}
                \frac{\delta(x,y)^2}{\mu(x,y)} - \sum_{x} \frac{\delta(x)}{\mu(x)} \sum_y \delta(x,y) \\
            &= \sum_{x,y} \frac{\delta(x,y)^2}{\mu(x,y)} - \sum_{x} \frac{\delta(x)^2}{\mu(x)}  \\
            &= \sum_{x} \frac{\delta(x)^2}{\mu(x)} \left( \sum_y \frac{\delta(x,y)^2}{\delta(x)^2 \mu(y|x)} - 1 \right). \numberthis\label{line:beforabs}
    \end{align*}

    Now consider another discrete measure $|\delta|$, whose value at each component is the absolute value of the value of $\delta$ at that component, i.e., $|\delta|(x,y) := |\delta(x,y)|$.
    By construction, $|\delta|$ is now an unnormalized probability measure: $|\delta| = k q(X,Y)$, where $k = \sum_{x,y}|\delta(x,y)| > 0$ and $q \in \Delta\V(X,Y)$.

    Note also that $|\delta|(x)^2 = (\sum_{y} |\delta(x,y)|)^2 \ge (\sum_{y} \delta(x,y))^2$, and strictly so if there are $y,y'$ such that $\delta(x,y) < 0 < \delta(x,y')$.
    In other words, the vector $\bdelta_x = [\delta(x,y)]_{y \in \V Y}$ is either non-negative or non-positive: $\bdelta_x \ge 0$ or $\bdelta_x \le 0$ for each $x$.
     Meanwhile, $|\delta|(x,y)^2 = \delta(x,y)^2$ is unchanged.
    Thus, for every $x \in \V X$, we have:
    \begin{align*}
        \sum_y \frac{\delta(x,y)^2}{\delta(x)^2 \mu(y|x)} - 1
        &\ge \sum_y \frac{|\delta|(x,y)^2}{|\delta|(x)^2 \mu(y|x)} - 1 \\
        &= \sum_y \frac{k^2 q(x,y)^2}{k^2 q(x)^2 \mu(y|x)} - 1 \\
        &= \sum_y \frac{ q(y|x)^2}{\mu(y|x)} - 1   \\
        &= \chi^2 \Big( q(Y|x) \Big\Vert  \mu(Y|x) \Big) \ge 0.
    \end{align*}
    The final line depicts the $\chi^2$ divergence between the distributions $q(Y|x)$ and $\mu(Y|x)$, both distributions over $Y$.  Since it is a divergence, this quantity is non-negative and equals zero if and only if $q(Y|x)=\mu(Y|x)$.

    Picking up where we left off, we have:
    \begin{align*}
        \bdelta^{\sf T} \mat H(\mu) \bdelta
            &= \sum_{x} \frac{\delta(x)^2}{\mu(x)} \left( \sum_y \frac{\delta(x,y)^2}{\delta(x)^2 \mu(y|x)} - 1 \right) \\
            &\ge \sum_{x} \frac{\delta(x)^2}{\mu(x)} \left( \sum_y \frac{|\delta|(x,y)^2}{|\delta|(x)^2 \mu(y|x)} - 1 \right) \\
            &=
            \sum_{x} \frac{\delta(x)^2}{\mu(x)}
            \chi^2 \Big( q(Y|x) \Big\Vert  \mu(Y|x) \Big) \ge 0.
    \end{align*}
    As a non-negatively weighted sum of non-negative numbers, this final quantity is non-negative, and equals zero if and only if, for each $x \in \V X$, we have either $q(Y|x) = \mu(Y|x)$, or $\delta(x) = 0$.
    Furthermore, if $\bdelta^{\sf T} \mat H(\mu) \bdelta = 0$, then \emph{both} inequalities hold with equality. Therefore, we know that if $\delta(x) \ne 0$, then $\bdelta_x \ge \mat 0$ or $\bdelta_x \le \mat 0$.
    These two conditions are also sufficient to show that $\bdelta^{\sf T} \mat H(\mu) \bdelta = 0$.
    To summarize what we know so far:
    \begin{align*}
        \bdelta^{\sf T} \mat H(\mu) \bdelta = 0
            \quad\iff\quad\forall x \in \V\!X.\quad
                \text{either}&~~  (\bdelta_x \ge \mat 0 \text{ or } \bdelta_y \le \mat 0) \text{~~and~~ $|\delta|(Y|x) = \mu(Y|x) $}\\
                \text{or }&~~ \delta(x) = 0.
    \end{align*}

    The second possibility, however, is a mirage: it cannot occur.
    Let's now return to the expression we had in \eqref{line:beforabs} before considering $|\delta|$.
    We've already shown that the contribution to the sum at each value of $x$ is non-negative, so if $\bdelta^{\sf T} \mat H(\mu) \bdelta$ is equal to zero, each summand which depends on $x$ must be zero as well.
    So if $x$ is a value of $X$ for which $\delta(x) = 0$, then
    \begin{align*}
        0 = \frac{1}{\mu(x)} \left( \sum_y \frac{\delta(x,y)^2}{\mu(y|x)} - \delta(x)^2 \right)
         = \frac{1}{\mu(x)} \sum_y \frac{\delta(x,y)^2}{\mu(y|x)}
         = \sum_y \frac{\delta(x,y)^2}{\mu(x,y)},
    \end{align*}
    which is only possible if $\delta(x,y) = 0$ for all $y$.
    This allows us to compute, more simply, that
    \begin{align*}
        \bdelta^{\sf T} \mat H(\mu) \bdelta = 0
        \quad&\iff\quad
            (\forall x .~  \bdelta_x \ge \mat 0 \text{ or } \bdelta_x \le \mat 0)
            \qquad\text{and}\qquad
            \forall (x,y) \in \V(X,Y).~~
                \delta(x,y) \mu(x) = \delta(x) \mu(x,y)
        \end{align*}

    \medskip
    \hrule
    \medskip

    Finally, we are in a position to prove the lemma.
    Suppose that $\mu_0,\mu_1 \in \Delta\V(X,Y)$ and $(x^*,y^*) \in \V(X,Y)$
    are such that $\mu_0(x^*,y^*) \mu_1(x^*)  \ne \mu_1(x^*,y^*) \mu_0(x^*)$.
    So, the quantity
    \[
        \mathit{gap} := \mu_1(x^*,y^*) \mu_0(x^*) - \mu_0(x^*,y^*) \mu_1(x^*)
        \quad\text{is nonzero}.
    \]
    Then for all $t \in (0,1)$ the intermediate point $\mu_t = (1-t)\, \mu_0 + t\, \mu_1$ must have different conditional marginals from both $\mu_0$ and $\mu_1$, as
    \begin{align*}
        &\mu_t(x^*\!,y^*)\mu_0(x^*) - \mu_0(x^*\!,y^*) \mu_t(x^*) \\
            &= \Cancel{(1-t)\mu_0(x^*\!,y^*)\mu_0(x^*)} + t \mu_1(x^*\!,y^*)\mu_0(x^*)
                -  \Cancel{(1-t)\mu_0(x^*\!,y^*) \mu_0(x^*)} - t\mu_0(x^*\!,y^*) \mu_1(x^*) \\
            &= t \pqty[\big] { \mu_1(x^*\!,y^*)\mu_0(x^*) -\mu_0(x^*\!,y^*) \mu_1(x^*)}
            \\
            &=~ t \cdot \mathit{gap}
            ~\ne~ 0,
    \end{align*}
    and analogously for $\mu_1$,
    \begin{align*}
        &\mu_t(x^*\!,y^*)\mu_1(x^*) - \mu_1(x^*\!,y^*) \mu_t(x^*) \\
            &= (1-t)\mu_0(x^*\!,y^*)\mu_1(x^*) + \Cancel{ t \mu_1(x^*\!,y^*)\mu_1(x^*) }
                -  (1-t)\mu_1(x^*\!,y^*) \mu_0(x^*) - \Cancel{ t\mu_1(x^*\!,y^*) \mu_1(x^*) } \\
            &= (1-t) (\mu_0(x^*,y^*)\mu_1(x^*) -\mu_1(x^*\!,y^*) \mu_0(x^*))
            \\
            &=~ -(1-t) \cdot \mathit{gap}
            ~\ne~ 0.
    \end{align*}

    Then for any direction $\delta := k(\mu_0 - \mu_1)$ parallel to the segment between $\mu_0$ and $\mu_1$ (intuitively a tangent vector at $\mu_t$, although this fact doesn't affect the computation), of nonzero length ($k\ne 0$), we have:
    \begin{align*}
        &
        \mu_t(x^*\!,y^*) \delta(x^*)  - \delta(x^*\!,y^*) \mu_t(x^*)
        \\
        &= k\; \mu_t(x^*\!,y^*) \pqty[\big]{\mu_0(x^*) - \mu_1(x^*)}  - k\; \pqty[\big]{\mu_0(x^*\!,y^*)-\mu_1(x^*\!,y^*)} \mu_t(x^*) \\
        &= k \pqty[\Big]{\mu_t(x^*\!,y^*)\mu_0(x^*) - \mu_t(x^*\!,y^*)\mu_1(x^*)
            -\mu_0(x^*\!,y^*)\mu_t(x^*) + \mu_1(x^*\!,y^*)\mu_t(x^*)} \\
        &= k \pqty[\Big]{ \pqty[\big]{ \mu_t(x^*\!,y^*)\mu_0(x^*) -\mu_0(x^*\!,y^*)\mu_t(x^*)} + \pqty[\big]{\mu_1(x^*\!,y^*)\mu_t(x^*) - \mu_t(x^*\!,y^*)\mu_1(x^*)}} \\
        &= k \pqty[\big]{ + t \,\mathit{gap} + (1-t) \,\mathit{gap}} \\
        &= k \,\mathit{gap} \qquad \ne~ 0.
    \end{align*}

    So at every $t$, directions parallel to the segment are not in the null space of $\mat H(\mu_t)$, meaning that
    $\bdelta^{\sf T} \mat H(\mu_t) \bdelta > 0$ and so our function is strictly convex along this segment.
\end{lproof}

\recall{prop:joint+idef-correct}
\begin{lproof}\label{proof:joint+idef-correct}
    Suppose that $(-\mat u, \mu, \mat k)$ is a solution to problem \eqref{prob:joint+idef}.
    The second constraint, by \cref{prop:marginonly}, ensures that $\mu \in \bbr{\dg M}_0^*$.
    Then
    \begin{align*}
        (-\mat u, \mu, \mat k) \in K^{\V\!\X}
        \quad\implies\quad
            \forall w \in \V\!\X.~ u_w &\ge \mu(w) \log \frac{ \mu(w) } { k_w } \\
            &=  \mu(w) \log \pqty[\bigg]{ \faktor{\mu(w)~}{~\prod_{a \in \Ar}\mu(\Tgt a(w)|\Src a(w))^{\alpha_a}}}.
    \end{align*}
    The same logic as in the
    \hyperref[proof:joint-inc-correct]{proofs}
    \hyperref[proof:joint-small-gamma-correct]{of}
    \cref*{prop:joint-inc-correct,prop:joint-small-gamma-correct}
    shows that this inequality must be tight, or else
    $(-\mat u, \mu, \mat k)$ would not be optimal for \eqref{prob:joint+idef}.
    So, $\mat u$ is a function of $\mu$.  Also, by \cref{prop:idef-frozen}, the problem objective satisfies
    \[
        \mat 1^{\sf T}\mat u = \sum_{w \in \V\!X} u_w = \SInc_{\dg M}(\mu).
    \]

    Finally, because $\mu$ is optimal, it must be the unique distribution
    $\bbr{\dg M}^*$, which among those distributions that minimize $\OInc_{\dg M}$, also minimizes $\SInc_{\dg M}$, meaning $\mu = \bbr{\dg M}^*$.
\end{lproof}

\recall{prop:cluster-inc-correct}
\begin{lproof}\label{proof:cluster-inc-correct}
    The final constraints alone are enough to ensure that $\bmu$ is calibrated.
    Much like before, the exponential conic constraints tell us that
    \[
        \forall (a, s,t) \in \V\!\Ar.\quad
            u_{a,s,t} \ge \mu_{C_{\!a}}\!(s,t) \log \frac{\mu_{C_{\!a}}\!(s,t)}{\mu_{C_{\!a}}\!(s)\p_a(t|s)}
    \]
    and they hold with equality (at least at those indices where $\beta_a > 0$) because $\mat u$ is optimal.
    So
    \begin{align*}
        \sum_{(a,s,t) \in \V\!\Ar} \beta_a u_{a,s,t}
        &= \sum_{(a,s,t) \in \V\!\Ar} \beta_a \mu_{C_{\!a}}\!(s,t) \log \frac{\mu_{C_{\!a}}\!(s,t)}{\mu_{C_{\!a}}\!(s)\p_a(t|s)} \\
        &= \sum_a \beta_a \sum_{(s,t) \in \V a}\mu_{C_{\!a}}\!(s,t) \log \frac{\mu_{C_{\!a}}\!(s,t)}{\mu_{C_{\!a}}\!(s)\p_a(t|s)} \\
        &= \OInc_{\dg M}(\Pr\nolimits_{\bmu}).
    \end{align*}
    Because $\bmu$ is optimal, it is the choice of \cactree\ that minimizes this quantity.
    By \cref{coro:can-use-cliquetree}, the distribution $\bbr{\dg M}^*$ can be represented by such a \actree, and by \textcite[Prop. 3.4]{pdg-aaai},
    this distribution minimizes $\OInc_{\dg M}$.
    All this is to say that there exist \actree s of this form whose corresonding distributions attain the minimum value $\OInc_{\dg M}(\Pr_{\bmu}) = \aar{\dg M}_0$.
    So $\bmu$ must be one of them, as it minimizes $\OInc(\Pr_{\bmu})$ among such \actree s by assumption. Thus $\Pr_{\bmu} \in \bbr{\dg M}^*_0$ and the objective value of \eqref{prob:cluster-inc} equals $\aar{\dg M}_0$.
\end{lproof}

\recall{prop:cluster-small-gamma-correct}

\begin{lproof}\label{proof:cluster-small-gamma-correct}
    Suppose that $(\bmu, \mat u, \mat v)$ is a solution to \eqref{prob:cluster-small-gamma}.
    The first and fourth lines of constraints ensures that $\bmu$ is indeed a \cactree.  The second line of constraints, plays exactly the same role that it did in the previous problems, most directly in the variant \eqref{prob:cluster-inc} for $\gamma=0$. In particular, it tells says
    \[
        \forall (a, s,t) \in \V\!\Ar.\quad
            u_{a,s,t} \ge \mu_{C_{\!a}}\!(s,t) \log \frac{\mu_{C_{\!a}}\!(s,t)}{\mu_{C_{\!a}}\!(s)\p_a(t|s)}
    \]
    as before, this holds with equality (at least at those indices where $\beta_a > \alpha_a\gamma$) because $\mat u$ is optimal.
    Because $\bbeta \ge \gamma \alpha$ by assumption, either $\beta_a > \gamma \alpha_a$ or the two are equal, for every $a \in \Ar$. Either way,
    the argument used at this point in \hyperref[proof:cluster-inc-correct]{the proof of} \cref{prop:cluster-inc-correct} goes through, giving us:
    \begin{align*}
        \sum_{(a,s,t) \in \V\!\Ar} (\beta_a - \alpha_a\gamma) u_{a,s,t}
        &= \sum_{(a,s,t) \in \V\!\Ar} ((\beta_a - \alpha_a\gamma) \mu_{C_{\!a}}\!(s,t) \log \frac{\mu_{C_{\!a}}\!(s,t)}{\mu_{C_{\!a}}\!(s)\p_a(t|s)} \\
        &= \sum_a (\beta_a - \alpha_a\gamma) \sum_{(s,t) \in \V a}\mu_{C_{\!a}}\!(s,t) \log \frac{\mu_{C_{\!a}}\!(s,t)}{\mu_{C_{\!a}}\!(s)\p_a(t|s)} \\
        &= \sum_a (\beta_a - \alpha_a\gamma)~
            \kldiv[\Big]{\mu_{C_{\!a}}\!(\Src a, \Tgt a)}
                  {\mu_{C_{\!a}}\!(\Src a)\, \p_a(\Tgt a|\Src a)}
    \end{align*}
    This time, though, that's not the problem objective. In this regard, our problem \eqref{prob:cluster-small-gamma} is more closely related to \eqref{prob:cluster-small-gamma}.

    Before we get to that, we have to first bring in the final collection of exponential constraints, which show that
    \begin{align*}
        \forall C \in \C.~ \forall c \in \V(C).\quad
            v_{C,c} \ge \mu_{C}(c) \log \frac{\mu_C(c)}{ \Pash_C(c) },
    \end{align*}
    and yet again these constraints hold with equality,
    for otherwise $\mat v$ would not be optimal (since we assumed $\gamma > 0$). Therefore,
    \[
        \sum_{\mathclap{(C,c) \in \V\C}} v_{C,c}
        ~=~
        \sum_{\mathclap{(C,c) \in \V\C}}
        \mu_{C}(c) \log \frac{\mu_C(c)}{ \Pash_C(c) }
        ~=~ - \H(\Pr\nolimits_{\bmu})\quad\text{by \cref{eq:cluster-ent-decomp}}.
    \]

    The objective of our problem \eqref{prob:cluster-small-gamma} is essentially the same as that of \eqref{prob:joint-small-gamma}, so the analysis in \hyperref[proof:joint-small-gamma-correct]{the proof of} \cref{prop:joint-small-gamma-correct} applies with only a handful of superficial modifications.
    Using that proof to take a shortcut, the objective of \eqref{prob:cluster-small-gamma} must equal
    \begin{align*}
        &\sum_{\mathclap{(a,s,t) \in \V\!\Ar}}
            (\beta_a \!- \alpha_a \gamma) u_{a,s,t}
        ~+ \gamma \sum_{\mathclap{(C,c) \in \V\C}} v_{C,c}
        ~- \sum_{\mathrlap{\!\!\!(a,s,t) \in \smash{\V\!\Ar^+}}}
            \alpha_a\gamma \, \mu_{C_{\!a}}\!(s,t) \log \p_a (t|s) \\
    &=
        \sum_{\mathclap{(a,s,t) \in \V\!\Ar}}
             (\beta_a - \alpha_a\gamma) \mu_{C_{\!a}}\!(s,t) \log \frac{\mu_{C_{\!a}}\!(s,t)}{\mu_{C_{\!a}}\!(s)\p_a(t|s)}
        ~-~ \gamma \H(\Pr\nolimits_{\bmu})
        ~-~ \sum_{\mathrlap{\!\!\!(a,s,t) \in \smash{\V\!\Ar^+}}}
            \alpha_a\gamma \, \mu_{C_{\!a}}\!(s,t) \log \p_a (t|s) \\
    &=
        \sum_{a \in \Ar} \beta_a
           \Ex_{\mu_{C_{\!a}}}[ \log \p_a(\Tgt a | \Src a)]
        + \sum_{a \in \Ar} (\alpha_a \gamma \!-\! \beta_a)
           \H_{\Pr_{\bmu}}(\Tgt a | \Src a)
        - \gamma \H(\Pr_{\bmu})  \\
    &= \bbr{\dg M}_\gamma(\Pr\nolimits_{\bmu}),
    \quad.
    \end{align*}
    Finally, since $\bmu$ is such that this quantity is minimized, and because
    its unique minimizer can be represented as a cluster tree (by \cref{coro:can-use-cliquetree}), we conclude that $\bmu$ must be the cluster tree representation of it.
    Therefore, $\Pr_{\bmu}$ is the unique element of $\bbr{\dg M}^*_\gamma$, and the objective at $(\bmu, \mat u, \mat v)$ equals $\aar{\dg M}_\gamma$, as desired.
\end{lproof}

\recall{prop:cluster-idef-correct}
\begin{lproof}\label{proof:cluster-idef-correct}
    Suppose that $(\bmu, \mat u)$ is a solution to \eqref{prob:cluster+idef}.
    The exponential cone constraints state that
    \begin{align*}
        \forall C \in \C.~ \forall c \in \V(C).\quad
        u_{C,c} &\ge \mu_{C}(c) \log \frac{\mu_C(c)}{k_{C,c} \Pash_C(c) } \\
        &= \mu_{C}(c) \log \frac{\mu_C(c)}{\Pash_C(c)}
            - \mu_{C}(c) \log \prod_{a \in \Ar_C} \nu_C (\Tgt a (c) | \Src a (c))^{\alpha_a} \\
        &= \mu_{C}(c) \log \frac{\mu_C(c)}{\Pash_C(c)}
         - \mu_{C}(c) \sum_{a \in \Ar_C} \alpha_a \log \nu_C (\Tgt a (c) | \Src a (c)),
    \end{align*}
    and once again this holds with equality, as each $u_{C,c}$ is minimal with this property.
    The third line of constraints
    \[
        \forall a \in \Ar.~~\mu_{C_{\!a}}\!(\Src a, \Tgt a) \nu_{C_{\!a}}\!(\Src a) = \mu_{C_{\!a}}\!(\Src a) \nu_{C_{\!a}}\!(\Src a, \Tgt a)
    \]
    and the assumption that $\Pr_{\boldsymbol\nu} \in \bbr{\dg M}^*_0$, suffice to ensure that $\Pr_{\bmu} \in \bbr{\dg M}^*_0$ by \cref{prop:marginonly}.
    They also allow us to replace each $\nu_{C_{a}}(\Tgt a(c) | \Src a (c))$ with
    $\nu_{C_{a}}(\Tgt a(c) | \Src a(c))$, in cases where $\Src a(c) \ne 0$.
    Therefore, we calculate the objective to be:
    \begin{align*}
        \mat 1^{\sf T} \mat u &=
        \sum_{C \in \C} \sum_{c \in \V(C)} \left( \mu_{C}(c) \log \frac{\mu_C(c)}{\Pash_C(c)} -
            \mu_{C}(c) \sum_{a \in \Ar_C} \alpha_a \log \nu_C (\Tgt a (c) | \Src a (c))
            \right)\\
        &= \sum_{C \in \C} \sum_{c \in \V(C)}
                \mu_{C}(c) \log \frac{\mu_C(c)}{\Pash_C(c)}
            - \sum_{C \in \C} \sum_{c \in \V(C)}
            \mu_{C}(c) \sum_{a \in \Ar} \mathbbm1[C = C_a] \alpha_a \log \nu_C (\Tgt a (c) | \Src a (c)) \\
        &= -\H(\Pr\nolimits_{\bmu}) - \sum_{a \in \Ar} \alpha_a \sum_{C \in \C} \mathbbm1[C = C_a] \sum_{c \in \V(C)}
            \mu_{C}(c)\log \nu_C (\Tgt a (c) | \Src a (c))
            \qquad\qquad\big[\,\text{by \eqref{eq:cluster-ent-decomp}}\,\big]\\
        &= -\H(\Pr\nolimits_{\bmu}) - \sum_{a \in \Ar} \alpha_a \sum_{c \in \V(C)_a}
                \mu_{C_{\!a}}\!(c) \log  \nu_{C_{\!a}}\! (\Tgt a (c) | \Src a (c)) \\
        &= -\H(\Pr\nolimits_{\bmu}) - \sum_{a \in \Ar} \alpha_a \sum_{c \in \V(C)_a}
                \mu_{C_{\!a}}\!(c) \log { \mu_{C_{\!a}}\! (\Tgt a (c) | \Src a (c))}
            \qquad\Big[~\text{since $\mu_{C_{\!a}}\!(\Src a(c)) > 0$ whenever
                $\mu_{C_{\!a}}\!(c) > 0$}~\Big]\\
        &= -\H(\Pr\nolimits_{\bmu}) + \sum_{a \in \Ar} \alpha_a \H_{\Pr_{\bmu}}(\Tgt a | \Src a ) \\
        &= \SInc_{\dg M}(\Pr\nolimits_{\bmu}).
    \end{align*}

    To summarize: $\Pr_{\bmu}$ minimizes $\SInc_{\dg M}(\Pr\nolimits_{\bmu})$ among \cactree s with condtional marginals matching those of $\boldsymbol\nu$.
    Since we know that there is a unique distribution that minimizes $\SInc_{\dg M}$ among the elements $\bbr{\dg M}_0^*$, and also that this distribution can be represented by a \actree\ (by \cref{coro:can-use-cliquetree}), we conclude that $\bmu$ must represent this distribution. Thus, $\Pr_{\bmu} = \bbr{\dg M}^*$ as desired.
\end{lproof}

The next lemma packages the results of \textcite{dahl2022primal,nesterov1996infeasible} in a precise form that we will be able to make use of.

\begin{lemma} \label{lem:mainlemma}
    Fix integers $n_{\sf o}, n_{\sf e} \in \mathbb N$, and let $n:= 3n_{\sf e} + n_{\sf o}$.
     Suppose that
     $K = \mathbb R_{\ge 0}^{n_{\sf o}} \times K^{n_{\sf e}}_{\exp} \subset \mathbb R^n$ is a product cone, consisting of $n_{\sf o}$ copies of the non-negative orthant and $n_{\sf e}$ copies of the exponential cone.
    Let
    $\mat c \in [-1,1]^{n}$ and $ \mat b \in [-1,1]^m$
    be vectors, and $A \in [-1,1]^{m \times n}$
    be a matrix, defining
    an exponential conic program
    \begin{align*}
        &
        \minimize_{
            \mat x \in K}~~ \mat c^{\sf T} \mat x
        \quad\subjto\quad A \mat x = \mat b,
        \tag{\ref{eq:exp-conic-program}}
        .
    \end{align*}
    If this program
    is strictly feasible (i.e., if there exists $\mat x \in \mathrm{int}\, K$  such that $A \mat x = \mat b$),
    as is its dual problem
    \[
        \maximize
            \limits_{
            \mat s \in K^*,\, \mat y \in \Rext^m} ~~ \mat b^{\sf T} \mat y
        \quad\subjto\quad  A^{\sf T} \mat y  +  \mat s = \mat c,
    \]
    (i.e, if there exists $\mat s \in \mathrm{int}\,K_*$ such that $A^{\sf T} \mat y + \mat s = \mat c$),
    then both
    can be simultaneously
    solved to precision $\epsilon$
    in $O(n (m+n)^{\omega} \log\frac{n+m}{\epsilon}
    )$ time,
    where $\omega$ is the smallest exponent such that a linear system of $k$ variables and equations can be solved in $O(k^\omega)$ time.
    Furthermore, MOSEK solves this problem in $O(n (m+n)^3 \log \frac{n+m}{\epsilon})$ time.
\end{lemma}
\begin{lproof}
    For this, we begin by appealing to the algorithm and analysis of
    \textcite{badenbroek2021algorithm}, threading details through for this specific choice of cone $K$.
    To finish the proof, however, we will also need to supplement that analysis with some other well-established results of \textcite{nesterov1996infeasible} that the authors were no doubt familiar with, but did not bother referencing.

    First, we'll need some background material from convex optimization.
    A \emph{logarithmically homogeneous self-concordant barrier} with parameter $\nu$ ($\nu$-LHSCB) for a cone $K$ is a thrice differentiable strictly convex function $F: \mathrm{int}\, K \to \mathbb R$ satisfying
    $F(tx) = F(x) - \nu \log t$
    for all $t > 0$ and $x \in \mathrm{int}\, K$.
    In some sense, the point of such a barrier function is to augment the optimization objective so that we remain within the cone during the optimization process.

    For the positive orthant cone $\mathbb R_{\ge 0}$, the function
        $x \mapsto - \log x$ is a 1-LHSCB.
    We now fill in some background facts about exponential cones.
    The \emph{dual} of the exponential cone is
    \begin{align*}
        K_{\exp}^* &:= \big\{ (s_1, s_2, s_3) \in \mathbb R^3 \;:\;
            \forall (x_1, x_2, x_3) \in K_{\exp}.~~
            x_1 s_1 + x_2 s_2 + x_3 s_3 \ge 0   \big\}\\
            &= \big\{
                (s_1, s_2, s_3) \;:\; - s_1 \log (- s_1 / s_3) + s_1 - s_2 \le 0,
                    \, s_1 \le 0,\,  s_3 \ge 0
            \big\}.
    \end{align*}
    Consider points $x = (x_1, x_2, x_3) \in K_{\exp}$.
    The function
    \begin{equation}
        F_{\exp}(x) := - \log \Big(x_2 \log\frac{x_1}{x_2} - x_3\Big) - \log x_1 x_2
    \end{equation}
    is a $3$-LHSCB for $K_{\exp}$, since
    \begin{align*}
        F_{\exp}(t x) &=
            -\log \Big( t x_2 \log \frac{t x_1}{t x_2} - t x_3\Big) - \log(t^2 x_1 x_2) \\
        &= - \log \Big(t \big(\log \frac{x_1}{x_2} - x_3\big)\Big) - \log(x_1 x_2) - 2 \log t \\
        &= F_{\exp}(x) - 3 \log t
    \end{align*}
    Such barrier functions can be combined to act on product cones by summation.
    Concretely, suppose that for each $i \in \{1, \ldots, k\}$,
    we have a $\nu_i$-LHSCB $F_i: \mathrm{int}\, K_i \to \Rext$.
    For $x = (x_i)_{i=1}^k \in \prod_i K_i$, the function
    $F(x) := \sum_{i=1}^k F_i(x_i)$ is a $(\sum_i \nu_i)$-LHSCB for $\prod_i K_i$,
    since
    \[
        F(tx) = \sum_{i=1}^k F_i(t x_i)
            = \sum_{i=1}^k ( F(x_i) - \nu_i \log t)
            = F(x) - \sum_{i=1}^k \nu_i.
    \]
    In this way, our product cone $K = \mathbb R_{\ge 0}^{n_{\sf o}} \times K^{n_{\sf e}}_{\exp}$ admits a LHSCB $F$ with parameter $\nu = n_{\sf o} + 3 n_{\sf e} = n$.
    Furthermore it can be evaluated in $O(n)$ time, as can each component of
    its gradient $F'(x)$ and Hessian $F''(x) \in \mathbb R^{n \times n}$ at $x$, all of which can be expressed analytically.
    In addition, the convex conjugate of $F$
    also has a known analytic form.

    Generally speaking,
    the idea behind primal-dual interior point methods \parencite{nesterov1994book} such as the one behind MOSEK, is
    to maintain both a point $x \in K$ and a dual point $s \in K_*$ (as well as $y \in \mathbb R^m$)
    and iteratively update them, as we slowly relax the barrier and approach a point on the boundary of the cone.
    The quantity $\mu(z) := \nf{\langle s, x \rangle}{\nu} \ge 0$, called the complementarity gap, is a measure of how close the process is to converging.

    Because the initial points may not satisfy the constraints, instead
    the standard algorithms work with ``extended points'' $\bar x = (x, \tau)$ and $\bar s = (s, \kappa)$, for which the analogous complementarity gap is
    $\mu^{\sf e}(\bar x, \bar s) := (\langle x, s\rangle + \kappa\tau)/(\nu+1)$.
    Altogether, the data at each iteration may be summarized as a point $z = (y, x, \tau, s, \kappa) \in \mathbb R \times (K \times \mathbb R_{\ge 0}) \times (K_* \times \mathbb R_{\ge 0})$.
    The primary object of interest is then something called the \emph{homogenous self-dual} model.
    Originally due to \textcite{nesterov1996infeasible} and also used by others \parencite{skajaa2015homogeneous},
    it can be defined as a linear operator:
    \begin{align*}
        G
            &: \Rext^{m + 2n + 2} \to \Rext^{n + m + 1} \\
        G(y,x,\tau,s,\kappa)
            &:= \begin{bmatrix}
                0          &  A  &  -b \\
                -A^{\sf T} &  0  &  c \\
                b^{\sf T}  & -c^{\sf T} & 0
        \end{bmatrix}
        \begin{bmatrix}
            y \\ x \\ \tau
        \end{bmatrix}
        -
        \begin{bmatrix}
            0 \\ s \\k
        \end{bmatrix}.
    \end{align*}
    The reason for our interest is that if $z$ is such that $G(z) = 0$ and $\tau > 0$, then $(\nf x\tau)$ is a solution to the primal problem, and $\nf{(y,s)}{\tau}$ is a solution to the dual problem \parencite[Lemma 1]{skajaa2015homogeneous}, while if $G(z) = 0$ and $\kappa > 0$, then at least one of the two problems is infeasible.

    We now are in a better position to describe the algorithm.
    According to the MOSEK documentation \parencite{dahl2022primal},
        for the exponential cone,
            begins with an initial point
        \[
            \mat v := (1.291, 0.805, -0.828) \in (K_{\exp} \,\cap\, K_{\exp}^*)
        \]

        for this particular cone $K$,
    the algorithm begins at the initial point
    \begin{align*}
        z_0 := (y_0, x_0, \tau_0, s_0, \kappa_0)
        \qquad
                    \text{where}\quad
                x_0 &= s_0 = (\overbrace{\vphantom| 1, \ldots, 1}^{n_{\sf o} \text{ copies}},~
                    \overbrace{\vphantom|
                        \mat v, \ldots, \mat v}^{n_{\sf e}\text{ copies}})
                \in (\mathbb R_{\ge 0})^{n_{\sf o}} \times (K_{\exp} \,\cap\, K_{\exp}^*)^{n_{\sf e}}, \\
            \qquad
                y_0 &= \mat 0  \in \mathbb R^m,
            \quad
                \tau_0 = \kappa_0 = 1.
    \end{align*}

    \def\daff#1{\Delta {#1}^{\text{aff}}}
    \def\dcen#1{\Delta {#1}^{\text{cen}}}

    At each iteration, the first step is to predict a direction for which
    \textcite{badenbroek2021algorithm}
    compute a scaling matrix $W$.
    To describe it, we first need to define \emph{shadow iterates}
    \[
        \tilde x := -F'_*(s)
        \qquad \text{and} \qquad
        \tilde s := - F'(x).
    \]
    which are in a sense reflections of $s$ and $x$ across their barrier functions, and can be computed in in $O(n)$ time.
    The analogous notion of complementarity can then be defined as $\tilde \mu(z) := \nicefrac{\langle \tilde x, \tilde s \rangle}{\nu}$.
    The scaling matrix, which we do not interpret here, can then be calculated as:
    \begin{equation}
        W :=
            \mu F''(x) + \frac{s s^{\sf T}}{\nu \mu}
            - \frac{\mu \tilde s \tilde s^{\sf T}}{\nu}
            + \frac{(s- \mu\tilde s)(s-\mu\tilde s)^{\sf T}}
                    {(s-\mu\tilde s)^{\sf T} (x - \mu \tilde x) }
            - \frac{\mu [ F''(x) \tilde x - \tilde \mu \tilde s]
                [ F''(x) \tilde x - \tilde \mu \tilde s]^{\sf T}}
                { \tilde x^{\sf T} F''(x) \tilde x - \nu \tilde \mu^2}
        \label{eq:scalemat}
    \end{equation}
    Doing so requires $O(n^2)$ steps (although it may be parallelized).
    The first four terms clearly require $O(n^2)$ steps, since each one is an outer product resulting in a $n \times n$ matrix.
    The last term computes a matrix-vector product (which requires $O(n^2)$ steps), and computes an outer product with the resulting vector, which takes $O(n^2)$ steps as well.

    The next step involves finding a solution $\daff z = (\cdots)$
    to the system of equations
    \begin{subequations} \label{eqns:sys1}
    \begin{align}
        G(\daff z) &= -G(z) \\
        \tau \daff\kappa + \daff\tau &= - \tau \kappa \\
        W \daff x + \daff s &= -s.
    \end{align}
    \end{subequations}

    (\ref{eqns:sys1}a-c) describe a system of $(n + m + 1) + 1 + (n) = 2n + m + 2$ equations and equally many unknowns,
    and solved in $O((n+m)^\omega)$ steps.
    It may be possible to exploit the sparsity of $G$ to do better.

    The next step is to center that search direction so that it lies on the central path. This is done by finding a solution $\dcen z$ to
    \begin{subequations}\label{eqns:sys2}
    \begin{align}
        G(\dcen z) &= G(z) \\
        \tau \dcen \kappa + \kappa \dcen \tau &= \mu^e \\
        W \dcen x + \dcen s &= \mu^e \tilde s,
    \end{align}
    \end{subequations}
    which again can be done in $O((n+m)^3)$ steps with Gaussian elimination, or
    with a fancier solver in $O((n+m)^2.332)$ steps.
    The two updates are then applied to the current point $z$ to obtain
    \[
        z_+ = (y_+, x_+, \tau_+, s_+, \kappa_+) := z + \alpha (\daff z + \gamma \dcen z).
    \]

    Finally, a ``correction step'', which is the primary innovation of \textcite{badenbroek2021algorithm} and used in MOSEK's algorithm,
    is a third direction $\Delta z_+^{\text{cor}}$, which is found by solving the system of equations
    \begin{subequations}\label{eqns:sys3}
    \begin{align}
        G(\Delta z^{\text{cor}}) &= 0 \\
        \tau_+ \Delta \kappa^{\text{cor}}  + \kappa_+ \Delta \tau^{\text{cor}} &= 0 \\
        W_+ \Delta {x_+}^{\text{cor}} + \dcen s &= \mu^e \tilde s,
    \end{align}
    \end{subequations}
    where
    $W_+$ is defined the same way that $W$ is, except that it uses the components of $z_+$ instead of $z$.
    After adding the correction step $\Delta z_+^{\text{cor}}$ to $z$, we repeat the entire process. The full algorithm, then, is summarized as follows:

    \begin{algorithmic}
        \STATE $z \gets (y_0, x_0, \tau_0, s_0, \kappa_0)$;
        \WHILE{}
            \STATE Compute scaling matrix $W$ as in \eqref{eq:scalemat};
            \STATE Find the solution $\daff z$ to (\ref{eqns:sys1}a-c),
                and the solution $\dcen z$ to (\ref{eqns:sys2}a-c);
            \STATE $z_+ \gets z + \alpha (\daff z + \gamma \dcen z)$;
            \STATE Compute the saling matrix $W_+$;
            \STATE Find the solution $\Delta z^{\text{cor}}_+$ to (\ref{eqns:sys3}a-c);
            \STATE $z \gets z_+ + \Delta z_+^{\text{cor}}$;
        \ENDWHILE
    \end{algorithmic}

    We have verified that each iteration of this process can be done in $O((n+m)^\omega))$ time.
    Their main result \parencite[Theorem 3]{badenbroek2021algorithm}, states that for every $\epsilon \in (0,1)$,
    the algorithm results in a solution $z$ satisfying
    \[
        \mu^{\sf e}(z)
        \le \epsilon
        \qquad \text{and}\qquad
        \Vert G(z) \Vert \le \epsilon
            \Vert G(z_0) \Vert
    \]
    in $O(n \log (1/\epsilon))$ iterations,
    for a total cost of
    $O(n (m+n)^3 \log (1/\epsilon) )$ time with Gaussian elimination, or
    $O(n (m+n)^{2.332} \log (1/\epsilon) )$ time using the linear solver with best
         known asymptotatic complexity as of 2022 \cite{duan2022faster}.

    \medskip
    \hrule

    \textbf{Verifying that the solution is approximately optimal.}
    What we have at this point is not quite enough: simply because the residual quantity $G(z)$ is approximately zero (so that we have approximately solved the homogenous model), does not mean that we've approximately solved the original problem.
    Specifically, it's entirely possible a priori that the parameter $\tau$ goes to zero at the same rate as everything else, and the quantity $(x/\tau)$ does not converge to a solution to the primal problem.
    To address this issue, we must also trace the analysis of the seminal work of \textcite{nesterov1996infeasible}, who use slightly different quantities, conflicting with the notation we have been using thus far.

    Following \textcite[pg. 231]{nesterov1996infeasible}, fix an initial point $z_0$, and let \emph{shifted feasible set}
    $\mathcal F := \{ z
        \in \mathbb R \times K \times \mathbb R_{\ge 0} \times K^* \times \mathbb R_{\ge 0}
        : G(z) = G(z_0)\}$
    be the collection of all points that have the same residual as $z_0$.
    \citeauthor*{nesterov1996infeasible} also refer to a complementary gap by $\mu(z)$ and define it identically, but the meaning of this parameter is different, because the set $\mathcal F$ on which it's defined is quite distinct from (if closely related to) the iterates of \citeauthor*{badenbroek2021algorithm}'s algorithm.
    In the service of clarity,
    will call this quantity $\mu^{\sf N}(z^{\sf N})$, for
    $z^{\sf N}
    = (y^{\sf N}, x^{\sf N}, \tau^{\sf N}, s^{\sf N}, \kappa^{\sf N})
    \in \mathcal F$.

    Although we made a point of emphasizing that the two are distinct,
    the actual relationship between them is straightforward.
    Let $z = (y, x,\tau, s, \kappa)$ be the final output of \textcite{badenbroek2021algorithm}.
    In proving their main theorem, they also prove that
    $G(z) = \epsilon G(z_0)$,  and $\mu^{\sf e}=\epsilon$;
    because $G$ is linear, we know that $G(\nf{z}{\epsilon}) = G(z_0)$.
    This means that $z^{\sf N}
    := \nf z\epsilon \in \mathcal F$.
    Therefore,
    \[
        \mu^{\sf N}(z^{\sf N} )
        = \frac{1}{\nu+1} \left( \left\langle\frac{s}{\epsilon}, \frac{x}{\epsilon}\right\rangle + \frac{\tau}{\epsilon} \frac{\kappa}{\epsilon} \right)
    = \frac{1}{\epsilon^2} \mu^{\sf e}(z) = \frac{1}{\epsilon}.
    \]
    So, roughly speaking, $\mu^{\sf N}$ and $\mu^{\sf e}$ are reciprocals.
    \citeauthor*{badenbroek2021algorithm} also
    prove that, every iterate $z$ satisfies their assumption (A2): for a fixed constant $\beta$ (equal to 0.9 in their analysis), $\beta \mu^{\sf e}(z) \le \tau \kappa$.
    Consequently,
    it happens that the same inequality holds with Nesterov's notation:
    \[
        \tau^{\sf N} \kappa^{\sf N} =
        \frac{\tau}{\epsilon} \frac{\kappa}{\epsilon}
            = \frac{\tau\kappa}{\epsilon^2} \ge \frac{\beta \epsilon}{\epsilon^2}
            = \frac\beta\epsilon = \beta \mu^{\sf N}
                (z^{\sf N}).
    \]
    This witnesses that $z^{\sf N} = \frac{z}{\epsilon}$ satisfies equation (81) of \citeauthor{nesterov1996infeasible}, which allows us to apply one of their main theorems, which addresses these issues.
    Supposing that the orignal problem is solvable,
    let $(x^*, s^*)$ be any solution to the primal and dual problems,
    and define the value $\psi := 1 + \langle s_0, x^*\rangle + \langle s^*, x_0 \rangle \ge 1$,
    which depends only on the problem and the choice of initialization.
    Then Theroem 1, part 1 of \citeauthor*{nesterov1996infeasible}, allows us to conclude that
    \[
        \frac{\kappa}{\epsilon}  \le \psi
        \quad\text{and}\quad
        \frac\tau\epsilon \ge \frac{\beta}{\epsilon\psi}
    \qquad\qquad \iff\qquad\qquad
        \kappa \le \epsilon\psi
        \quad\text{and}\quad \tau \ge \frac\beta\psi.
    \]
    Finally, the original theorem guarantees that
    $\Vert G(x) \Vert \le \epsilon \Vert G(z_0)\Vert$, meaning that
    \[
        \Big\Vert A \Big(\frac x\tau\Big) - b \Big\Vert \tau
        ~+~\Big\Vert A^{\sf T} \Big(\frac y\tau\Big) - \frac{s}\tau - c \Big\Vert \tau
        ~+~\Big\Vert b^{\sf T} \left(\frac{y}{\tau}\right) - c^{\sf T} \left(\frac{x}{\tau}\right) - \frac\kappa\tau \Big\Vert \tau \le \epsilon \Vert G(z_0) \Vert.
    \]
    Since the euclidean norm is an upper bound on the deviation in any component
    ( $\Vert v \Vert := \sqrt{\sum_i v_i^2} \ge \sqrt{\max_i v_i^2} = \max_i v_i =: \Vert v\Vert_\infty$ ), this means that in light of our bound on $\tau$ above, we have
    \[
        \Big\Vert A \Big(\frac x\tau\Big) - b \Big\Vert_\infty
        ~+~\Big\Vert A^{\sf T} \Big(\frac y\tau\Big) + \frac{s}\tau - c \Big\Vert_\infty
        ~+~\Big\Vert b^{\sf T} \left(\frac{y}{\tau}\right) - c^{\sf T} \left(\frac{x}{\tau}\right) - \frac\kappa\tau \Big\Vert_\infty
        \le \epsilon  \frac{\beta \Vert G(z_0) \Vert}{\psi }.
    \]
    The first two components show that the total constraint violation (in the primal and dual problems, respectively) is at most $\nicefrac{\epsilon\beta}{\psi} \Vert G(z_0) \Vert$.
    Meanwhile, the final component shows that the duality gap
    $\mathit{gap} = b^{\sf T} (\frac{y}{\tau}) - c^{\sf T} (\frac{x}{\tau})$,
    which is positive and an upper bound on the difference between the objective at $x/\tau$ and the optimal objective value, satisfies
    \[
        \mathit{gap}  \le \mathit{gap} + \frac\kappa\tau \le \frac{\epsilon \beta \Vert G(z_0) \Vert}{\psi}.
    \]
    Thus $x/\tau$ is an $(\epsilon \Vert G(z_0)\Vert)$-approximate solution to the original exponential conic problem.
    Since also $\psi \ge 1$, we may freely drop it to get a looser bound.
    All that remains is to investigate $\Vert G(z_0) \Vert$,
     the residual norm of the initial point
        chosen by the MOSEK solver, which equals:
    \[
        \Vert G(z_0) \Vert
        = \Vert A x_0 - b \Vert
            + \Vert A^{\sf T} y_0 + s_0 - c  \Vert
            + | c^{\sf T} x - b^{\sf T} y + 1 |.
    \]
    Making use of our assumption that every component of $A$, $b$, and $c$ is at most one, we find that
    \begin{align*}
        \Vert A x_0 - b \Vert^2
            &= \sum_j (\sum_i A_{j,i} (1.3)  - b_j)^2
            \le m (1.3n+1)^2 \in O(m n^2)
                &\subset O((m+n)^3)\\
        \Vert A^{\sf T}y_0 + s_0 - c \Vert^2
            &= \sum_i (\sum_j (A_{j,i})
            \le n (m + 2)^2 \in O(n m^2)
                &\subset O((m+n)^3) \\
        | c^{\sf T} x - b^{\sf T} y + 1 | ^2
            &\le (1.3n + m + 1)^2 \in O( (n+m)^2 )
                &\subset O((n+m)^3).
    \end{align*}
    Therefore, the residual of the initial point is
        $G(z_0) \in O((n+m)^{\nf32})$.

    To obtain a solution at most $\epsilon_0$ away from the true
    solution in any coordinate, we need to select $\epsilon$ small enough
    that the final output of the algorithm $z$ satisfies
    \[
        \epsilon \Vert G(z_0) \Vert  \le \epsilon_0
        \qquad\iff\qquad
        \frac1\epsilon \ge \frac1{\epsilon_0} \Vert G(z_0) \Vert
    \]
    It therefore suffices to choose
    $\frac{1}{\epsilon} \in O(\frac1{\epsilon_0} (n+m)^{\nf32})$,
    leading to $\log \frac{1}{\epsilon} = O( \log \frac{n+m}{\epsilon_0} ) )$
    iterations.

    Thus, we arrive at our total advertised asymptotic complexity of time
    \[
        O \Big( n (n+m)^\omega \log \frac{n+m}{\epsilon_0} \Big).
    \]

    In particular, to attain machine precision, we can fix
    $\epsilon_0$ to be the smallest gap between numbers representable
    (say with 64-bit floats, leading to $\epsilon_0 = 10^{-78}$ in the worst case), and omit the dependance on $\epsilon_0$ for the price of relatively
    small constant (78, for 64-bit floats).
\end{lproof}

\discard{
\begin{prop}
    It can also be solved with the methods of \textcite{skajaa2015homogeneous} in
    time
    \[ O(\sqrt{n} (m+n)^\omega \log n )
        \quad\subset\quad O( (m+n)^{2.872} \log n )
        \]
\end{prop}
\begin{lproof}

\end{lproof}
}

Having combed through all of the details of the analysis of  \textcite{badenbroek2021algorithm} and \textcite{nesterov1996infeasible} for exponential
conic programs as we have defined them, we are ready to show that this algorithm solves the problems presented in \cref{sec:clique-tree-expcone} within polynomial time.

In the results that follow, we use the symbol $O_{\text{BP}}( \, \cdot \,)$ to describe the complexity under the \emph{bounded precision} assumption:  the numerical values of $(\balpha,\bbeta, \mathbb P)$ that describe the PDG, as well as $\gamma$, lie within a fixed range, e.g., are 64-bit floating point numbers.
Correspondingly, we use $\tilde O_{\text{BP}}(\,\cdot\,)$ to describe the complexity
under the same assumption, but hiding logarithmic factors for parameters on which the complexity also depends polynomially.

\begin{lemma}\label{lem:cluster-inc-polytime}
Problem \eqref{prob:cluster-inc} can be solved to $\epsilon$ precision in time
\[
    O\Big( (\V\!\Ar + \V\C)^{1 + \omega} ( \log \frac{|\V\!\Ar| + |\V\C|}{\epsilon}
        + \log \frac{\beta^{\max}}{\beta^{\min}} \Big)
    \quad \subset \quad
        \tilde O_{\text{BP}}
        \Big( |\V\!\Ar + \V\C|^4 \log \frac1\epsilon \Big),
\]
where $\beta^{\max} := \max_{a \in \Ar} \beta_a$ is the largest value of $\beta$,
 and $\beta^{\min}:= \min_{a \in \Ar} \{ \beta_a : \beta_a > 0\}$ is the smallest positive one.
\end{lemma}
\begin{lproof}
    Problem \eqref{prob:cluster-inc}
    can translated via the DCP framework to
    the following exponential conic program, which has:
    \begin{itemize}[label=$\blacktriangleright$]
    \item variables
        $x = (\mat u, \mat v, \mat w, \bmu) \in K_{\exp}^{\V\!\Ar} \times \mathbb R_{\ge 0}^{\V \C}$,
        where
        \begin{itemize}[label=\textbullet]
        \item $\mathbf{u, v,w} \in \Rext^{\V\!\Ar}$
            are all vectors over $\V\!\Ar$,
            that at index $\iota = (a,s,t) \in \V\!\Ar$, have
            components $u_\iota$, $v_\iota$, and $w_\iota$, respectively;
        \item
            $\bmu = [\mu_{C}(C{=}c)]_%
            {C \in \C,\,c \in\V(C)}
             \in \Rext^{\V \C}$ is a vector representation of a \actree\ over clusters $\C$;
    \end{itemize}

    \item constraints as follows:
        \begin{itemize}[label=\textbullet]
            \item
            two linear constraints for every $(a,s,t) \in \V\!\Ar$ to ensure that
            \begin{align*}
                v_{a,s,t} &= \mu_{C_{\!a}}\!(s,t)
                    \qquad\Big(~= \sum_{\bar c \in \V(C_a \setminus\{\Src a, \Tgt a\})}
                        \mu_{C_{\!a}}\!(\bar c, s, t) \Big) \\
                \qquad\text{and}\qquad
                w_{a,s,t} &= \mu_{C_{\!a}}\!(\Src a{=}s)\, \p_a(\Tgt a{=}t\mid\Src a{=}s)
                    \qquad\Big(~= \p_a(\Tgt a{=}t\mid\Src a{=}s) \sum_{\bar c \in \V(C_a \setminus\{\Src a\})}
                        \mu_{C_{\!a}}\!(\bar c, s) \Big);
            \end{align*}
            \item for every edge $(C\!{-}\!D) \in \cal T$, and every value $\omega \in \V(C \cap D)$ of the variables that clusters $C$ and $D$ have in common, a linear constraint
            \[
                \sum_{\bar c \in \V(C\setminus D)} \mu_C(\bar c, \omega)
                    =
                \sum_{\bar d \in \V(D \setminus C)} \mu_D(\bar d, \omega);
            \]
            \item and one constraint for each cluster $C \in \C$ to ensure that $\mu_{C}$ lies on the probability simplex, i.e.,
            \[
                \sum_{c \in \V(C)} \mu_C(c) = 1.
            \]
        \end{itemize}
    \end{itemize}

    Altogether this means that we have an exponential conic program in the form
    of \cref{lem:mainlemma}, with
        $n = 3|\V\!\Ar| + |\V\C|$ variables,
        and
        $m = 2 |\V\!\Ar| + |\V \mathcal T| +  |\C|$ constraints,
    where
    $\V\mathcal T = \{ (C\!{-}\!D, \omega) :  C\!{-}\!D \in \mathcal T, \omega \in \V(C\cap D)\}$.
    Since we can simply disregard variables whose value sets are singletons, we can assume $\V(C) > 1$; summing over all clusters yields $\V\C > |\C|$.
    At the same time, since $\V\mathcal T \le \V\C$,
    we have
    \[ m,n,(m + n) \in O(\V\!\Ar, + \V\C).  \]

    We now give the explicit construction of the data $(A, b,c)$ of the exponential conic program that \eqref{prob:cluster-inc} compiles to.
    The variables are indexed by tuples
    of the form $i = (\ell,a,s,t)$ for $(a,s,t) \in \V\!\Ar$ and $\ell \in \{u,v,w\}$,
    or by tuples of the form $(C,c)$, for $c \in \V(C)$ and $C \in \C$,
    while the
    constraints are indexed by tuples of the form
    $j = (\ell,a,s,t)$ for $(a,s,t) \in \V\!\Ar$ and $\ell \in \{v,w\}$,
    of the form $(C\!{-}\!D, \omega)$, for an edge $(C\!{-}\!D) \in \mathcal T$ and $\omega \in \V(C \cap D)$,
    or simply by $(C)$, the name of a cluster $C \in \C$.
    The problem data $A = [A_{j,i}],b = [b_j],c = [c_i]$ of this program are zero, except (possibly) for the
        components:
    \begin{align*}
        c_{(u,a,s,t)} &= \beta_a \\
        A_{(v,a,s,t), (C,c)} &=
        \mathbbm1[C {=} C_{a} ~\land~ \Src a(c) {=} s ~\land~ \Tgt a(c) {=} t] \\
        A_{(w,a,s,t), (C,c)} &=
           \p_a(\Tgt a{=}t\mid \Src a{=}s)
            \mathbbm1[C {=} C_a ~\land~ \Src a(c) {=} s] \\
        A_{(w,a,s,t), (w,a,s,t)} &= -1 \\
        A_{(v,a,s,t), (v,a,s,t)} &= -1 \\
        A_{(C\!{-}\!D, \omega),(C',c)} &= \mathbbm1[C{=}C'] - \mathbbm1[C'{=}D]\\
        A_{(C),(C,c)} &= 1 \\
        b_{(C)} &= 1,
    \end{align*}
    where $\mathbbm1[\varphi]$ is equal to 1 if $\varphi$ is true, and zero if $\varphi$ is false.
    We note that we can equivalently divide each $\beta_a$ by $\max_a \beta_a$ without affecting the problem,
    although this could affect the approximation accuracy by the same factor.
    Thus, we get another factor of
    \[
        \log ( \max \{1 \} \cup \{ \beta_a : a \in \Ar\} ) ~\subseteq~ O( \log (1 + \max_a \beta_a ) ).
    \]

    Finally, to find a point that is $\epsilon$-close (say, in 2-norm) to the limiting point $\mu^*$ on the central path, as opposed to simply one that for which the suboptimality gap is at most $\epsilon$, we can appeal to strong concavity of the objective function.
    (Conditional) relative entropy is 1-strongly convex, and each relative entropy term is scaled by $\beta_a$.
    Furthermore, we're only considering marginal conditional entropy, so this convexity may not hold in all directions.
    Still, if the next step direction $\delta$ is not far from the gradient (as is the case if the interior point method has nearly converged), then, in that direction, the objective will be at least ($\min_a \{ \beta_a : \beta_a > 0\}$)-strongly convex.
    Therefore, by multiplying the requested precision by an additional factor of $\min_a \{ \beta_a : \beta_a > 0\}$, we can guarantee that our point is $\epsilon$-close to $\mu^*$, and not just in complementarity gap.

    To summarize, applying \Cref{lem:mainlemma}, we find that we can solve problem \eqref{prob:cluster-inc} in time
    \begin{align*}
        O\Big( \big(|\V\!\Ar| + |\V\C|\big)^{1 + \omega} \Big( \log \frac{ |\V\!\Ar| + |\V\C| }{\epsilon} + \log \frac{\beta_{\max}}{\beta_{\min}} \Big) \Big)
        \quad
        \subset \tilde O_{BP}\Big( \big(|\V\!\Ar| + |\V\C| \big)^4 \log \frac1\epsilon\Big).
    \end{align*}
    The factor of $\log\frac{\beta_{\max}}{\beta_{\min}}$ can be treated as a constant under the bounded precision assumption.
\end{lproof}

We now quickly step through the analogous construction for problems \eqref{prob:cluster-small-gamma} and \eqref{prob:cluster+idef}, which
solve the $\zogamma$-inference problem, and $0^+$-inference, respectively.

\begin{lemma}\label{lem:smallgamma-polytime}
    Problem \eqref{prob:cluster-small-gamma} is solved to precision $\epsilon$ in time
    \[
        O\bigg( |\V\!\Ar + \V\C|^{1 + \omega} \Big( \log \frac{ |\V\!\Ar| + |\V\C|}{\epsilon} + \log(1+ \Vert\bbeta\Vert_\infty) + \log \log \frac{1}{p^{\min}} \Big) \bigg)
        \quad
        \subset
        \quad
        \tilde O_{\text{BP}}\Big( |\V\!\Ar + \V\C|^4 \log\frac1\epsilon \Big)
    \]
    where $p^{\min}$ is the smallest nonzero probability in the PDG.
\end{lemma}
\begin{lproof}
    Problem \eqref{prob:cluster-small-gamma} has
    \begin{itemize}[label=$\blacktriangleright$]
    \item variables
        $x = (\mat u, \mat y, \mat w,\,\, \mat v, \bmu, \mat z)
        \in K_{\exp}^{\V\!\Ar} \times K_{\exp}^{\V \C}$
        where
        \begin{itemize}[label=\textbullet]
        \item
        $\mathbf{u,y,w} \in \Rext^{\V\!\Ar}$
            are all vectors over $\V\!\Ar$
            that at index $\iota = (a,s,t) \in \V\!\Ar$, have
            components $u_\iota$, $v_\iota$, and $w_\iota$, respectively;
        \item
        Meanwhile,
        $\mathbf{v,\bmu,z} \in \Rext^{\V\C}$
            are all vectors over $\V\C$
            which at index $(C,c)$, have
            components $v_{C,c}$, $\mu_C(c)$, and $z_{C,c}$, respectively.
            Once again, $\bmu = [\mu_{C}(C{=}c)]_%
            {C \in \C,\,c \in\V(C)}
             \in \Rext^{\V \C}$ is intended to be a vector representation of a \actree.
    \end{itemize}

    \item constraints as follows:
        \begin{itemize}[label=\textbullet]
            \item
            two linear constraints for each $(a,s,t) \in \V\!\Ar$, to ensure that
            \[
                y_{a,s,t} = \mu_{C_{\!a}}(s,t)
                \qquad\text{and}\qquad
                w_{a,s,t} = \mu_{C_{\!a}}\!(\Src a{=}s)\, \p_a(\Tgt a{=}t\mid\Src a{=}s),
            \]
            \item for every edge $(C\!{-}\!D) \in \cal T$, and every value $\omega \in \V(C \cap D)$ of the variables that clusters $C$ and $D$ have in common, a linear constraint
            \[
                \sum_{\bar c \in \V(C\setminus D)} \mu_C(\bar c, \omega)
                    =
                \sum_{\bar d \in \V(D \setminus C)} \mu_D(\bar d, \omega)
            \]
            \item for every $(a,s,t) \in \V\!\Ar^0$, a linear constraint
            that ensures
            \[
                0 = \mu_{C_a}\!(\Src a{=}s,\Tgt a{=}t)
                \qquad\Big(~~
                   = \sum_{\bar c \in \V(C \setminus \{\Src a, \Tgt a\})} \mu_{C_a}(\bar c, s, t) \Big)
            \]

            \discard{\TODO[POSSIBLE ISSUE: doesn't this mean the problem isn't strictly feasible?]}

            \item a linear constraint for every value $c \in \V(C)$ of every cluster $C \in \C$, to ensure that
            \[
                z_{C,c} = \mu_C( \Pash_C(c) )
                    \qquad \Big(~~= \sum_{\bar c \in \V(C \setminus \Pash_C)}
                        \mu_{C}(\bar c, \Pash_C(c)) ~~\Big)
            \]
            \item and one constraint for each cluster $C \in \C$ to ensure that $\mu_{C}$ lies on the probability simplex, i.e.,
            \[
                \sum_{c \in \V(C)} \mu_C(c) = 1.
            \]
        \end{itemize}
    \end{itemize}
    So in total, there are
    $n = |3 \V \!\Ar + 3 \V \C|$ variables,
    and
    $m = 2 |\V\!\Ar| + |\V \mathcal T| + |\V\!\Ar^0| + |\V \C| + |\C|$ constraints.
    The same arguments made in \cref{lem:cluster-inc-polytime} show that both $n,m \in O(|\V \!\Ar + \V \C|)$.

    Also like before, it is easy to see that the components of $A$ and $b$ are all at most 1.  However, we will need to rescale the objective $c$ in order for each of its components to be most 1. We can do this by dividing it by
    $\max \{ - \beta_a \log {p_a(t|s)} \}_{(a,s,t) \in \V\!\Ar} \cup \{ 1 \}$.

    Finally, to ensure that we have a solution that is $\epsilon$-close to the end of the central path, as opposed to one that is merely $\epsilon$-close in compelementarity gap, we must appeal to convexity.
    As in the proof of \cref{lem:cluster-inc-polytime}, this amounts to reducing the target accuracy by a factor of the smallest possible coefficient of strong convexity, along the next step direction.
    In this case, the bound is simpler: because negative entropy is (unconditionally) 1-strongly convex, and since $\bbeta \ge \balpha \gamma$, the remaining terms are convex, this could be, at worst, $\frac1\gamma$.

    This gives rise to our result: problem \eqref{prob:cluster-small-gamma} can be solved in
    \begin{align*}
        O\left( |\V \!\Ar + \V \C|^{1+\omega}
            \left(\log \frac{|\V \!\Ar + \V \C|}{\epsilon} + \log \frac1\gamma \left(1 + \max_{(a,s,t) \in \V\!\Ar} \beta_a \log \frac{1}{\p_a(t|s)} \right) \right)  \right) \\
        \subset
        O\left( |\V \!\Ar + \V \C|^{1+ \omega}
        \left\{
            \log \frac{|\V \!\Ar + \V \C|}{\epsilon}+
            \log  \frac{\beta^{\max}}{\gamma} +
            \log\log \frac{1}{p^{\min}}
        \right\}
        \right) \\
        \subset
        \tilde O_{\text{BP}}\left( |\V \!\Ar + \V \C|^{1+ \omega}
        \Big(
            \log \frac{1}{\epsilon}
        \Big)
        \right)
    \end{align*}
    operations, where $p$ is the smallest nonzero probability in the PDG, and $\beta^{\max}$ is the largest confidence in the PDG larger than 1.
\end{lproof}

\begin{lemma}\label{lem:cluster+idef-polytime}
    Problem \eqref{prob:cluster+idef} is solved to precision $\epsilon$ in
    \[
        O\Big( |\V \C| |\V \!\Ar + \V \C|^{\omega}
            \log \frac{|\V \!\Ar + \V \C|}{\epsilon} \Big)
        \quad\subset\quad
        \tilde O_{\text{BP}}\Big( |\V \C + \V \! \Ar|^{4} \log\frac1\epsilon \Big) \text{~~time}.
    \]
\end{lemma}
\begin{lproof}
    Problem \eqref{prob:cluster+idef} is slightly more straightforward; having done
    \cref{lem:cluster-inc-polytime,lem:smallgamma-polytime} in depth, we do this one more quickly.
    In the standard form, problem \eqref{prob:cluster+idef}, has variables
    $x = (\mat u, \bmu, \mat w)
        \in K_{\exp}^{\V\C}$.
    The constraints  are:

    \begin{itemize}[label=\textbullet]
        \item
        one linear constraint for each $(C,c) \in \V\C$, to ensure that
        \[
            w_{C,c} = k_{(C,c)} \mu_C( \Pash_C(c) )
            \qquad\Big(~~= \sum_{\bar c \in \V(C\setminus \Pash_C )} \mu_C(\bar c, \Pash_C(c))
                \Big)
        \]
        \item for every edge $(C\!{-}\!D) \in \cal T$, and every value $\omega \in \V(C \cap D)$ of the variables that clusters $C$ and $D$ have in common, a linear constraint
        \[
            \sum_{\bar c \in \V(C\setminus D)} \mu_C(\bar c, \omega)
                =
            \sum_{\bar d \in \V(D \setminus C)} \mu_D(\bar d, \omega)
        \]
        \item for every $(a,s,t) \in \V\!\Ar$, a linear constraint
        that ensures
        \[
            \mu_{C_a}(\Src a{=}s, \Tgt a{=}t) \, \nu_{C_a}(\Src a{=}s)
                =
            \nu_{C_a}(\Src a{=}s, \Tgt a{=}t) \, \mu_{C_a}(\Src a{=}s).
        \]
        This is linear, because recall that $\nu$ is a constant in this optimization
        problem, found by having previously solved \eqref{prob:cluster-inc}.

        \item and one constraint for each cluster $C \in \C$ to ensure that $\mu_{C}$ lies on the probability simplex.
    \end{itemize}
    So in total, there are
    $n = 3 |\V \C|$ variables,
    and
    $m =  |\V \C| + |\V \mathcal T| + |\V\!\Ar| + |\C|$ constraints.
    Once again the components of $A$ and $b$ are all at most one, and now the components of the cost function $c = \mat 1$ are identically one.
    Furthermore, our objective is 1-strongly convex, so no additional multaplicative terms are required to convert an $\epsilon$-close solution in the sense of suboptimality, to an $\epsilon$-close solution in the sense of proximity to the true solution.

    Therefore \eqref{prob:cluster+idef} can be solved in
    \begin{align*}
        O\Big( |\V \C| |\V \!\Ar + \V \C|^{\omega}
            \log \frac{|\V \!\Ar + \V \C|}{\epsilon} \Big)
        \quad\subset\quad
        \tilde O_{\text{BP}}( |\V \C + \V \! \Ar|^{4} \log\frac1\epsilon )
    \end{align*}
    operations.
\end{lproof}

\recall{theorem:main}
\begin{lproof}\label{proof:main}
    Suppose that the PDG has $N$ variables
    (each of which can take at most $V$ distinct values),
    and $A$ hyperarcs, which together form a structure has tree-width $T$.
    Then each cluster (of which there are at most $N$)
    can have at most $T+1$ variables, and so can take at most $V^T$ values.
    Therefore, $|\V \C| \le N V^{T+1}$.
    Since each arc must be entirely contained within some cluster,
    $|\V\!\Ar| \le A V^T$.
    So, $|\V\!\Ar + \V \C| \le (N+A) V^{T+1}$.

    By \cref{lem:smallgamma-polytime},
    we know that, for $\gamma \in (0, \min_a \frac{\beta_a}{\alpha_a}]$,
    a \actree\ $\epsilon$-close (in $\ell_2$ norm) to
    the one that represents the unique distribution in the
    $\zogamma$-semantics can be found in time
    in time
    \[
        O\Big(  (N+A)^4 V^{4T+4} \log \Big(V^{T+1}(N+A)\frac{1}{\epsilon}  \frac{\beta^{\max}}{\gamma}
         + \log \frac{1}{p^{\min}} \Big)  \Big).
    \]
    Similarly, by \cref{lem:cluster-inc-polytime,lem:cluster+idef-polytime}
    a \actree\ $\epsilon$-close to the one representing the $0^+$
    semantics can be found in time
    \begin{align*}
        O\Big( |\V\C + \V\!\Ar|^{4} \log \frac{\V\!\Ar+ \V\C}{\epsilon}\Big) +
        O\Big( |\V\C + \V\!\Ar|^{4} \log \frac{\V\!\Ar+ \V\C}{\epsilon}\frac{\beta^{\max}}{\beta^{\min}}\Big)
        \\
        \subseteq
        O\Big(  (N+A)^4 V^{4(T+1)} \log \Big( V^{T+1} (N+A)
            \frac{1}{\epsilon}  \frac{\beta^{\max}}{\beta^{\min}} \Big)
            \Big).
    \end{align*}
    Either way, a \actree\ $\epsilon$-close to the one that represents the $\zogamma$-semantics, for $\gamma \in \{0^+ \} \cup (0, \min_a \frac{\beta_a}{\alpha_a})$,
    can be found in
    \[
        O \bigg(|\V\!\Ar + \V \C|^{4}  \log \Big( \frac{|\V\!\Ar + \V \C|}{\epsilon} \frac{\beta^{\max}}{\beta^{\min}} + \log \frac1{p^{\min}} \Big) \bigg)
        ~\subseteq~
        \tilde O_{\text{BP}}\Big(|\V\!\Ar + \V \C|^{4}\log \frac{1}{\epsilon} \Big)
    \]
    arithmetic operations, each of which can be done in $O(k \log k)$ time. 
    
    If $\bbeta, \mathbb P$, and $\gamma$ are all binary numbers specified in $k$
    bits, then $\log_2 \frac{\beta^{\max}}{\beta^{\min}} \le 2k$ and $\log \log \frac{1}{p^{\min}} \le \log k + \log(2)$, 
    Thus, under these assumptions, such a \actree\ 
    can be found in 
    \begin{align*}
        O \bigg(|\V\!\Ar + \V \C|^{4}  \Big( \log \frac{|\V\!\Ar + \V \C|}{\epsilon} + k + \log k \Big) k \log k \bigg) 
    ~\subseteq~
        \tilde O \bigg(|\V\!\Ar + \V \C|^{4}  \log \Big( \frac{1}{\epsilon} \Big) k  \bigg) 
    \end{align*}
    time.    
    Finally, we prove part (b).  The $\infty$-norm is smaller than the $\ell_2$ norm,
    so if $\Vert \bmu - \bmu^* \Vert_2 < 2^{-(k+1)}$, then any change to $\bmu$ 
    of size $2^{-k}$ or larger will cause it to be further from $\bmu^*$.
    Thus, selecting $\epsilon = 2^{-(k+1)}$ produces the \actree\ of $k$-bit
    numbers that is closest to $\bmu^*$.  Plugging in this value of $\epsilon$, 
    we find that finding it takes 
    $
        \tilde O (|\V\!\Ar + \V \C|^{4} k^2 )
    $ time. 
\end{lproof}

\begin{lemma} \label{lem:logeps-conditioner}
    Let $k \ge 1$ be a fixed integer, and $\Phi$, $K_0$, $K_1, \ldots, K_k$ be parameters. 
     Given a procedure that
    produces $\epsilon$-approximate unconditional probabilities in 
    $O(\Phi \cdot (K_0 +  \sum_{i=1}^k K_i \log^i \frac1\epsilon) )$ time, 
    we can approximate conditional probabilities $\Pr(B|A)$ to within $\epsilon$ in 
    $
    O(\Phi  \cdot (K_0 \log \log \frac1{\Pr(A)} + \sum_{i=1}^k K_i \log^i \frac1{\epsilon \Pr(A)} ))
    $ 
    time.
\end{lemma}
\begin{lproof}
    Let $f$ be our algorithm for approximating unconditional probabilities.
    If $A$ is an event and $\epsilon > 0$, we write $f(A ; \epsilon)$
    for the corresponding approximation to $\Pr(A)$, which by definition satisfies
    \[
        \Pr(A) - \epsilon ~\le~ f(A; \delta) ~\le~ \Pr(A) + \epsilon.
    \]
    
    Now suppose that $A$ and $B$ are both events, and
    we want to find the conditional probability 
    $\Pr(B|A)$.
    To do so, we can run the following algorithm.     
    
    \rule{4in}{0.2ex}
    \begin{algorithmic}[1]
        \STATE $\delta \gets \epsilon$;
        \LOOP
            \STATE let $a \gets f(A; \delta)$;
            \smallskip
            \IF{$a > 2 \delta$}
                \STATE let $\delta^* \gets \epsilon (a - \delta) /3$;
                \STATE let $p \gets f(A ; \delta^*)$~~and~~ $q \gets f(A \cap B ; \delta^*)$;
                \STATE \textbf{return}~~$q ~/~ (p + \delta^*)$.
            \ELSE
                \STATE $\delta \gets \delta^2$;
            \ENDIF
        \ENDLOOP

    \end{algorithmic}
    \rule{4in}{0.2ex}
    
    \textbf{Proof of correctness.}
    We claim that the final output of the algorithm is within $\epsilon$ of the true conditional probability $\Pr(B|A)$.
    In the first iteration in which $a > 2\delta$ (line 4),
    we know that $\delta \le a - \delta \le \Pr(A)$.

    By assumption,
    \begin{align*}
        \Pr(A) - \delta^* ~\le~ p ~\le~ \Pr(A) + \delta^*
        \qquad\text{and}\qquad
        \Pr(A \cap B) - \delta^* ~\le~ q ~\le~ \Pr(A \cap B) + \delta^*,
    \end{align*}
    from which it follows that
    \begin{align*}
        \frac{\Pr(A \cap B) - \delta^*}{\Pr(A) + 2\delta^*}
            ~\le~ \frac{q}{p + \delta^*}
            ~\le~ \frac{\Pr(A \cap B) + \delta^*}{\Pr(A)}
        . \numberthis\label{eq:bound1}
    \end{align*}
    We now extend the bounds on $q/(p+\delta^*)$ in both directions, 
    starting with the upper bound. 
    Because $a - \delta \le \Pr(A)$, the RHS of \eqref{eq:bound1} is at most
    \[
        \frac{\Pr(A \cap B) + \delta^*}{\Pr(A)}
        ~=~ \Pr(B|A) + \frac{\delta^*}{\Pr(A)}
        ~=~ \Pr(B|A) + \frac{\epsilon}{3}\frac{(a-\delta)}{\Pr(A)}
        ~\le~ \Pr(B|A) + \frac{\epsilon}{3} \frac{\Pr(A)}{\Pr(A)}
        ~<~ \Pr(B|A) + \epsilon.
    \]
    The analysis of the lower bound (the LHS of \eqref{eq:bound1}) is slightly more complicated, but we still find that
    {\allowdisplaybreaks
    \begin{align*}
        \frac{\Pr(A \cap B) - \delta^*}{\Pr(A) + 2 \delta^*}
        &= \Pr(B|A) - \frac{\mu^*(x,y)}{\Pr(A)} + \frac{\Pr(A \cap B) - \delta^*}{\Pr(A) + 2 \delta^*}
        \\
        &= \Pr(B|A) + \frac{\Cancel{-\Pr(A)\, \Pr(A \cap B)} - 2 \delta^* \Pr(A \cap B) + \Cancel{\Pr(A) \Pr(A \cap B)} - \delta^* \Pr(A)}{\Pr(A) (\Pr(A) + 2 \delta^*)} \\
        &= \Pr(B|A) + \frac{ - 2 \delta^* \Pr(B|A) - \delta^* }{\Pr(A) + 2 \delta^*} \\
        &= \Pr(B|A) - \delta^* \Big( \frac{2 \Pr(B|A) + 1}{ \Pr(A) + 2 \delta^*} \Big) \\
        &\ge \Pr(B|A) - \delta^* \frac{3}{\Pr(A)+ \delta^*}
            \qquad  \Big[\text{ since $\Pr(B|A) \le 1$, and thus $- 2\Pr(B|A) \ge -2$ }\Big]\\
        & \ge \Pr(B|A) - \delta^* \frac{3}{\Pr(A)}
            \qquad \Big[ \text{ as eliminating $\delta^*$ makes this more negative }\Big] \\
        &= \Pr(B|A) - \frac{\epsilon(a-\delta)}{3} \frac{3}{\Pr(A)}
            \qquad\Big[ \text{ by definition of $\delta^*$ }\Big]\\
        &\ge \Pr(B|A) - \frac{\epsilon \Pr(A)}{\Pr(A)}
            \qquad\Big[ \text{ since $-(a-\delta) \ge - \Pr(A)$ }\Big]\\
        &= \Pr(B|A) - \epsilon.
    \end{align*}}
    These two arguments extend the bounds of \eqref{eq:bound1} in both directions.  Chaining all of these inequalities together, we have shown that our procedure returns a number $\mathtt{output}$ satisfying
    \[
    \Pr(B|A) - \epsilon ~\le~ \mathtt{output} ~\le~ \Pr(B|A) + \epsilon,
    \]
    and hence calculates the desired conditional probability to within $\epsilon$.
    
    \bigskip

    \textbf{Analysis of Runtime.}
    Let $m$ denote the total number of iterations of the algorithm.
    We deal with the simple case of $m=1$ separately. 
    If $m = 1$, then already in the first iteration
    $a > 2 \delta = 2 \epsilon$, so by definition $\delta^* > \frac13 \epsilon^3$.
    Line 6 is just two calls to the procedure, and takes
    \begin{equation}
    O\left( \Phi \Big(K_0 + \sum_{i=1}^k K_i \log^i \frac{1}{\delta^*}  \Big) \! \right)
    =
    O\left( \Phi \Big(K_0 +  \sum_{i=1}^k K_i \log^i \frac{3}{\epsilon^3} \Big) \!\right)
    \subseteq
    O\left( \Phi \Big(K_0 + \sum_{i=1}^k K_i \log^i \frac{1}{\epsilon} \Big) \!\right)
    ~\text{time}.
        \label{eq:cost-case1}
    \end{equation}

    Now consider the case where $m > 1$.
    Observe that, in the final iteration, $\delta = \epsilon^{2^{m-1}}$.
    The procedure halts when $a > 2 \delta$, and the smallest possible value of $a$ that our approximation can return is $\Pr(A) -\delta$.  Thus, the procedure must halt by the time $\Pr(A) > 3 \delta = 3 \epsilon^{2^{m-1}}$.
    On the other hand, since $m-1$ iterations are not enough to ensure termination, it must be that  $\Pr(A) - \delta' \le 2\delta'$,
    where $\delta' := \epsilon^{2^{m-2}}$ is the value of $\delta$ in the penultimate iteration.
    Together, these two facts give us a relationship between $m$ and $\Pr(A)$:
    \begin{align*}
        &&3 \epsilon^{2^{m-2}} ~&\ge&~ & \Pr(A) ~&&>~ 3 \epsilon^{2^{m-1}} \\
        &\iff\qquad&
        - \log_2 3 - 2^{m-2} \log_2 \epsilon ~&\le&~ - \log_2 & \Pr(A) ~&&<~ - \log_2 3 - 2^{m-1} \log_2 \epsilon \\
        &\iff\qquad&
        2^{m-2} ~&\le&~ \Big(\log_2 \frac3{\Pr(A)} \Big) & / \log_2(\nf1\epsilon) ~&&<~ 2^{m-1}
        .
            \numberthis\label{eq:bound2}
    \end{align*}
    In particular, the first inequality tells us that the number of required iterations is at most
    \[
        m \le 2 + \log_2 \log_2 \frac3{\Pr(A)} - \log_2 \log_2 \frac1\epsilon
            \quad = 2 + \log_2 \log_\epsilon \frac{\Pr(A)}{3}
            .
    \]

    Across all iterations, the total cost of line 3 is on the order of
    {\allowdisplaybreaks\begin{align*}
        &m \Phi K -  \Phi \sum_{i=1}^k K_i \sum_{j=1}^{m}  \log^i (\epsilon^{2^{j-1}}) \\
        &= m \Phi K -  \Phi \sum_{i=1}^k K_i \log^i (\epsilon) \sum_{j=0}^{m-1} 2^{kj} \\
        &= m \Phi K +  \Phi \sum_{i=1}^k K_i \log^i \frac1\epsilon\, \, 
            \frac{2^{im} - 1}{2^i -1}
             \\
        &< \Big(\log \log \frac3{\Pr(A)} - \log \log \frac1\epsilon \Big) \Phi K_0  + \Phi \,  
            \sum_{i=1}^k K_i \Cancel{\log^i (\nf1\epsilon)} \cdot \left[ 
            4^i \Big(\log^i \frac3{\Pr(A)} \Big) ~/~ \Cancel{\log^i(\nf1\epsilon)}
            \right] / (2^i-1) \\
        &\le \Phi K_0 \log \log \frac3{\Pr(A)}   + \Phi \sum_{i=1}^k K_i \frac{4^i}{2^i-1} \log^i \frac3{\Pr(A)} 
            \\
        &\subseteq O \Big( \Phi \cdot  \Big(K_0 \log \log \frac1{\Pr(A)} + \sum_{i=1}^k K_i \log^i \frac1{\Pr(A)}\Big)\Big)
        .            
            \numberthis
            \label{eq:cost-line4}
    \end{align*}}

    Line 6 is the last part of the procedure that
    incurs a nontrivial cost.  The procedure executes it one time, in the final iteration.
    Because $a > 2 \delta$ at this point, we know that
    \[
        \delta^*
        ~=~ \frac{\epsilon}3(a-\delta)
        ~>~  \frac\epsilon3\delta
        ~=~  \frac \epsilon3 \epsilon^{2^{m-1}}
        ~=~  \frac\epsilon3 \frac99 \epsilon^{2(2^{m-2})}
        ~=~ \frac \epsilon{27} \Big( 3 \epsilon^{2^{m-2}}\Big)^2
        ~\ge~ \frac\epsilon{27} \Pr(A)^2.
    \]
    Thus line 6 requires time
    \begin{align*}
        &O\Big(\Phi \cdot \Big( K_0 + \sum_{i=1}^k K_i \log^i \frac{27}{\Pr(A)^2\epsilon} \Big)\Big) 
        \subseteq
            O\Big(\Phi K_0  + \Phi \sum_{i=1}^k K_i \Big(\log \frac{1}{\Pr(A)} + \log \frac1\epsilon \Big)^{\!i\,}\Big) 
        \numberthis\label{eq:cost-line8}.
    \end{align*}
    Summarizing, the total running time is (at most) the sum of \eqref{eq:cost-case1}, \eqref{eq:cost-line4}, and \eqref{eq:cost-line8}, 
    or explicitly,
    \begin{align*}
        O \Big( \Phi \cdot \Big(
            K \log \log \frac1 {\Pr(A)} + \sum_{i=1}^k K_i \log^i \frac1{\epsilon \Pr(A)} %
        \Big)\Big)
        .
    \end{align*}
\end{lproof}

\recall{theorem:approx-infer}
\begin{lproof}  \label{proof:approx-infer}
    \Cref{theorem:main} gives us an approximation to a calibrated
    \actree\ that represents the distribution of interest,
    and \cref{lem:logeps-conditioner} allows us to approximate conditional probabilities
    once we can approximate unconditional ones.     
    The final ingredient is tto approximate unconditional
    probabilities using an approximate \actree. 
    
    Concretely, suppose that we are
    looking to find $\mu^*(X{=}x)$, where $\mu^* \in \bbr{\dg M}_\gamma$. 
    Once we have a \cactree\ $\bmu$ that represents $\mu^*$, 
    calcluating a marginal $\mu^*(X{=}x)$ (exactly) from $\bmu$ 
    can be done with standard methods \parencite[][\S 10.3.3]{koller2009probabilistic}.
    In the worst case, it requires taking a marginal of every cluster,
    which can be done in $O(|\V\C|) \subseteq O(N V^{T+1})$ arithmetic operations. 
    
    The wrinkle is that $\bmu$ only \emph{approximately} represents $\mu^*$, 
        in the sense that there is some $\bmu^*$ that does represent $\mu^*$ such that the L2 norm of $\bmu^* - \bmu$ is small. 
    As usual, we write $\bmu_C$ for the components of $\bmu$ that
        are associated with cluster $C$.
    For each $C \in \C$, let $E_{C}$ denote the event that 
    $(X \cap C) = x|_{C}$. 
    That is, the variables of $X$ that lie in cluster $C$ take the values
    prescribed by $x$. Then    
    \begin{align*}
        \big| \Pr\nolimits_{\!\bmu}(X{=}x) - \Pr\nolimits_{\!\bmu^*}(X{=}x) \big|
        ~\le~ \sum_{C \in \C} \big|
            \Pr\nolimits_{\!\bmu^*_C}(E_C) - 
            \Pr\nolimits_{\!\bmu_C}(E_C)
        \big|
        ~\le~
        \sum_{C \in \C} \big\Vert \bmu^*_C - \bmu_C \big\Vert_1
        ~=~
        \big\Vert \bmu^* - \bmu \big\Vert_1
        ~.
    \end{align*}

    Applying the L2-L1 norm inequality to the vector $\bmu - \bmu^*$, 
    we find
    \begin{align*}
        \big\Vert \bmu - \bmu^* \big\Vert_1
        \le 
        \big\Vert \bmu - \bmu^* \big\Vert_2 \sqrt{|\V\C|}
        \le \sqrt{N V^{T+1}} \big\Vert \bmu - \bmu^* \big\Vert_2.
    \end{align*}
    Thus, to answer unconditional queries about $X$ within (absolute) precision
    $\epsilon$, it suffices to find a \actree\ within $\epsilon / \sqrt{N V^{T+1}}$
    of $\bmu^*$ by L2 norm. 
    
    From the proof of \cref{theorem:main}, we know that we can find
    such a $\bmu$ 
    in
    \begin{align*}
        O\bigg( 
            (N{+}A)^4 V^{4(T+1)}
            \log \Big( \frac{(N{+}A)^4 V^{4(T+1)}\cdot N^{\frac12} V^{\frac{T+1}{2}}}{\epsilon} \frac{\beta^{\max}}{\beta^{\min}} + \log \frac1{p^{\min}} \Big) 
        \bigg) \\
        ~\subseteq~
        \tilde O\bigg( 
            (N{+}A)^4 V^{4(T+1)}
            \Big(
            \log \frac1\epsilon + \log \frac{\beta^{\max}}{\beta^{\min}}\Big) 
        \bigg) \\
        ~\subseteq~
        \tilde O_{\text{BP}}\Big(|\V\!\Ar + \V \C|^{4}\log \frac{1}{\epsilon} \Big)
    \end{align*}
    arithmetic operations,
    which dominates the number of operations required to then find the marginal probability $\Pr_{\!\bmu}(X{=}x)$ given the \actree\ $\bmu$.
    Thus, the complexity of finding unconditional probabilities is the same. 
    The arithmetic operations need to be done to precision at most $k \in O(\log\nf1\epsilon)$, and can be done in time $O(k\log k)$. 
    Thus, unconditional inference can be done in 
    \begin{align*}
        O\bigg( 
            (N{+}A)^{4.5} V^{4.5(T+1)}
            \log \Big( \frac{(N{+}A)^4 V^{4(T+1)}\cdot N^{\frac12} V^{\frac{T+1}{2}}}{\epsilon} \frac{\beta^{\max}}{\beta^{\min}} + \log \frac1{p^{\min}} \Big) 
            \log\frac{1}{\epsilon} \log\log\frac{1}{\epsilon}
        \bigg) \\
        ~\subseteq~
        \tilde O\bigg( 
            (N{+}A)^4 V^{4(T+1)}
            \Big(
            \log \frac1\epsilon + \log \frac{\beta^{\max}}{\beta^{\min}}\Big) 
            \log \frac1\epsilon
        \bigg)
        \quad\text{time}. 
    \end{align*}
    
    \def\mustar{\mu^{\mskip-2mu*\!}}
    Now that we have characterized the cost of unconditional inference, we can apply \cref{lem:logeps-conditioner} with $\Phi := (N+A)^4 V^{4(T+1)}$,
    $k = 2$,
    $K_0 = 0$, 
    $K_1 := \log \Phi + \log \frac{\beta^{\max}}{\beta^{\min}} + \log \log \frac1{p^{\min}}$, 
    and
    $K_2 = 1$
    to find that conditional probabilities can be found in
    \begin{align*}
        \tilde O \bigg(\! (N{+}A)^4\,V^{4(T+1)} 
        \log \frac1{\epsilon\,\mustar(x)}
        \Big(
              \log \frac{\beta^{\max}}{\beta^{\min}} 
               + \log \frac1{\epsilon\, \mustar(x)} 
         \Big)\bigg) 
        \quad \text{time},
    \end{align*}
    where $\mustar(x)$ is shorthand for
    $\mu^*(X{=}x)$. 
\end{lproof}

\subsection{Hardness Results and Reductions}
    \label{proofs:hardness-results}

We now turn to \cref{theorem:consistent-NP-hard}. We begin by proving
parts (a) and (b) directly
by reduction to SAT and \#SAT, respectively.

\recall{theorem:consistent-NP-hard}
\begin{lproof} \label{proof:consistent-NP-hard}
    \textbf{(a).}
	We can directly encode SAT problems in PDGs.
	Choose any CNF formula
	$$\varphi = \bigwedge_{j \in \mathcal J} \bigvee_{i \in \mathcal I(j)} (X_{j,i})$$
	over binary variables $\mat X := \bigcup_{j,i} X_{j,i}$,
    and let $n := |\mat X|$ denote the total number of variables in $\varphi$.
    Let
	$\dg M_\varphi$ be the PDG containing every variable $X \in \mat X$ and a binary
	variable $C_j$ (taking the value 0 or 1) for each clause $j \in \mathcal J$, as well as the following edges, for each $j \in \mathcal J$:
	\begin{itemize}
		\item a hyperedge $\{X_{j,i} : i \in \mathcal I(j)\} \tto C_j$, together with a degenerate cpd
			encoding the boolean OR function (i.e., the truth of $C_j$ given $\{X_{j,i}\}$);
		\item an edge $\pdgunit \tto C_j$, together with a cpd asserting $C_j$ be equal to 1.
	\end{itemize}
	First, note that the number of nodes, edges, and non-zero entries in the cpds are polynomial in the $|\mathcal J|, |\mat X|$, and the total number of parameters in a simple matrix representation of the cpds is also polynomial if $\mathcal I$ is bounded (e.g., if $\varphi$ is a 3-CNF formula).
	A satisfying assignment $\mat x \models \varphi$ of the variables $\mat X$ can be regarded as a degenerate joint distribution $\delta_{\mat X = \mat x}$ on $\mat X$, and extends uniquely to a full joint distribution $\mu_{\mat x} \in \Delta \V(\dg M_\varphi)$ consistent with all of the edges, by
	\[ \mu_{\mat x} = \delta_{\mat x} \otimes \delta_{\{C_j = \vee_i  x_{j,i}\}} \]

 	Conversely, if $\mu$ is a joint distribution consistent with the edges above, then any point $\mat x$ in the support of $\mu(\mat X)$ must be a satisfying assignment, since the two classes of edges respectively ensure that $1 =\mu(C_j\!=\! 1 \mid \mat X \!=\! \mat x) = \bigvee_{i \in \mathcal I(j)} \mat x_{j,i}$ for all $j \in \mathcal J$, and so $\mat x \models \varphi$.

	Thus, $\SD{\dg M_\varphi} \ne \emptyset$ if and only if $\varphi$ is satisfiable, so
	an algorithm for determining if a PDG is consistent can also be adapted (in polynomial space and time) for use as a SAT solver, and so the problem of determining if a PDG consistent is NP-hard.

    \medskip\hrule\smallskip

	\textbf{(b) Hardness of exact computation.}
    We prove this by reduction to \#SAT. Again, let $\varphi$ be some CNF formula over $\mat X$, and construct
	$\dg M_\varphi$ as in \hyperref[proof:consistent-NP-hard]{the proof} of
	\Cref{theorem:consistent-NP-hard}.
	Furthemore, let $\bbr{\varphi} := \{ \mat x : \mat x \models \varphi \}$ be the set of  assingments to $\mat X$ satisfying $\varphi$, and $\#_\varphi := |\bbr{\varphi}|$ denote the number such assignments. We now claim that
	\begin{equation}\label{eqn:number-of-solns}
		\#_\varphi = \exp \left[- \frac1\gamma \aar{ \dg M_\varphi }_\gamma \right].
	\end{equation}
 	Once we do so, we will have a reduced the \#P-hard problem of
    computing
    $\#_\varphi$ to the problem of
    computing
    $\aar{\dg M}_\gamma$ (exactly).

    We now prove \eqref{eqn:number-of-solns}.
	By definition, we have
	\[ \aar{\dg M_\varphi}_\gamma = \inf_\mu \Big[ \OInc_{\dg M_\varphi}(\mu) + \gamma \SInc_{\dg M_\varphi}(\mu) \Big]. \]
	We start with a claim about first term.

	\begin{iclaim} \label{claim:separate-inc-varphi}
		$\OInc_{\dg M_\varphi}\!(\mu) =
		\begin{cases}
			0 & \text{if}~  \supp \mu \subseteq \bbr{\varphi} \times \{ \mat 1\} \\
			\infty & \text{otherwise.}
		\end{cases}$
	\end{iclaim}
	\vspace{-1em}
	\begin{lproof}
		Writing out the definition explicitly, the first can be written as
		\begin{equation}
			\OInc_{\dg M_\varphi}\!(\mu) = \sum_{j} \left[ \kldiv[\Big]{\mu(C_j)}{\delta_1} +
				\Ex_{\mat x \sim \mu(\mat X_j)} \kldiv[\Big]{\mu(C_j \mid \mat X_j = \mat x)}{\delta_{\lor_i \mat x_{j,i}}} \right], \label{eqn:explicit-INC-Mvarphi}
		\end{equation}
		where $\mat X_j = \{X_{ij} : j \in \mathcal I(j)\}$ is the set of variables that
		appear in clause $j$, and $\delta_{(-)}$ is the probability distribution placing all mass on the point indicated by its subscript.
		As a reminder, the relative entropy is given by
		\[ \kldiv[\Big]{\mu(\Omega)}{\nu(\Omega)} := \Ex_{\omega \sim \mu} \log \frac{\mu(\omega)}{\nu(\omega)},
		\quad\parbox{1.4in}{\centering and in particular, \\ if $\Omega$ is binary,}\quad
			\kldiv[\big]{\mu(\Omega)}{\delta_\omega} = \begin{cases}
				0 &  \text{if}~\mu(\omega) = 1 ; \\
				\infty & \text{otherwise}.
		\end{cases} \]
		Applying this to \eqref{eqn:explicit-INC-Mvarphi}, we find that either:
		\begin{enumerate}[itemsep=0pt]
			\item Every term of \eqref{eqn:explicit-INC-Mvarphi} is finite (and zero) so $\OInc_{\dg M_\varphi}(\mu) = 0$, which happens when $\mu(C_j = 1) = 1$ and $\mu(C_j = \vee_i~ x_{j,i}) = 1$ for all $j$.  In this case, $\mat c = \mat 1 = \{ \vee_i~x_{j,i} \}_j$ so $\mat x \models \varphi$ for every $(\mat{c,x}) \in \supp \mu$;
			\item Some term of \eqref{eqn:explicit-INC-Mvarphi} is infinite, so that $\OInc_{\dg M_\varphi}(\mu) = \infty$, which happens if some $j$, either

			\begin{enumerate}
				\item $\mu(C_j \ne 1) > 0$ --- in which case there is some $(\mat{x,c}) \in \supp \mu$ with $\mat c \ne 1$, or
				\item $\supp \mu(\mat C) = \{\mat 1\}$, but $\mu(C_j \ne \vee_i~ x_{j,i}) > 0$ --- in which case there is some $(\mat{x,1}) \in \supp \mu$ for which $1 = c_j \ne \vee_i~x_{j,i}\;$, and so $\mat x \not\models \varphi$.
			\end{enumerate}
		\end{enumerate}
		Condensing and rearranging slightly, we have shown that
		\[
			\OInc_{\dg M_\varphi}(\mu) =
			\begin{cases}
				0 & \text{if}~  \mat x \models \varphi~\text{and}~\mat c = \mat 1
				 	~\text{for all}~(\mat x, \mat c) \in \supp \mu\\
				\infty & \text{otherwise}
			\end{cases}~.
		\]
	\end{lproof}

	Because $\SInc_{}$ is bounded, it follows immediately that
 	$\aar{\dg M_\varphi}_\gamma$, is finite if and only if
	there is some distribution $\mu \in \Delta\V(\mat X,\mat C)$ for which $\OInc_{\dg M_\varphi}(\mu)$ is finite, or equivalently, by \Cref{claim:separate-inc-varphi}, iff there exists some $\mu(\mat X) \in \Delta \V(\mat X)$ for which $\supp \mu(\mat X) \subseteq \bbr{\varphi}$, which in turn is true if and only if $\varphi$ is satisfiable.

	In particular, if $\varphi$ is not satisfiable (i.e., $\#_\varphi = 0$), then $\aar{\dg M_\varphi}_\gamma = +\infty$, and
	\[
		\exp \left[ -\frac1\gamma \aar{\dg M_\varphi}_\gamma \right] =
	 		\exp [ - \infty ] = 0 = \#_\varphi,
	\]
	so in this case \eqref{eqn:number-of-solns} holds as promised. On the other hand, if $\varphi$ \emph{is} satisfiable, then, again by \Cref{claim:separate-inc-varphi}, every $\mu$ minimizing $\bbr{\dg M_\varphi}_\gamma$, (i.e., every $\mu \in \bbr{\dg M_\varphi}_\gamma^*$) must be supported entirely on $\bbr{\varphi}$ and have $\OInc_{\dg M_\varphi}\!(\mu) = 0$.  As a result, we have
	\[
		\aar{\dg M_\varphi}_\gamma =
			\inf\nolimits_{\mu \in \Delta \big[\bbr{\varphi} \times \{\mat 1\}\big]} \gamma\; \SInc_{\dg M_\varphi}(\mu) .
	\]
	A priori, by the definition of $\SInc_{\dg M_\varphi}$, we have
	\[
		\SInc_{\dg M_\varphi}(\mu) =
		 	- \H(\mu) + \sum_{j} \Big[ \alpha_{j,1} \H_\mu(C_j \mid \mat X_j)
						+ \alpha_{j,0} \H_\mu(C_j) \Big],
	\]
	where $\alpha_{j,0}$ and $\alpha_{j,1}$ are values of $\alpha$ for the edges of $\dg M_\varphi$, which we have not specified because they are rendered irrelevant by the fact that their corresponding cpds are deterministic. We now show how this plays out in the present case.
	Any $\mu \in \Delta\big[\bbr{\varphi} \times \{\mat 1\}\big]$ we consider has a degenerate marginal on $\mat C$. Specifcally, for every $j$, we have $\mu(C_j) = \delta_1$, and since entropy is non-negative and never increased by conditioning,
	$$
		0 \le \H_\mu(C_j \mid \mat X_j) \le \H_\mu(C_j) = 0.
	$$
	Therefore, $\SInc_{\dg M_\varphi}(\mu)$ reduces to the negative entropy of $\mu$.
	Finally, making use of the fact that the maximum entropy distribution $\mu^*$ supported on a finite set $S$ is the uniform distribution on $S$, and has $\H(\mu^*) = \log | S |$, we have
	\begin{align*}
		\aar{\dg M_\varphi}_\gamma &= \inf\nolimits_{\mu \in \Delta \big(\bbr{\varphi} \times \{\mat 1\}\big)} \gamma\; \SInc_{\dg M_\varphi}(\mu) \\
			&= \inf\nolimits_{\mu \in \Delta \big(\bbr{\varphi} \times \{\mat 1\}\big)} -\, \gamma\, \H(\mu) \\
			&= - \gamma\, \sup\nolimits_{\mu \in \Delta \big(\bbr{\varphi} \times \{\mat 1\}\big)}  \H(\mu) \\
			&= - \gamma\, \log (\#_\varphi),
	\end{align*}
	\hspace{1in}giving us
	$$
		\#_\varphi = \exp \left[- \frac1\gamma \aar{ \dg M_\varphi }_\gamma \right],
	$$
	as desired. We have now reduced \#SAT to computing $\aar{\dg M}_\gamma$, for $\gamma > 0$ and an arbitrary PDG $\dg M$, which is therefore \#P-hard.

    To show the same for $\gamma = 0$, it suffices to add an additional hyperedge pointing to all variables, and associate it with a joint uniform distribution, and confidence 1, resulting in a new PDG $\dg M_\varphi'$.
    Because this new edge's contribution to $\OInc_{\dg M}$
    equals $\kldiv{\mu}{\mathsf{Unif}(\X)} = \log |\V\!\X| - \H(\mu)$,
    we have
    \[
        \bbr{\dg M_\varphi'}_0(\mu)
            = \OInc_{\dg M_\varphi'}(\mu)
            = \bbr{\dg M_\varphi}(\mu) + \log | \V\!\X | - \H(\mu)
            = \bbr{\dg M_\varphi}_{1}(\mu)
             + \log |\V\!\X |.
    \]
    Since this is true for all $\mu$,
    we can take the of this equation over $\mu$,
    and so conclude that
    \begin{align*}
        \aar{\dg M_\varphi'}_0 = \aar{\dg M_\varphi}_1 + \log |\V\!\X|
            = \log \big( |\V\!\X| / \#_\varphi \big) \\
            \implies \qquad \#_\varphi = |\V\!\X| \exp( -\aar{\dg M'_\varphi}_0 )
    \end{align*}
    Thus, the number of satisfying assignments can be found through
    via an oracle for $\aar{-}_0$, as well.  This shows that calculating
    this purely observational inconsistency is \#P-hard as well.

    \medskip\hrule\smallskip

    \textbf{(c) Hardness of approximation.}
    To calculate $\#_\varphi$ exactly, it turns out that 
    we do not need to know $\aar{\dg M_\varphi}_\gamma$ exactly.
    Instead, we claim it suffices to approximate it to within
    $\epsilon <  \gamma \log(1 + 2^{-(n+1)})$.

    Suppose that $|r - \aar{\dg M_\varphi}_\gamma| < \epsilon$.
    Then
    \begin{align*}
    \exp \Big( -\frac r \gamma \Big) & \in
    \exp \Big[ - \frac1\gamma \Big(\aar[\big]{ \dg M_\varphi }_\gamma \pm \epsilon\Big) \Big]
    \\
    &= \exp \Big[ - \frac1\gamma \aar[\big]{ \dg M_\varphi }_\gamma \Big] \cdot \exp(\pm \nf\epsilon\gamma)
    \\
    &= \#_\varphi \cdot \exp ( \pm \nf\epsilon\gamma )
    \\
    &= [\#_\varphi \exp ( - \nf\epsilon\gamma ),~
        \#_\varphi \exp ( + \nf\epsilon\gamma )].
    \end{align*}

    Since $\#_\varphi$ is a natural number and at most $2^n$,
    If we can get a relative approximation of it to
    within a factor of $2^{-(n+1)}$, then rounding that approximate value
    to the nearest whole number gives the exact value of $\#_\varphi$.
    Thus, it suffices to choose $\epsilon$ small enough that
    \[
        \exp(-\nf \epsilon\gamma) > 1 - 2^{-(n+1)}
        \qquad\text{and}\qquad
        \exp(+\nf \epsilon\gamma) < 1 + 2^{-(n+1)};
    \]
    this is satisfied any choice of $\epsilon < \gamma \log(1 + 2^{-(n+1)})$.
    Thus, being able to approximate $\aar{\dg M_\varphi}_\gamma$ sufficiently closely
    will tell us whether or not $\varphi \in $\textsf{SAT}.
    Note that for large $n$, the maximum value of $\epsilon$ for which this is true
    is on the order of $\epsilon_{\max} \in \Theta( \gamma 2^{-n})$.
    It follows that $\log(\nf1{\epsilon_{\max}}) \in O(n)$, and so
    values of $\epsilon$ small enough to determine the satisfiability of
    a formula $\varphi$ with $n$ variables can be specified in time $O(n)$.
    Thus, the problem \ApproxPDGInc\ is \#P hard (in the size of its input).
\end{lproof}

\paragraph{Inference via Inconsistency Minimization.}
We now address
\cref{theorem:inf-via-inc-oracle},
which is closely related to \textcite{pdg-aaai}'s original idea for an inference
algorithm.  While that idea does not yield an efficient inference algorithm,
it does yield an efficient reduction from inconsistency minimization to inference.
In order to prove this, we first need another construction with PDGs.
A probability over a (set of) variables can be viewed as a vector whose
    elements sum to one.
It turns out that it is possible to use the machinery of PDGs
    to, effectively, give only one value of such a probability vector.
That is, for any $p \in [0,1]$, we can construct a PDG
    that represents the belief that $\Pr(Y{=}y) = p$, but say nothing about
    how the probability splits between other values of $y$.
We now describe that construction.

\def\Yy{Y{=}y}
\def\Yyshort{{Y_y}}

We first introduce an auxiliary binary variable
$\Yyshort$, with $\V(\Yyshort) = \{y, \lnot y\}$,
    and takes the value $y$ if $Y=y$, and $\lnot y$ if $Y \ne y$.
Note that this variable is a function of the value of variable $Y$
(although we will need to enforce this with an additional arc), and
therefore there is a unique way to extend a distribution over
variables including $Y$ to also include the variable $\Yyshort$.

With this definition, there is now an obvious way to add a hyperarc with no source and
    target $\Yyshort$, together with a asserting that $\Pr(Y{=}y)=p$.
This cpd is written as a vector $\hat p$ on the right of the figure below.
The PDG we have just constructed is illustrated on the left of the figure below.
In addition to $\hat p$ and the new variable, this PDG
    includes the structural constraint $s$ needed to define the variables
    $\Yyshort$ in terms of $Y$; it is a deterministic function,
    drawn with a double-headed gray arrow.

\begin{center}
    \begin{tikzpicture}[center base]
        \node[tpt={y1|$y$}] at (0, 1.5){};
        \node[tpt={y2|$\lnot y$},right=0.5 of y1]{};
        \node[%
            Dom={$\Yyshort$[label distance=-1.4em, xshift=1.1em] (Yy)
            around {\lab{y1}\lab{y2}(0,1.7)}} ] {};

        \node[dpadded,below=1.0 of Yy](Y){$Y$};
        \draw[black!35!proofmatt, arr2, ->>] (Y) to node[below right]{$s$} (Yy);

        \draw[arr2, <-] (Yy) to node[above]{$\hat p$} +(-2, 0);
    \end{tikzpicture}
    \hspace{1cm}
    \begin{minipage}{0.3\textwidth}
        \begin{align*}
            s(\Yyshort|Y) &:= \begin{cases}
                y & \text{if $Y=y$} \\
                \lnot y & \text{if $Y\ne y$} \\
            \end{cases}\\[1ex]
            \hat p(\Yyshort) &:= ~~\begin{matrix}
                  \begin{matrix} y & \lnot y \end{matrix} \\
                    \begin{bmatrix}
                        \;p & 1-p \;
                    \end{bmatrix}
            \end{matrix}
        \end{align*}
    \end{minipage}
\end{center}
    \medskip

So, when we add $\Pr(Y=y) = p$ to a PDG $\dg M$, what we really mean is:
first convert construct a widget as above, and add that structure (i.e., the new variable $\Yyshort$, its definition $s$, and the cpd $\hat p$) to $\dg M$.
In what sense does this ``work''?
The first order of business
is to prove that it behaves as we should expect,
semantically, in the case we're interested in.

\begin{lemma}\label{lem:inc-inc-eq}
    Suppose that $\dg M$ is a PDG with variables $\X$ and $\bbeta \ge \mat 0$.
    Then, for all $Y \subseteq \X$, $y \in \V Y$, $p \in [0,1]$ and $\gamma \ge 0$,
    we have that:
    \[
        \aar[\Big]{\dg M + ~\Pr(Y{=}y) = p}_\gamma \ge \aar{\dg M}_\gamma,
    \]
    with equality if and only if there exists $\mu \in \bbr{\dg M}^*_\gamma$
    such that $\mu(Y{=}y) = p$.
\end{lemma}
\begin{lproof}
    The inequality is immediate; it is an instance of monotonicity of inconsistency
    \cite[Lemma 1]{one-true-loss}. Intuitively: believing more cannot make you any less
    inconsistent.  We now prove that equality holds iff there is a minimizer with the appropriate conditional probability.

    $(\impliedby)$. Suppose that there is some $\mu \in \bbr{\dg M}^*_\gamma$ with $\mu(Y{=}y) = p$.
    Because $\mu \in \bbr{\dg M}^*_\gamma$, we know that
    $\bbr{\dg M}_\gamma(\mu) = \aar{\dg M}$.
    Let $\hat \mu$ be the extension of $\mu$ to the new variable $``\Yyshort$'',
        whose value is a function of $Y$ according to $s$. Then
    \begin{align*}
        \aar[\Big]{\dg M + ~\Pr(Y{=}y) = p}_{\!\gamma}
            &\le \bbr[\Big]{\dg M + ~\Pr(Y{=}y) = p}_\gamma(\hat \mu) \\
            &= \bbr{\dg M}_\gamma(\mu) + \Ex_{\mu}\left[
                \log \frac{\hat\mu(\Yyshort)}{ \hat p(\Yyshort)} \right]\\
            &= \bbr{\dg M}_\gamma(\mu) +
                \mu(Y{=}y) \log \frac{\mu(Y{=}y)}{p}
                + \mu(Y{\ne} y) \log \frac{\mu(Y{\ne} y)}{1-p} \\
            &= \bbr{\dg M}_\gamma(\mu) +
                \mu(Y{=}y) \log(1)
                + \mu(Y{\ne} y) \log(1) \\[1ex]
            &= \bbr{\dg M}_\gamma(\mu)
            = \aar{\dg M}_\gamma.
    \end{align*}

    To complete this direction of the proof, it suffices to observe
    that we already knew the inequality held in the opposite direction
    (by monotonicity), so the two terms are equal.

    $(\implies)$.  Suppose that the two inconsistencies are equal, i.e.,
    $
    \aar[\Big]{\dg M + ~\Pr(Y{=}y) = p}_\gamma = \aar{\dg M}_\gamma.
    $

    This time, choose $\hat\mu \in \bbr{\dg M+ ~\Pr(Y{=}y) = p}^*_\gamma$,
        and define $\mu$ to be its marginal on the variables of $\dg M$
        (which contains the same information as $\hat \mu$ itself).
    Let $q := \mu(Y{=}y)$. Then
    \begin{align*}
        \aar{\dg M}_\gamma &= \aar[\Big]{\dg M + ~\Pr(Y{=}y) = p}_{\!\gamma} \\
         &= \bbr[\Big]{\dg M + ~\Pr(Y{=}y) = p}_\gamma(\hat \mu) \\
         &= \bbr{\dg M}_\gamma(\mu) +
             \mu(Y{=}y) \log \frac{\mu(Y{=}y)}{p}
             + \mu(Y{\ne} y) \log \frac{\mu(Y{\ne} y)}{1-p} \\
        &= \bbr{\dg M}_\gamma(\mu) +
            \left[ q \log \frac qp + (1-q) \log \frac{1-q}{1-p} \right] \\
        &= \bbr{\dg M}_\gamma(\mu) +  \kldiv qp \\
        &\ge \aar{\dg M}_\gamma + \kldiv qp
    \end{align*}
    Therefore $0 \ge  \kldiv qp$. But relative entropy is non-negative 
    (Gibbs inequality; see any introductory text on information theory, such as
     \textcite{mackay2003information}), 
    so we actually know that $\kldiv qp=0$, and thus
    $p = \mu(Y{=}y)$.
    In addition, the algebra above shows that $\mu \in \bbr{\dg M}^*_\gamma$, as its
        score is $\aar{\dg M}_\gamma$.
    Thus, we have found $\mu \in \bbr{\dg M}^*_\gamma$ such that $\mu(Y{=}y) = p$, completing the proof.
\end{lproof}

We next show that the overall inconsistency is strictly convex in the parameter $p \in [0,1]$.
It is (notationally) simpler to state (and equally easy to prove) this result in the general case.

\begin{linked}{lemma}{inc-strictly-cvx}
    Fix $Y \subseteq \X$,
    and  $\gamma \in (0, \min_a \frac{\beta_a}{\alpha_a})$.
    As $h = h(Y)$ ranges over $\Delta \V Y$,
    the function $h \mapsto \aar{\dg M + h}_\gamma$
    is strictly convex.
\end{linked}
\begin{lproof}\label{proof:inc-strictly-cvx}
	We start by expanding the definitions. If $h$ is a cpd on $Y$ given $X$, then
	\begin{align*}
		\aar{\dg M + h}_\gamma &= \inf_\mu ~\bbr{\dg M + h}_\gamma(\mu) \\
			&= \inf_\mu \left[ \bbr{\dg M }_\gamma(\mu)
				+  \kldiv[\Big]{\mu(Y)}{h(Y} \right].
	\end{align*}
	Fix $\gamma \le \min_a \frac{\beta_a}{\alpha_a}$. Then we know that $\bbr{\dg M}_\gamma(\mu)$ is a $\gamma$-strongly convex
    (so, in particular, strictly convex)
    function of $\mu$,
    and hence there is a unique joint distribution which minimizes it. We now  show that the overall inconsistency is strictly convex in $h$.

	Suppose that $h_1(Y)$ and $h_2(Y)$ are two distributions over $Y$.
    Let $\mu_1, \mu_2$ and $\mu_\lambda$ be the joint distributions that minimze $\bbr{\dg M + h_1}_\gamma$ and $\bbr{\dg M + h_2}_\gamma$, respectively.
	For every $\lambda \in [0,1]$, define $h_\lambda := (1-\lambda) h_1 + \lambda h_2$,
    $\mu_\lambda := (1-\lambda) \mu_1 + \lambda \mu_2$, and
    $\mu_\lambda^*$ to be a minimizer of $\bbr{\dg M + h_\lambda}_\gamma$.
    The following is a simple consequence of these definitions:
	\begin{align*}
		\aar{\dg M + h_\lambda}_\gamma
            &= \bbr{\dg M + h_\lambda}_\gamma( \mu_\lambda^* ) \\
            &\le \bbr{\dg M + h_\lambda}_\gamma( \mu_\lambda ) \\
            &= \bbr{\dg M}_\gamma(\mu_\lambda) + \kldiv[\Big]{\mu_\lambda(Y)}{h_\lambda(Y)}
            .
	\end{align*}
	By the convexity of $\bbr{\dg M}_\gamma$ and $\thickD$, we have
	\begin{align}
		\bbr{\dg M}_\gamma(\mu_\lambda)
		 	&\le
            (1-\lambda)
            \bbr{\dg M}_\gamma(\mu_1) + \lambda \bbr{\dg M}_\gamma(\mu_2)
			 	\label{eqn:score-cvx}\\
		\text{and}\qquad \kldiv[\Big]{\mu_\lambda(Y)}{h_\lambda(Y) }
			&\le (1-\lambda)
            \kldiv[\Big]{\mu_1(Y)}{h_1(Y)} \nonumber \\
			&\qquad+ \lambda\;\;\kldiv[\Big]{\mu_2(Y)}{h_2(Y) }.
				\label{eqn:D-cvx}
	\end{align}
	If $\mu_1 \ne \mu_2$ then since $\bbr{\dg M}$ is strictly convex, \eqref{eqn:score-cvx} must
	be a strict inequality. On the other hand, if $\mu_1 = \mu_2$, then since $\mu_\lambda = \mu_1 = \mu_2$ and $\thickD$ is stricly convex in its second argument when its first argument is fixed, \eqref{eqn:D-cvx} must be a strict inequality.
	In either case, the sum of the two inequalities must be strict.
    Combining this with the first inequality, we get
	\begin{align*}
		\aar{\dg M \bundle h_\lambda}_\gamma &\le
		\bbr{\dg M}_\gamma(\mu_\lambda) + \kldiv[\Big]{\mu_\lambda(Y)}{h_\lambda(Y) } \\
		&<
		 (\lambda-1) \left[\bbr{\dg M}_\gamma(\mu_1)
			 	+ \kldiv[\Big]{\mu_1(Y)}{h_1(Y)} \right]
			 \\[-0.3em]&\qquad\qquad
			 + \lambda \left[ \bbr{\dg M}_\gamma(\mu_2)
			 	+ \kldiv[\Big]{\mu_2(Y)}{h_2(Y)}
			 	\right] \\
		 &= (\lambda-1) \aar{\dg M \bundle h_1} + \lambda\,\aar{\dg M \bundle h_2},
	\end{align*}
	which shows that $\aar{\dg M \bundle h}$ is \emph{strictly} convex in $h$, as desired.
\end{lproof}

Let $\dg M$ be a PDG with $\bbeta \ge \mat 0$ and
variables $\X$, and fix $Y \subseteq \X$, $y \in \V Y$, and
$\gamma \in (0, \min_a \frac{\beta_a}{\alpha_a})$.
For $p \in [0,1]$, define
\begin{equation}
    f(p) := \aar[\Big]{\dg M + ~\Pr(Y{=}y)}_\gamma.
        \label{eqn:f-defn}
\end{equation}

The next several results (\cref{coro:special-inc-strictly-cvx,lem:fprime,lem:strongly-cvx-ish,lem:score-gradient-bound,lem:D-lowerbound,lem:togetherbound}) are properties of this function $f(p)$.

\begin{coro} \label{coro:special-inc-strictly-cvx}
    The function $f$ defined in \eqref{eqn:f-defn}
    is strictly convex.
\end{coro}
\begin{lproof}
    Simply take $h$ to be the cpd $\hat p$,
    absorb the the other components of (the PDG representation of) $\Pr(Y{=}y)=p$
    into $\dg M$, and then apply \cref{lemma:inc-strictly-cvx}.
\end{lproof}

The results from this point until
\hyperref[proof:inf-via-inc-oracle]{the proof of Theorem 14}
are all technical results that support the more precise analysis of part (b).
We recommend returning to these these results as needed, after first reading 
the proof of part (a).

\begin{lemma} \label{lem:fprime}
    For $p \in (0,1]$, let $\mu_p$ be the unique element of $\bbr{\dg M + \Pr(Y{=}y)=p}^*_\gamma$.
    Then
    $ \displaystyle
        f'(p) = \frac{p - \mu^*_p(Y{=}y)}{p(1-p)}
        .
    $
\end{lemma}
\begin{lproof}

    First, suppose $\mu_p^*$ is in the interior of the simplex.
    Since it minimizes the differentiable function
    \[
        \bbr[\big]{\dg M + \Pr(Y{=}y)=p}_\gamma
            = \mu \mapsto \bbr{\dg M}_\gamma +  \kldiv[\big]{\mu(Y{=}y)}{p}
        ,
    \]
    the gradient of that function at $\mu^*_p$ must be zero:
    \begin{align*}
        \nabla \bbr[\big]{\dg M + \Pr(Y{=}y)=p}_\gamma(\mu^*_p)
            &= \nabla_\mu \Big[ \kldiv[\big]{\mu(Y{=}y)}{p} \Big]_{\mu=\mu^*_p} + \nabla \bbr{\dg M}(\mu^*_p)
            = 0
        .
        \numberthis\label{eqn:matching-gradients-at-optimum}
    \end{align*}

    What is the derivative of $f$?
    Observe that $f$ is the sum of two compositions of differentiable maps:
    \begin{align*}
        f_\dg M &:= \qquad
            p \mapsto ~~~\mu^*_p~ \mapsto  \bbr{\dg M}_\gamma(\mu^*_p) \\
        \text{and}\qquad
        f_{\Pr} &:= \qquad
        p \mapsto (p, \mu^*_p) \mapsto\kldiv[\big]{\mu^*_p(Y{=}y)}{p}.
    \end{align*}
    Thus, we can use the multivariate chain rule.
    for any differentiable functions $h : \mathbb R^m \to \mathbb R^k$,
    and $g : \mathbb R^k \to \mathbb R^n$, their composition $g \circ h : \mathbb R^m \to \mathbb R^n$ is also a differentiable map whose Jacobian is
     $\mat J_{g \circ h}(x) = \mat J_{g}(h(x)) \mat J_{h}(x)$.
    In our case, $g$ will be a scalar map ($n=1$), so $\mat J_g(h(x)) = [ \ldots, \frac{\partial g}{\partial x_j}, \ldots] (h(x)) = \nabla g(h(x))^{\sf T}$.
    Let $\mat J_{\mu^*_p}(p)$ be the Jacobian of the map $p \mapsto \mu^*_p$.
    Then
    \begin{align*}
        f'(p) &= f'_{\dg M}(p) + f'_{\Pr}(p) \\
        &= \big(\nabla \bbr{\dg M}(\mu^*_p)\big)^{\sf T} \mat J_{\mu^*_p}
            + \left(\nabla_\mu
                \Big[ \mu(x) \kldiv{\mu(y|x)}{p}\Big]_{\mu = \mu^*_p} \right)^{\sf T} \mat J_{\mu^*_p} +
                \frac{\partial}{\partial p} \Big[ \kldiv{\mu^*_p(Y{=}y)}{p}\Big] \\
\shortintertext{\centering
    (Alternatively, the line above can be derived from the law of total
    derivative.\footnotemark)
}
        &= \Big( \Cancel{\nabla_\mu
            \Big[ \kldiv{\mu(Y{=}y)}{p}\Big]_{\mu = \mu^*_p}
            } + \Cancel{\nabla \bbr{\dg M}(\mu^*_p)} \Big)^{\sf T}
            \mat J_{\mu^*_p}(p)
            + \frac{\partial}{\partial p} \kldiv{\mu^*_p(Y{=}y)}{p}
            \\
        &= \frac{\partial}{\partial p} \kldiv{\mu^*_p(Y{=}y)}{p}
        & \text{by \eqref{eqn:matching-gradients-at-optimum}}.
    \end{align*}

    Finally,
    \begin{align*}
        \frac{\mathrm d}{\mathrm d p} \Big[ \kldiv{q}{p} \Big]
        &= \frac{\mathrm d}{\mathrm d p}
            \Big[ q \log \frac{q}{p} + (1-q) \log \frac{1-q}{1-p} \Big] \\
        &= q \Big(\frac pq \Big) \frac{\mathrm d}{\mathrm d p} \Big[ \frac{q}{p}\Big]
            + (1-q)\Big(\frac{1-p}{1-q}\Big)  \frac{\mathrm d}{\mathrm d p} \Big[\frac{1-q}{1-p} \Big] \\
        &= pq \Big( \frac{-1}{p^2} \Big)
            + (1-p)(1-q) \Big( \frac{-1}{(1-p)^2} \Big)(-1) \\
        &= - \frac{q}{p}
            + \frac{1-q}{1-p} \\
        &= \frac{-(1-p)q + p(1-q)}{p(1-p)} \\
        &= \frac{pq - q + p - pq}{p(1-p)} 
        = \frac{p - q}{p(1-p)}.
    \end{align*}
    Thus, we find
    \[
        f'(p) = \frac{p - \mu^*_p(Y{=}y)}{p(1-p)},
        \qquad\qquad\text{as promised}. \qedhere
    \]
\end{lproof}
\footnotetext{{
Law of total derivative:
$\displaystyle
    \frac{\mathrm d f}{\mathrm d p}
    = \sum_{w \in \V\!\X} \frac{\partial \kldiv{\mu(y|x)}{p}}{\partial \mu^*_p(w)}
        ( \mu_p^*, p )\frac{\partial \mu^*_p(w)}{\partial p}
    ~+~ \frac{\partial}{\partial p} \thickD( \mu^*_p, p)
    + \sum_{w \in \V\!\X}
        \frac{\partial \bbr{\dg M}_\gamma}{\partial \mu_p^*(w)}
        ( \mu_p^* )\frac{\partial \mu^*_p(w)}{\partial p}.
$}}
\begin{lemma} \label{lem:strongly-cvx-ish}
    If $0 < p_1 < p_2 < 1$
    and $\mu_1^*, \mu_2^*$ are respective minimizing distributions,
    then
    \[
        f(p_2) \ge f(p_1) + f'(p_1) (p_2-p_1) +
            \frac 12\gamma
            \Vert \mu_1^* - \mu_2^* \Vert_1^2
            .
    \]
\end{lemma}
\begin{lproof}
    The general approach is to adapt and strengthen
    the \hyperref[proof:inc-strictly-cvx]{the proof of} \cref{lemma:inc-strictly-cvx},
    to show something like strong convexity, in this special case.
    Define
    \[
        \dg M_1 :=  \dg M ~+~ \Pr(Y{=}y ) = p_1
            \qquad\text{and}\qquad
        \dg M_2 :=  \dg M ~+~ \Pr(Y{=}y ) = p_2.
    \]
    Choose $\mu_1 \in \bbr{\dg M_1}^*_\gamma$ and $\mu_2 \in \bbr{\dg M_2}^*_\gamma$.
    As before, let $m_1 := \mu_1(Y{=}y)$
        and
        $m_2 := \mu_2(Y{=}y)$.
    Then
    \begin{align*}
        f(p_1) &= \aar{\dg M_1}_\gamma = \bbr{\dg M_1}_\gamma(\mu_1)
            = \bbr{\dg M}_\gamma(\mu_1) +  \kldiv {m_1} {p_1}
            \\
            \qquad\text{and}\qquad
        f(p_2) &= \aar{\dg M_2}_\gamma = \bbr{\dg M_2}_\gamma(\mu_2)
            = \bbr{\dg M}_\gamma(\mu_1) + \kldiv{m_2}{p_2}
        ,
    \end{align*}
    where, as before, $\kldiv m p = m \log \frac{m}{p} + (1-m) \log \frac {1-m}{1-p}$ is
    the relative entropy between Bernouli distributions with respective parameters $m$ and $p$.

    For each $\lambda \in [0,1]$, define
    \begin{align*}
        p_\lambda := (1-\lambda)p_1  + \lambda p_2,
        \qquad\quad
            \dg M_\lambda := \dg M + \Pr(Y{=}y)=p_\lambda,
            \qquad\text{and}\qquad
        \mu_\lambda := (1-\lambda) \mu_1  + \lambda \mu_2
        ~.
    \end{align*}

    We now provide stronger analogues of \eqref{eqn:score-cvx} and \eqref{eqn:D-cvx}.
    Since $\bbr{\dg M}_\gamma$ is not just convex but also
    $\gamma$-strongly convex, with respect to the 1-norm (\cref{lem:negent-strongly-cvx-1norm}, below), 
    we can strengthen \eqref{eqn:score-cvx} to
    \begin{align*}
        \bbr{\dg M}_\gamma(\mu_\lambda)
            &\le (1-\lambda) \bbr{\dg M}_\gamma(\mu_1) + \lambda \bbr{\dg M}_\gamma(\mu_2)
            - \frac\gamma2 (1-\lambda)\lambda \Vert \mu_1 - \mu_2 \Vert^2_1.
	\end{align*}
    Adding this inequality to the analogous one describing joint convexity of $\thickD$ in its two arguments \eqref{eqn:D-cvx}, we find that
    \begin{align*}
        &\bbr{\dg M}_\gamma(\mu_\lambda) + \kldiv{m_{\lambda}}{p_\lambda} \\
        & \le (1-\lambda) \Big( \bbr{\dg M}_\gamma(\mu_1) + \kldiv{m_1}{p_1} \Big)
            + \lambda \Big( \bbr{\dg M}_\gamma(\mu_2) + \kldiv{m_2}{p_2} \Big)
             - \frac\gamma2 (1-\lambda)\lambda \Vert \mu_1 - \mu_2 \Vert^2_1
            \\
        &= (1-\lambda) f(p_1)
            + \lambda f(p_2) - \frac\gamma2 (1-\lambda)\lambda \Vert \mu_1 - \mu_2 \Vert^2_1
            . \numberthis\label{eq:mu-lambda-intermediate}
    \end{align*}
    Putting it all together, we find that
    \begin{align*}
    f(p_\lambda)
        &=
        \aar{\dg M_\lambda}_\gamma
        \\&\le \bbr{\dg M_\lambda}_\gamma(\mu_\lambda)
        \\&= \bbr{\dg M}_\gamma(\mu_\lambda) + \kldiv{m_{\lambda}}{p_\lambda} \\
        &\le(1-\lambda) f(p_1)
            + \lambda f(p_2) - \frac\gamma2 (1-\lambda)\lambda \Vert \mu_1 - \mu_2 \Vert^2_1
            & \text{\eqref{eq:mu-lambda-intermediate}}.
            \\
    \end{align*}
    Since this is true for all $\lambda \in [0,1]$, we can divide by $\lambda$, rearrange, and take the limit as $\lambda \to 0$, to find:
    \begin{align*}
        \frac{\gamma}{2} (1-\lambda) \Vert \mu_1 - \mu_2 \Vert^2_1
         &\le\frac{ (1-\lambda) f(p_1) - f(p_1 + \lambda(p_2-p_1))}{\lambda} + f(p_2)
     \\
        &=
         \frac{ f(p_1) - f(p_1 + \lambda(p_2-p_1))}{\lambda} + f(p_2) - f(p_1)
     \\ \implies\qquad
          f(p_2) - f(p_1)
         &\ge \lim_{\lambda \to 0}
         f(p_1) + \frac{f(p_1 + \lambda(p_2-p_1))-f(p_1)}{\lambda} \frac{\gamma}{2} (1-\lambda) \Vert \mu_1 - \mu_2 \Vert^2_1
     \\
         &= f'(p_1) (p_2-p_1) + \frac \gamma2\Vert \mu_1 - \mu_2 \Vert^2_1,
    \end{align*}
    as desired.
\end{lproof}

\begin{lemma} \label{lem:negent-strongly-cvx-1norm}
    Negative entropy is 1-strongly convex with 
    respect with respect to the L1 norm, i.e.,
    $f(\mat p) = \sum_i p_i \log p_i$ satisfies
    \[
        f(\mat q) \ge f(\mat p) + \nabla f(\mat p)^{\sf T} (\mat q - \mat p)
            + \frac12 \Vert \mat p - \mat q \Vert_1^2.
    \]
\end{lemma}
\begin{lproof}
    First, $\nabla f(\mat p) = \log(\mat p) + 1$. 
    Thus, 
    \begin{align*}
        &f(\mat q) - f(\mat p) - \nabla f(\mat p)^{\sf T} (\mat q - \mat p)\\
        &= \sum_i \left[ q_i \log q_i - p_i \log p_i - (\log p_i + 1)(q_i-p_i) \right]\\
        &= \sum_i \left[ q_i \log q_i - \Cancel{p_i \log p_i} - q_i \log p_i + \Cancel{p_i \log p_i} \right]\\
        &= \sum_i q_i \log \frac{q_i}{p_i} \\
        &= \kldiv{q}{p} \\
        &\ge 2 \delta(\mat p, \mat q)^2
            &[\text{by Pinsker's inequality
                \parencite{pinsker-inequality}
            }] \\
        &= 2 \Big(\frac12 \Vert \mat p - \mat q \Vert_1 \Big)^2 \\
        &= \frac12 \Vert \mat p - \mat q \Vert_1^2
    \end{align*}
    where $\delta(\mat p, \mat q)$ is the total variation distance between $\mat p$ and $\mat q$ as measures, and $\Vert p - q \Vert_1 = \sum_i |p_i - q_i|$ is the L1 norm of their difference, as
    points on a simplex. 
\end{lproof}

\cref{lem:strongly-cvx-ish} guarantees
that if the optimal distributions corresponding to adding $p_1$ and $p_2$ to the PDG
($\mu_1$ and $\mu_2$, respectively) are far apart, then so are
$f(p_1)$ and $f(p_2)$.  But what if these optimal distributions are close together?
It turns out that if $\mu_1 \approx \mu_2$ then it's still the case that $f(p_1)$ and $f(p_2)$ are far apart, provied that $p_1$ and $p_2$ are.
However, showing this requires an entirely appraoch, which we pursue in \cref{lem:D-lowerbound}.
But first, we need two intermediate technical results.

\begin{lemma} \label{lem:score-gradient-bound}
    For all $p_1, p_2 \in [0,1]$,
    \[
        \bbr{\dg M}_\gamma(\mu_1^*) - \bbr{\dg M}_\gamma(\mu_2^*)
            \ge (m_2 - m_1) \log \frac{m_2}{p_2} \frac{1-p_2}{1-m_2},
    \]
    where
    $m_1 = \mu_1^*(Y{=}y)$ is the marginal of $\mu_1^* \in \bbr{\dg M + \Pr(Y{=}y) = p_1}^*_\gamma$,
    and $m_2 = \mu_2^*(Y{=}y)$ is defined symmetrically.
\end{lemma}
\begin{lproof}
    For $p, q \in [0,1]$, $\kldiv pq := p \log \frac{p}{q} + (1-p) \log \frac{1-p}{1-q}$ is the relative entropy between the Bernouli distributions described by their parameters.
    First, we calculate
    \begin{align*}
        \frac{\mathrm d}{\mathrm d p} \kldiv{p}{q}
        &= \log \frac{p}{q} + \frac{\Cancel{p}}{\Cancel{p}} \Cancel{q} \frac{\mathrm d}{\mathrm dp} \left[ \frac{p}{\Cancel{q}} \right]
        + (-1) \log \frac{1-p}{1-q} + \Cancel{(1-q)} \frac{\mathrm d}{\mathrm d p}\left[\frac{1-p}{\Cancel{1-q}}\right] \\
        &= \log \frac pq + 1 - \log \frac{1-p}{1-q} -1  \\
        &= \log \left( \frac{p}{q} \frac{1-q}{1-p} \right).
    \end{align*}
    Thus,
    \begin{align*}
        \nabla_\mu \big[ \kldiv {\mu(Y{=}y)}{q} ]
        &= \nabla_\mu[ \mu(Y{=}y) ] \log \left( \frac{\mu(Y{=}y)}{q} \frac{1-q}{1-\mu(Y{=}y)} \right) \\
        &= \mathbbm1[Y{=}y] \log \left( \frac{\mu(Y{=}y)}{q} \frac{1-q}{1-\mu(Y{=}y)} \right).
    \end{align*}

    Recall the stationary conditions, which state that, 
    since $\mu_2^* \in \bbr{\dg M + \Pr(Y{=}y) = p_2}$, we have
    \[
        \nabla_\mu \Big[\kldiv{\mu(Y{=}y)}{p_2} \Big]_{\mu=\mu_2^*} = -
            \nabla_\mu \bbr{\dg M}_\gamma(\mu_2^*).
    \]
    Now use the above to compute the directional derivative of interest:
    \begin{align*}
        &\nabla_\mu \bbr{\dg M}_\gamma(\mu_2^*)^{\sf T} (\mu_1^* - \mu_2^*)
        \\
        &= -(\mu_1^* - \mu_2^*)^{\sf T} \nabla_\mu \big[ \kldiv{\mu(Y{=}y)}{p_2} \big]_{\mu=\mu_2^*} \\
        &= \big( \mu_2^*(Y{=}y) - \mu_1^*(Y{=}y) \big) \log \left( \frac{\mu_2^*(Y{=}y)}{p_2} \frac{1-p_2}{1-\mu_2^*(Y{=}y)} \right) \\
        &= (m_2-m_1) \log \frac{m_2}{p_2}\frac{1-p_2}{1-m_2}.
    \end{align*}
    Finally, since $\bbr{\dg M}_\gamma$ is convex, we have
    \begin{align*}
        \bbr{\dg M}_\gamma(\mu_1^*) - \bbr{\dg M}_\gamma(\mu_2^*)
            &\ge \nabla_\mu \bbr{\dg M}_\gamma(\mu_2^*)^{\sf T} (\mu_1^* - \mu_2^*)
                \\
            &= s_1
                (m_2 - m_1 )
                \log
                \frac{m_2}{p_2}
                \frac{1-p_2}{1-m_2}
    \end{align*}
    as promised.
\end{lproof}

\begin{lemma}\label{lem:D-lowerbound}
    Suppose that $0 < b < z < p^* < 1$,
    and let
    \[
        \mu_z^* \in \bbr[\Big]{\dg M ~+~ \Pr(Y{=}y )=z}^*_\gamma
        \qquad\text{and}\qquad
        \mu_b^* \in \bbr[\Big]{\dg M ~+~ \Pr(Y{=}y )=b}^*_\gamma
    \]
    be the respective optimal distributions for their corresponding PDGs.
    If $\Vert \mu_b - \mu_z \Vert_1 \le \delta$, then
    \[
        f(b) - f(z) \ge (z-b)^2 - \left( \frac2b   + \log \frac{1-b}{1-z} + \log \frac1b \right) \delta.
    \]
\end{lemma}
\begin{lproof}

    Because we have assumed $b < z < p^*$ and $f$ is strictly convex (\cref{lemma:inc-strictly-cvx}),  it follows from \cref{lem:fprime} that
        $m_z := \mu^*_z(Y{=}y) \ge z$ and
        $m_b := \mu^*_b(Y{=}y) \ge b$.
    Now
    \begin{align*}
        f(b) - f(z) &=
        \Big( \bbr{\dg M}_\gamma(\mu_b^*) +  \kldiv{ m_b }{b} \Big)
        -\Big( \bbr{\dg M}_\gamma(\mu_z^*) + \kldiv{ m_z }{z} \Big) \\
        &= \Big( \bbr{\dg M}_\gamma(\mu_b^*) - \bbr{\dg M}_\gamma(\mu_z^*) \Big)
         + \kldiv{m_b}{b} - \kldiv{m_z}{z}.
    \end{align*}
    \cref{lem:score-gradient-bound} will give us a lower bound for the first half;
    we now investigate the second half.
    The first step is some algebraic manipulation.
\colorlet{err4}{blue!50!black}
\colorlet{err3}{violet}
\colorlet{err2}{red}
\colorlet{err1}{orange}
\colorlet{err0}{.}
\colorlet{zero}{gray}
\def\squa#1{{\color{err#1}\blacksquare_{#1}}}
\begin{align*}
&\!\! \kldiv{m_b}{b} - \kldiv{m_z}z \\
&=
    -  m_z \log \frac{m_z}{z} - (1-m_z) \log \frac{1-m_z}{1-z}
    +  m_b \log \frac{m_b}{b} + (1-m_b) \log \frac{1-m_b}{1-b}
\\ &=
    m_z \log \frac{z}{m_z} + (1-m_z) \log \frac{1-z}{1-m_z}
    \quad+( m_b + {\color{zero}m_z - m_z} ) \log \frac{m_b}b
    \\&\qquad\qquad
     + \Big( (1-m_b) +{\color{zero}(1-m_z) - (1-m_z)} \Big)\log\frac{1-m_b}{1-b}
     &\text{(add {\color{zero}zero})}
 \\ &=
     m_z \Big( \log \frac{z}{m_z} + \log \frac{m_b}{b} \Big) +
     (1-m_z) \Big( \log \frac{1-z}{1-m_z} + \log \frac{1-m_b}{1-b} \Big)
        \\&\qquad\qquad
     + (m_b - m_z) \log \frac{m_b}{b}
     + \Big( (1-m_b) - (1-m_z) \Big) \log \frac{1-m_b}{1-b}
    &\text{(collect $z$-marginal terms)}
\\&=
    m_z \log \frac z b {\color{err1}\frac{m_b}{m_z}}
    + (1-m_z) \log \frac{1-z}{1-b} {\color{err1}\frac{1-m_b}{1-m_z}}
    \\&\qquad\qquad
    + {\color{err2}( m_b - m_z) \log \frac{m_b}{b} \frac{1-b}{1-m_b}}
\\&=:~~ \squa0 ~+~ \squa1 ~+~ \squa2,
\end{align*}
where, to be explicit, we have defined
\begin{align*}
  \squa0~ &=  m_z \log \frac zb + (1-m_z) \log \frac{1-z}{1-b}
\\\squa1~ &= \color{err1}
        m_z \log \frac{m_b}{m_z} + (1-m_z) \log \frac{1-m_b}{1-m_z}
    \qquad  = - \kldiv{m_z}{m_b}
\\\squa2~ &= \color{err2}
        ( m_b - m_z) \log \frac{m_b}{b}
                \frac{1-b}{1-m_b}
        .
\intertext{There is one final quantity that will play a similar role. Let}
\squa4~&:=  (m_z - m_b ) \log \frac{m_z}{z} \frac{1-z}{1-m_z}
\end{align*}
be the
lower bound on $\bbr{\dg M}_\gamma(\mu_b^*) - \bbr{\dg M}_\gamma(\mu_z^*)$
obtained by applying \cref{lem:score-gradient-bound} with $p_1 = b$ and $p_2 = z$.
With these definitions, we have $f(b) - f(z) \ge
\squa0 + \squa1 + \squa2 + \squa4$.

Observe that if $m_z$ were equal to $z$, then $\blacksquare_0$ would equal $\kldiv{z}{b}$. But in fact we know that $m_z > z$.
It is easy to see that that $\blacksquare_0$ is linear in $m_z$ with positive slope, since $z > b$ and $1-b  > 1-z$. It follows that $\blacksquare_0 > \kldiv zb$.

Let's step back for a momemnt.
\cref{lem:strongly-cvx-ish} shows that $f(b)$ and $f(z)$ cannot be too close, provided that $\mu_z^*$ and $\mu_b^*$ are far apart.
In the equations above, we can see the beginnings of a complementary argument: if $\mu_z^*$ and $\mu_b^*$ are close together, then $m_b \approx m_b$, and so all terms apart from $\blacksquare_0$ (i.e., $\squa1+\squa2+\squa4$) go to zero. And yet, because of $\squa0$, $f(b)$ and $f(z)$ remain far apart.
To make this argument precise, we now merge $\squa1$, $\squa2$, and $\squa4$ back together, calculating
\begin{align*}
    \squa2 + \squa4  + \squa1
        &= (m_b - m_z) \log \frac{m_b}{b} \frac{1-b}{1-m_b} \frac{z}{m_z} \frac{1-m_z}{1-z}
            - \kldiv{m_z}{m_b} \\
        &= (m_b - m_z) \log \frac{z}{b} \frac{1-b}{1-z}
            + (m_b - m_z) \log \frac{m_b}{m_z} \frac{1-m_z}{1-m_b} + m_z \log \frac {m_b}{m_z} + (1-m_z) \log \frac{1-m_b}{1-m_z} \\
        &= (m_b - m_z) \log \frac{z}{b} \frac{1-b}{1-z}
            + (m_b - \Cancel{m_z} + \Cancel{m_z}) \log \frac{m_b}{m_z}
            + (1- \Cancel{m_z} + \Cancel{m_z} - m_b) \frac{1-m_b}{1-m_z} \\
        &=  (m_b - m_z) \log \frac{z}{b} \frac{1-b}{1-z} + \kldiv{m_b}{m_z} \\
        &\ge (m_b - m_z) \log \frac{z}{b} \frac{1-b}{1-z}.
\end{align*}

Suppose that $\Vert \mu_z^* - \mu_b^*\Vert_1 \le \delta$.
Because the total variation distance $\mathrm{TV}(p,q)$ is half the L1-norm $\Vert p-q\Vert_1$ for discrete distributions,
\[
\delta 
\ge \Vert \mu_z^* - \mu_b^*\Vert_1 
= 2\, \mathrm{TV}(\mu_z^*, \mu_b^*) 
\ge 
2 \big|\mu_z^*(Y{=}y) - \mu_b^*(Y{=}y)\big| 
= 2 |m_z -  m_b|
.\]
Thus,
\[
    \squa1 + \squa2 + \squa4
    ~\ge~
    -\frac\delta2
     \log \frac{z}{b} \frac{1-b}{1-z}.
\]

All that remains is $\squa0$. To put things in a convenient form,
we apply Pinkser's inequality.  The total variation distance between two Bernouli distributions  (i.e., binary distributions) with respective positive probabilities $p$ and $q$ is just $|p-q|$.
Pinsker's inequality \parencite{pinsker-inequality} in this case says: $\frac12 \kldiv{p}{q} \ge (p-q)^2$.
Thus, $\squa0 > \kldiv zb \ge 2 (z-b)^2$,
and so we have
\begin{align*}
    f(b) - f(z) \ge 2 (z-b)^2 - \left(
        \frac12
         \log  \frac zb \frac{1-b}{1-z} \right) \delta
         \qquad\text{as promised.}\qquad
         \qedhere
\end{align*}
\end{lproof}

\begin{lemma} \label{lem:togetherbound}
Suppose that $b < z < p^*$.
Furthermore, suppose that $| \log \frac{z}{b} \frac{1-b}{1-z}| \le k$.
Not only is it the case that $f(b) > f(z)$, but also
\[
    f(b) - f(z) \ge
        \frac{k^2}{32 \gamma}
        \log^2 \left( 1 +  16 \frac{\gamma}{k^2} (z-b)^2 \right)
        .
\]
\end{lemma}
\begin{lproof}
We now have two bounds that work in different regimes.
If $\delta = \Vert \mu_b^* - \mu_z^* \Vert_2$ is large, then the
    argument of \cref{lem:strongly-cvx-ish} is effective, as it shows that a separation between $f(b)$ and $f(z)$ that scales with $\delta^2$.
On the other hand, if $\delta$ is small, we saw in \cref{lem:D-lowerbound} a very different approach that still gets us a separation of $2(z-b)^2$
    even if $\delta = 0$.
We now combine the two cases to eliminate the (unknown) parameter $\delta$ from our complexity analysis.
(Our algorithm is no different in the two cases; all that differs is the analysis.)

Taken together, we know that we attain the maximum of the two lower bounds, which is weakest when they coincide. We can then solve for the worst-case value of $\delta$, which leads to the smallest possible separation between $f(z)$ and $f(b)$. Setting the two bounds equal to one another:
\begin{align*}
    \frac12 \gamma \delta_{\text{worst}}^2 =
        2 (z-b)^2 - \left(\frac{1}{2}
         \log  \frac zb \frac{1-b}{1-z} \right) \delta\\
    \iff\qquad
    \frac\gamma2  \delta_{\text{worst}}^2
    + \Big(\frac{1}{2} \log  \frac zb \frac{1-b}{1-z} \Big)\delta_{\text{worst}}
    - 2 (z-b)^2 = 0.
\end{align*}
The quadratic equation then tells us that
\begin{align*}
     \delta_{\text{worst}} = \frac{1}{\gamma}
        \left( -B + \sqrt{ B^2 + 4 \gamma (z-b)^2 } \right),
        \qquad\text{where}\qquad
        B := \frac{1}{2} \log  \frac zb \frac{1-b}{1-z} 
        .
\end{align*}
It is easily verified that this expression for $\delta_{\text{worst}}$ is decreasing in $B$.
Therefore, we get a lower bound on it by plugging in our upper bound $\frac{k}2$ for $B$.
Thus $\delta_{\text{worst}} \ge \frac{1}{\gamma}
   \left( -k + \sqrt{ k^2 + 4 \gamma (z-b)^2 } \right)$.

Square roots are not easy to manipulate in general, and this expression in particular has involves a nested subtraction that makes it hard to characterize.
To make things clearer, we now begin to loosen this bound to get a quantity that
    is easier to think about.
The first observation is that, for any numbers $A, B > 0$,
\[
    - B + \sqrt{B^2 + A} = B \left( -1 + \sqrt{1 + \frac{A}{B^2}} \right).
\]
This manipulation puts the square root in the standard form $\sqrt{1 + x}$.
Here is the second observation: for all $x$, $-1 + \sqrt{1+ x} \ge \frac12 \log(1+ x)$
    (verified in \cref{lem:sqrt-log-bound} below).
Although it gives a looser bound, the logarithm is easier to manipulate and no longer involves subtraction.
Applying these two transformations in our case, we find:
\begin{align*}
     \delta_{\text{worst}} &=
     \frac{1}{\gamma}
        \left( - B + \sqrt{ B^2 + 4 \gamma (z-b)^2 } \right)\\
    &\ge \frac{1}{\gamma}
        \left( - \frac k2 + \sqrt{ \frac{k^2}{4} + 4 \gamma (z-b)^2 } \right)\\
    &= \frac{k}{2 \gamma}
       \left( -1 + \sqrt{ 1 + \frac{16}{k^2} \gamma  (z-b)^2 } \right)
       \\
    &\ge \frac{k}{4\gamma}
        \log \left( 1 +  \frac{16}{k^2} \gamma (z-b)^2 \right)
    .
\end{align*}
Finally, since $f(b) - f(z) \ge \frac\gamma2 \delta_{\text{worst}}^2$, to get lower bound for $f(b) - f(z)$, we simply need to square this lower bound for $\delta_{\text{worst}}$ and mutliply by $\gamma/2$.
As a result,
\begin{align*}
    f(b) - f(z) ~&\ge~
    \frac\gamma2 \frac{k^2}{16 \gamma^2}
        \log^2 \left(
            1 +  \frac{16 \gamma}{k^2} (z-b)^2 
        \right)
        \\
    &= \frac{k^2}{32 \gamma}
            \log^2 \left(
                1 +  \frac{16 \gamma}{k^2} (z-b)^2
             \right)
        ,
\end{align*}
where $\log^2(x)$ means $(\log(x))^2$.
\end{lproof}

\begin{lemma}\label{lem:sqrt-log-bound}
    $-1 + \sqrt{1+ x} \ge \frac12 \log(1+ x)$.
\end{lemma}
\begin{lproof}
    Apply the well-known inequality $e^{y} \ge 1 + y$
    with $y = -1 + \sqrt{1 + x}$, to get
    \begin{align*}
        \exp(-1 + \sqrt{1+x}) &\ge \sqrt{1+x}, \\
        \text{which implies}\qquad
        -1 + \sqrt{1+x} &\ge \log \sqrt{1+x}
            = \frac12 \log(1+x).
            \qedhere
    \end{align*}
\end{lproof}

We are now ready to tackle the theorem itself.
\clearpage

\recall{theorem:inf-via-inc-oracle}
\begin{lproof}\label{proof:inf-via-inc-oracle}
\textbf{(a,b).}~~
Suppose that we have access to a procedure
    that can calculate a PDG's degree of inconsistency.
The idea behind the reductions of both (a) and (b) is to
to perform $\zogamma$-inference on a given PDG $\dg M$,
using this procedure as a subroutine.
The complexity of the reduction depends on 
the specification of the inconsistency calculation procedure. 
We will perform two analyses, the second building on the first.
\begin{enumerate}
    \item
        First, we assume the procedure simply
        tells us which of two PDGs
        is more inconsistent.
        With this assumption,
        we get an algorithm that can answer unconditional
        probability queries with the
        optimal complexity of part (a) of the theorem.

    \item 
    We then provide a refinement of that algorithm
        that still works if the inconsistency calculation procedure
        produces only finite-precision binary approximations to inconsistency values%
        ---thus reducing the problem of approximate inference that of approximately calculating PDG inconsistency.
    This considerably more difficult analysis gives us part (b).
    In addition, we extend the algorithm, using \cref{lem:logeps-conditioner},
        so that it can also answer conditional queries. 
\end{enumerate}

We begin by describing our algorithm,
which uses the first variant of the inconsistency procudure
    (the one that tells us which of two PDGs is more inconsistent)
to produce a sequence of approximations
$(p_1, p_2, \ldots)$
that converges exponentially to
\[
    p^* := \bbr{\dg M}_\gamma^*(Y{=}y)
        \overset{\vphantom{\big|}\text{(\cref{lem:inc-inc-eq})}}{=}
        \argmin_{p} \aar[\Big]{\dg M + (\Pr(Y{=}y)=p)}_\gamma
\]
through a variant of binary search.
The state of the algorithm
consists of three points in an interval
$a,b,c \in [0,1]$, where $a \le b \le c$.
Intuitively, $b$ is our current best guess at $p^*$, while $a$ is a lower bound, and $c$ is an upper bound.
Once again (as in \eqref{eqn:f-defn} and in preceding lemmas), let $f$ be
the function
\begin{align*}
    f : [0,1] &\to \Rext\\
    p &\mapsto \aar[\Big]{\dg M + (\Pr(Y{=}y)=p)}_\gamma.
\end{align*}
Both variants of the inconsistency calculation procedure will be employed
for the sole purpose of determining whether or not
    $f(p) > f(p')$, given $p, p' \in [0,1]$.
We start with the simpler variant, 
which can directly determine which of two
PDGs has greater inconsistency. 
In this case, the $\gg$ and $\ll$ in the
algorithm below should be interpreted simply as $>$ as $<$.
This will enable us to approximate the minimizer $p^*$ of $f$
arbitrarily closely, with the following algorithm.

    \medskip

\begin{minipage}{0.49\linewidth}
    \rule{4in}{0.2ex}
    \begin{algorithmic}
        \STATE Initialize $(a,b,c) \gets (0, \frac12, 1)$;
        \WHILE{$|c - a| > \epsilon$}
            \IF{$b-a \ge c-b$}
                \STATE Let $z := \frac{b + a}{2}$;
                \smallskip
                \IF{$f(b) \gg f(z)$}
                    \STATE $(a,b,c) \gets (a,z,b)$;
                \ELSE
                    \STATE $(a,b,c) \gets (z,b,c)$;
                \ENDIF
            \medskip
            \ELSIF{$c-b > b-a$}
                \STATE Let $z := \frac{b + c}{2}$;
                \smallskip
                \IF{$f(z) \ll f(b)$}
                    \STATE $(a,b,c) \gets (b,z,c)$;
                \ELSE
                    \STATE $(a,b,c) \gets (a,b,z)$;
                \ENDIF
            \ENDIF
        \ENDWHILE
        \STATE \textbf{return} $b$;
    \end{algorithmic}
    \rule{4in}{0.2ex}
\end{minipage}
\begin{minipage}{0.49\linewidth}
    \begin{tikzpicture}
        \coordinate (c1) at (-0.5,1);
        \coordinate (c2) at (0.75,0.7);
        \fill[red,opacity=0.3] (-3,-.2) rectangle (-2,0.2);
        \fill[orange,opacity=0.3] (-2,-.2) rectangle (-0.5,0.2);
        \fill[green,opacity=0.3] (2,-.2) rectangle (-0.5,0.2);
        \fill[red,opacity=0.3] (2,-.2) rectangle (3,0.2);
        \draw[very thick] (-3,0) -- (3,0);
        \draw[thick] (-2,0.2) -- (-2,-0.2) node[below]{$a$};
        \draw[thick] (2,0.2) -- (2,-0.2) node[below]{$c$};
        \draw[] (0.75,0.2)  -- (0.75,-0.2) node[below]{$z$};
        \draw[thick] (-0.5,0.2) -- (-0.5,-0.2) node[below]{$b$};

        \draw[cyan,gray] plot [smooth] coordinates { (-2,2) (c1) (c2) (1.5,0.6) (2,0.75)};
        \draw[cyan,gray] plot [smooth] coordinates { (-2,3) (c1) (c2) (2,2)};

        \fill (c1) circle (0.1);
        \fill (c2) circle (0.1);

        \node[align=left] at (1.7,2.7) {$\left.\displaystyle\begin{array}{c}
            f~\text{strictly convex}\\
            z > b \\
            f(b) > f(z)
    \end{array}\right\}\implies~~p^* > b$};
    \end{tikzpicture}
\end{minipage}

    We begin by proving that this algorithm does indeed output
    a point within $\epsilon$ of
    $p^*
    $.
    Because $f$ is convex,
    this algorithm satisfies an important invariant: 
    
    \begin{iclaim} \label{claim:reduction-works}
    Both $b$ and $p^*$
    always lie in the interval $[a,c]$.
    \end{iclaim}
    \textit{Proof.~}
        We proceed by induction on $i$.
        At the beginning, it is obviously true that $b$ and $p^*$ lie in
        $[a,b] = [0,1]$, which contains the entire domain of $f$.
        Now, suppose inductively that $p^* \in [a,b]$ at some
        iteration of the algorithm $i$.
        \begin{itemize}[leftmargin=4em]
            \item [(case 1)] If $b-a \ge c-b$, then $z \in [a,b]$.
            \begin{itemize}[leftmargin=-1em]
                \item Suppose $f(z) < f(b)$. Then for
                all $y > b$, it must be the case that $f(y) > f(b)$ by convexity of $f$.
                (For if $f(y) < f(b)$, then segment between $(z, f(z))$ and $(y, f(y))$ would lie entirely below $(b, f(b))$, which contradicts convexity).
                Thus, we can rule out all such $y$ as possible minimizers of $f$, so
                we can restrict our attention to $[a, b]$, which contains $p^*$ (and $x$).

                \item On the other hand, if $f(z) > f(b)$, then it must be the case that no $y < z$ can be a minimizer of $f$ by convexity, with the same reasing as above.
                (Namely, if $f(y) < f(z)$ then the segment between $(y,f(y))$ and $(b,f(b))$ lies below $(z,f(z))$, contradicting convexity).
                Thus the true minimizer $p^*$ lies in $[z,c]$, an interval which contains $b$.
            \end{itemize}
            \item [(case 2)] The other case is symmetric; we include it for completeness. Suppose $c-b > b-a$,
                and so $z = \frac{b+c}{2}$.
            \begin{itemize}[leftmargin=-1em]
                \item Suppose $f(z) < f(b)$. Then $f(y) > f(b)$ for  all $y < b$
                (because if $f(y) < f(b)$, then segment between $(y, f(y))$ and $(x, f(x))$ would lie  below $(b, f(b))$).
                So $p^*, z \in [b, c]$.

                \item On the other hand, if $f(z) > f(b)$, then $f(y) > f(z)$ for all $y > z$
                (because, if $f(y) < f(z)$ then the segment between $(y,f(y))$ and $(b,f(b))$ lies below $(z,f(z))$, contradicting convexity).
                So $p^*, b \in [a,z]$. \qedhere
            \end{itemize}
        \end{itemize}
        In every case, $p^*$ is still in what becomes
        the interval $[a,c]$ in the next iteration ($i+1$).
        So, by induction, $p^* \in [a,b]$ at every iteration of the algorithm,
        proving \cref{claim:reduction-works}.
        \qedsymbol

    We have shown that both $b$ and $p^*$ lie within $[a,c]$,
    and we know that, at termination, $|c-a| < \epsilon$.
    Therefore, the final value of $b$ (i.e., the ouput of the algorithm)
        must be within $\epsilon$ of $p^*$.

Next, we analyze the complexity of this algorithm, modulo the
    complexity of comparing the numbers $f(z)$ and $f(b)$, which
    we will later bound precisely.
    Each iteration reduces the size of the interval $[a,c]$
        by a factor of at least 3/4.
    This is because in each case we focus on the larger half of the interval,
        and ultimately discard either half or all of it---so
        we reduce the size of the interval by at least one quarter.
    It follows that, after $n$ iterations,
        the size of the interval is at most $(\nf34)^n$,
    and thus the total number of iterations is at most $\lceil\, \log(\nicefrac{1}{\epsilon}) / (\log \nicefrac43)\,\rceil$.
    Apart from the time needed to compare $f(b)$ and $f(z)$,
    it is easy to see that each iteration of the algorithm takes constant time.
    So overall, it requires
    $\log(\nf1\epsilon)$ (enough to track the numbers $\{a,b,c\}$, plus a reference to the PDG $\dg M$), and time $O(\log \nf1\epsilon)$, which is linear in the number of bits returned.
    This completes the proof of \cref{theorem:inf-via-inc-oracle} (a). 

    Although it is common to assume that numbers can be compared in $O(1)$ time, and this is an assumption well suited to modern computer architecture, it is arguably not appropriate in this context.
    How do we know a \texttt{float64} is has enough precision to do the comparison? 
The obvious approach to implementing the inconsistency calculation subroutine
    would be to repeatedly request more and more precise estimates of inconsistency,
    until one is larger than the other---but this procedure does not terminate if the two PDGs have the same inconsistency. 
    So, a priori, it's not even clear that the decision $f(b) > f(z)$
    is computable.
    It is not hard to show that it is in fact computable.
    Because $f$ is strictly convex,
    it must be the case that $f(b) > f(z)$, $f(z) > f(b)$, or
            $f(b) > f(\frac{b+z}{2})$.  
    In the last case, we can act as if $f(b) > f(z)$, and 
    the algorithitm will be correct, because
    the argument supporting \cref{claim:reduction-works} still applies.
    Thus, by running the subroutine on all three questions until one of them
        returns true, and then aborting the other two calculations, we can see that
        the problem of interest is decidable.
    But how long does it take?
    We now provide a deeper analysis of the comparison between $f(z)$ and $f(b)$
    when the inconsistency calculation procedure can give us only finite approximations to the true value. 

    \textbf{Part (b): reduction to approximate inconsistency calculation.}
    Instead of assuming that we have direct access to the numbers $f(b)$ and $f(z)$ and
    can compare them in one step, we now adopt a weaker assumption, that we only have
    access to finite approximations to them.
    With this model of computation, it is not obvious that we can determine which of $f(b)$ and $f(z)$ is bigger---but fortunately, we do not need to.
    This is because, when $f(z)$ and $f(b)$ are close, $p^*$ lies between $z$ and $b$, and so both branches of the algorithm maintain the invariant $p^* \in [a,c]$.
    To simplify our analysis, we will default to keeping the ``left'' branch (with the smaller numbers), if the queried approximations to $f(z)$ and $f(b)$ are too close to determine whether one is larger than the other.

    More precisely, the test ``$f(z) \ll f(b)$'' is now shorthand
    for the following procedure:
    \begin{itemize}
        \item
        Run the inconsitency calculation procedure to obtain
        approximations to $f(z)$ and $f(b)$ that are correct to within
        \begin{equation}
            \epsilon' :=
            \frac{1}{16 \gamma}
                    \log^2 \Big(
                         1 +  8 \gamma (z-b)^2
                     \Big).
            \qquad\qquad\text{(This number comes from \cref{lem:togetherbound}.)}
            \label{eq:query-precision}
        \end{equation}
        Call these approximations $\tilde f(z)$ and $\tilde f(b)$.
        By definition, they satisfy $|f(z) - \tilde f(z)| \le \epsilon'$ and $|f(b) - \tilde f(b)| \le \epsilon'$.
        If $|\tilde f(z) - \tilde f(b)| > \epsilon'$ (so that we know for sure which of $f(z)$ and $f(b)$ is larger based on these approximations), then return
        TRUE If $\tilde f(z) < \tilde f(b)$, and FALSE otherwise.
        \item
        On the other hand, if $|\tilde f(z) - \tilde f(b)| \le \epsilon'$,
        return TRUE if $z < b$ and FALSE if $b > z$.
    \end{itemize}

    The remainder of the proof of correctness demonstrates that this level of precision is enough to never mistakenly eliminate the branch containing $p^*$.
    
    We begin by proving a series of three of additional invariants about the values $(a,b,z,c)$ in each iteration, which are required for our analysis.
    The first property is
    that $b$ and $z$ are not too close to the boundary or each other.

    \begin{iclaim} \label{claim:b-z-eps-sep}
        At the beginning of each iteration,
        $b \in [\frac\epsilon2, 1-\frac\epsilon2]$,
        $z \in [\frac\epsilon4, 1-\frac\epsilon4]$,
       and $|b-z| \ge \frac{\epsilon}{4}$.
    \end{iclaim}
    We prove this by contradiction.
    Initially, $b = \frac12$
    so it's neither the case that $b < \frac\epsilon2$ nor that $b > 1-\frac\epsilon2$
    for any $\epsilon < 1$. (The procedure terminates immediately
        if  $\epsilon \ge 1$.)
    In search of a contradiction,
        suppose that either $b < \frac\epsilon2$ or $b > 1 - \frac\epsilon2$ later on.
    Specifically, suppose it first occurs in the $(t+1)^{\text{st}}$ iteration,
        and let $(a_{t+1},b_{t+1},c_{t+1})$ to refer to the values of
        $(a,b,c)$ in that iteration.
    Let $(a_t, b_t, z_t, c_t)$ denote the values of the variables
        in the previous iteration.
    We know that $b_{t+1} \notin [\frac\epsilon2, 1-\frac\epsilon2]$
        and $b_t \in [\frac\epsilon2, 1-\frac\epsilon2]$.
    In particular, $b_{t} \ne b_{t+1}$, which means the procedure
        cannot have taken the second or fourth branches in the $t^{\text{th}}$ iteration.
    There are two remaining cases, corresponding to the first
        and third branches.
    \begin{itemize}
    \item \textbf{(branch 1)~~} In this case, $b_t-a_t \ge c_t - b_t$ and $z_t = {(a_t+b_t)}/2$.  Furthermore, as a result of the assignment in this branch, we have $c_{t+1} = b_t$ and
            $b_{t+1} =  z_t  = {(a_t+b_t)}/2$.
    \begin{itemize}
        \item
        If $b_{t+1} < \frac\epsilon2$, this means $a_t + b_t < \epsilon$.
        As $a_t \ge 0$, this implies $b_{t} = c_{t+1} < \epsilon$.
        But then $|c_{t+1} - a_{t+1}| \le c_{t+1} < \epsilon$, so the algorithm must
        have already terminated! This is a contradiction.
        \item
        If $b_{t+1} > 1-\frac\epsilon2$, then
        $1- \frac\epsilon2 < b_{t+1} = z_t = (a_t+b_t)/2 < b_t$,
        contradicting our assumption that $b_t \le 1-\frac\epsilon2$.
    \end{itemize}

     \item \textbf{(branch 3)~~} In this case, $c_t-b_t > b_t - a_t$ and $z_t = ({b_t+c_t})/2$.
     The assignment at the end of this branch ensures that
        $a_{t+1}=b_t$ and $b_{t+1}=z_t$.
    \begin{itemize}
    \item
     If $b_{t+1} < \frac\epsilon2$, then
         $b_t = a_{t+1} < b_{t+1} < \frac\epsilon2$.
         which is a contradiction.
    \item
        If $b_{t+1} = ({b_t+c_t})/2 > 1-\frac\epsilon2$, then,
        since $c_t \le 1$, we know $b_t + 1 > 2 - \epsilon$, so $b_t = a_{t+1} > 1-\epsilon$.
        But now $|c_{t+1} - a_{t+1}| \le 1 - a_{t+1} < 1- (1-\epsilon) = \epsilon$.
        So the algorithm must have already terminated.
    \end{itemize}
    \end{itemize}
    Thus, it cannot be the case that $b < \frac\epsilon2$ or $b > 1-\frac\epsilon2$ in any iteration of the algorithm. The fact that $z \in [\frac\epsilon4,1-\frac\epsilon4]$ follows immediately from the definition of $z$ in either branch.
    Finally, $|z-b| = \frac12 \max\{c-b, b-a\} \ge \frac12(\frac{c-a}{2}) > \frac\epsilon4$.
    This completes the proof of \cref{claim:b-z-eps-sep}. \qedsymbol

    \begin{iclaim}\label{claim:b-middlethird}
        It is always the case that $b \in [ \frac{2a + c}{3}, \frac{a + 2c}{3}]$.
    \end{iclaim}
    We prove this by induction. It is clearly true at initialization;
    suppose it is also true at time $t$, i.e., $\frac{2 a_t + c_t}{3} \le b_t \le \frac{a_t + 2 c_t}{3}$.
    We now show the same is true at time $t+1$ in each of the four cases of the algorithm.
    \begin{itemize}
        \item \textbf{(branch 1)~~}
            At the end, we assign
                $a_{t+1} = a_t$,
                $b_{t+1} = z = \frac{a_t + b_t}{2}$, and
                $c_{t+1} = b_t$.
            So,
            \begin{align*}
                \frac{2 a_{t+1} + c_{t+1}}{3}
                = \frac{2 a_t + b_t}{3} <
                ~~\underbrace{~~\frac{a_t+b_t}{2}~~}_{\textstyle = b_{t+1}}~~
                < \frac{a_t + 2 b_{t}}{3} =
                \frac{a_{t+1} + 2 c_{t+1}}{3}.
            \end{align*}

        \item \textbf{(branch 2)~~}
        In this case, we must make use of the fact that, in the first
        two branches  $b-a \ge c-b$, meaning $a_t + c_t \le 2b_t$.
        As in the first branch, we have $z = \frac{a_t + b_t}{2}$.
        This time, however,
            $a_{t+1} = z = \frac{a_t + b_t}{2}$,
            $b_{t+1} = b_t$, and
            $c_{t+1} = c_t$.
        Thus, we find
        \begin{align*}
            \frac{2 a_{t+1} + c_{t+1}}{3}
            = \frac{a_t + b_t + c_t}{3}
            \le \frac{(2b_t) + b_t}{3}
            = b_t
             = b_{t+1}
            \le \frac{a_t + 2 c_{t}}{3}
            < \frac{\frac{a_t + b_t}2 + 2 c_{t}}{3}
            = \frac{a_{t+1} + 2 c_{t+1}}{3} .
        \end{align*}

        \item \textbf{(branch 3)~~} Symmetric with branch 1.
        Concretely,
            $a_{t+1} = b_t$,
            $b_{t+1} = z = \frac{b_t + c_t}2$, and
            $c_{t+1} = c_t$. Thus,
        \begin{align*}
            \frac{2 a_{t+1} + c_{t+1}}{3}
            = \frac{2 b_t + c_t}{3} <
            ~~\underbrace{~~\frac{b_t+c_t}{2}~~}_{\textstyle = b_{t+1}}~~
            < \frac{b_t + 2 c_{t}}{3} =
            \frac{a_{t+1} + 2 c_{t+1}}{3}.
        \end{align*}

        \item  \textbf{(branch 4)~~} Symmetric with branch 2.
        Concretely,
            $a_{t+1} = a_t$,
            $b_{t+1} = b_t$,
            $c_{t+1} = z = \frac{b_t+ c_t}2$,
        and we know $2 b_t < a_t + c_t$.
        Thus,
        \begin{align*}
            \frac{2 a_{t+1} + c_{t+1}}{3}
            = \frac{2 a_{t} + \frac{c_{t} + b_t}2}{3}
            < \frac{2 a_t + c_t}{3}
            \le b_t
            = b_{t+1}
            = \frac{2 b_t + b_t}{3}
            < \frac{(a_t + c_t) + b_t}3
            = \frac{a_{t+1} + 2 c_{t+1}}3
            .
        \end{align*}
    \end{itemize}

    The final result we need is a bound for \cref{lem:togetherbound}. 
    
    \begin{iclaim}\label{claim:logbz-bound}
        $| \log \frac{z}{b} \frac{1-b}{1-z} | \le \log 4
            ~~(< \sqrt{2})$.
    \end{iclaim}
    \textit{Proof.}
    Let $\bar a := 1-a$, $\bar b := 1-b$, and $\bar c := 1-c$. 
    \Cref{claim:b-middlethird} tells us that
    $ b \ge \frac{2a + c}3 \ge \frac c3$, and also that $b \le \frac{a+2c}3$,
    which gives us a symmetric fact:
    $
        \bar b = 1-b \ge \frac33 - \frac a3 - \frac{2c}3 = \frac{\bar a + 2 \bar c}{3} \ge \frac{\bar a}3.
    $
    Consider two cases, corresponding to the two definitions of $z$. 
    Either $z = (a+b)/2$ or $z = (b+c)/2$.
    
    \begin{minipage}{0.49\linewidth}
        \textbf{Case 1.} $\frac{a+b}2 = z < b$. Thus,
        \begin{align*}
            \Big| \log \frac zb & \frac{1-b}{1-z} \Big| 
            = \log \frac bz \frac{1-z}{1-b} \\
            &= \log \frac {2b}{a+b} \frac{1-\frac{a+b}2}{1-b} \\
            &= \log \frac {b}{a+b} \frac{2 - a - b}{1-b} \\
            &= \log \frac {b}{a+b} + \log \frac{\bar a + \bar b}{\bar b} \\
            &\le \log \frac{\bar a + \bar b}{\bar b}
                \qquad [\text{as first term is negative}] \\
            &= \log\Big(1 + \nicefrac{\bar a}{\bar b}\Big) \\
            &\le \log \Big(1 + \frac{\bar a}{(\bar a / 3)} \Big)
            = \log 4.
        \end{align*}
    \end{minipage}
    \begin{minipage}{0.49\linewidth}
        \textbf{Case 2.}
        $\frac{b+c}2 = z > b$. Thus,
        \begin{align*}
            \Big| \log \frac zb &\frac{1-b}{1-z} \Big| 
            = \log \frac zb \frac{1-b}{1-z} \\
            &= \log \frac {b+c}{2b} \frac{1-b}{1-\frac{b+c}2} \\
            &= \log \frac {b+c}{b} \frac{1-b}{2 - b - c} \\
            &= \log \frac {b+c}{b} + \log \frac{\bar b}{\bar b + \bar c} \\
            &\le \log \frac{b + c}{b}
                \qquad [\text{as second term is negative}] \\
            &= \log\Big(1 + \nicefrac{c}{b}\Big) \\
            &\le \log \Big(1 + \frac{c}{(c / 3)} \Big)
            = \log 4.
        \end{align*}
    \end{minipage}
    \qedsymbol

    We are now in a position to prove that we never mistakenly eliminate $p^*$
    when comparing truncated representations.
    Without loss of generality, suppose that $z > b$, as the two cases are symmetric.
    Since we choose the left branch in the event of a tie, we have made a mistake
        if we instead needed to have chosen the right branch: $p^* > z$.
    In search of a contradiction, suppose that indeed
        this is the case.
    Under these conditions (i.e., $b < z < p^*$), and in light of \cref{claim:logbz-bound},
    we can apply \cref{lem:togetherbound} with $k = \sqrt{2} > \ln 4$, which tells us that
    \begin{align*}
        f(b) - f(z)
        &>
        \frac{1}{16 \gamma}
        \log^2 \Big(
            1 +  8 \gamma (z-b)^2
        \Big).
    \end{align*}
    The definition of $\epsilon'$ in \eqref{eq:query-precision} is constructed
    precisely to make sure that this is never true.
Therefore, the algorithmn cannot choose the wrong branch.

    \textbf{Complexity Analaysis.}
    We now provide a more careful analysis of the runtime of the algorithm.
    We already have a bound on the number of iterations required;
    what is missing is a bound on how long it takes to compute $\ll$, i.e., to compare the approximations $\tilde f(z)$ and $\tilde f(b)$.
    Assuming these numbers are in binary format, $\tilde f(z)$ is of the form $A.B$, and $\tilde f(b)$ is of the form $A'.B'$, where $\{A, A', B, B'\}$ are binary sequences.

    Without loss of generality, assume that $|A| \le |A'|$. (Otherwise, swap their labels.)  The complexity of comparing the two numbers $\tilde f(z)$ and $\tilde f(b)$ does not depend on $|A'|$, the longer of the two sequences to the left of the radix point.
    This is because once we see the radix point in one number but not the other, we can immediately conclude the former is smaller.
    In the first iteration, $|A|$ is at most
    \begin{align*}
        |A| \le \max \left\{ 0,~ \log_2 \aar[\Big]{\dg M ~+~ \Pr(Y{=}y){=}\frac12}_{\!\gamma} \right\}
        &\le
            \max \left\{ 0,~ \log \Big( \aar{\dg M}_\gamma
        +
            \kldiv{p^*}{.5}
        \Big)
         \right\}\\
         &\le \log_2\Big( \max \{ 0, \aar{\dg M}_\gamma \} + 1 \Big)
         \qquad \big[~\text{since }\kldiv{p}{.5} \le 1~\big]
         \\&\in O( \log \aar{\dg M}_\gamma )
         .
    \end{align*}
    Furthermore, $|A|$ cannot increase as the algorithm progresses, because whichever of $\{z, b\}$ leads to a smaller value of $f$ becomes the new value of $b$ in the following iteration.

    Next, we derive an upper bound on the number of bits of $B$ and $B'$ that we must compare.
    Taking the (base-2) logarithm of \eqref{eq:query-precision}, we find that
    we need to examine at most
    \begin{align*}
        |B| ~&\le~ 4 + \log_2 (\gamma)
            - 2 \log_2 \log \left(1 +  8 \gamma (z-b)^2 \right)
        \\&\le~
            4 + \log_2 ({\gamma})
            - 2 \log_2 \log \left(1 +  \frac1{2} \gamma \epsilon^2\right)
            & \Big[~\text{since $(z-b) \ge \frac\epsilon4$}~\Big]
    \end{align*}
    bits to the right of the radix point in order to eliminate the possibility that $f(b) < f(z)$. 
    This expression is still not very friendly; we now loosen it even further to provide a bound of a simpler, more recognizable form.
    When $x \ge 0$, we know that $\log(1 + x) \ge  1- \frac{1}{1+x} = \frac{x}{x+1}$;
    it follows that $-\log_2 \log(1+x) \le -\log_2(\frac{x}{x+1})
         = \log_2( 1+ \frac1x)$.
    Thus,
    \begin{align*}
        |B| + |A| \ &\le 4 + \log_2(\gamma) + 2 \log_2
        \left(1 + \frac{2}{\gamma\epsilon^2 }\right)
        + \log_2\Big( \aar{\dg M}_\gamma + 1 \Big)\\
        &\in O \left(\log\frac1\epsilon + |\log {\gamma}\, | + \log\aar{\dg M}_\gamma \right).
    \end{align*}

    Recall that the process takes at most $O(\log \frac1\epsilon)$ iterations---but in the process, produces the same number of bits of ouput, since $\log(\nf1\epsilon)$ is the number of bits need to encode the final approximation to $p^*$.
    So, accounting for the time needed to compare $f(z)$ and $f(b)$,
        the algorithm runs in time
    \begin{align*}
        O \left(
        \log \frac1\epsilon
        ~\cdot~
        \Big(\log \frac1{\gamma}
        + \log \frac1\epsilon
        + \log \aar{\dg M}_\gamma
        \Big)
        \right).
    \end{align*}
    
    At this point, we have shown reduced unconditional inference to inconsistency calculation. To extend the reduction to conditional queries, we can apply \cref{lem:logeps-conditioner} with 
    $k=2$, 
    $K_0 = 0$, 
    $K_1 = \log \frac1\gamma + \log(1 + \aar{\dg M}_\gamma)$, 
    $K_2 = 1$, and $\Phi = 1$, 
    to get an algorithm that runs in
    \begin{align*}
        O \left(
        \Big(\log \frac1{\gamma}
        + \log \aar{\dg M}_\gamma
        \Big)
        \cdot
        \log \frac1{\epsilon\, \mu^*\mskip-2mu(X{=}x)}
        +
        \log^2 \frac1{\epsilon\, \mu^*\mskip-2mu(X{=}x)}
        \right)
        ~~\text{time}.
    \end{align*}
    It uses $O(\log \log \frac1{\mu^*(X{=}x)} \cdot \log \frac1\epsilon)$ 
        calls to the inconsistency calculation procedure.

    Finally, we remark that, 
    often we are only interested in doing inference up to the precision that
    is tracked by a typical computer.
    In this case, by selecting
$\epsilon \le 10^{-78}$, the procedure above
runs in constant time, making at most 1555 
    inconsistency procedure calls before outputting
     the 64-bit float that is closest to $p^*$.

    \textbf{The other direction: reducing inconsistency calculation to inference.}
    This reduction is much simpler, shares more techniques with the primary
    thrust of the paper. First find a tree
    decomposition $(\C, \mathcal T)$ of the PDG's structure,
    and then query the marginals of each clique.  Because of the
    work we've already done, we know this information is enough
    information to simply evaluate the scoring function,
    including the joint entropy, by \eqref{eq:cluster-ent-decomp}.
\end{lproof}

\section{The Convex-Concave Procedure, and Implementation Details}
    \label{sec:cccp}

Optimization problems \eqref{prob:joint-small-gamma} and \eqref{prob:cluster-small-gamma}
can be extended to apply slightly more broadly.
There are some cases where there is a unique optimal distribution
but $\gamma$ is large
enough that $\bbeta \not\ge \gamma\balpha$.
In these cases, our convex program will fail to satisfy the dcp requirements, and
    so we cannot compile it to an exponential conic program---but
    it turns out to still be a useful building block.
We now describe how we can still do inference in some of these cases with the
    convex-concave procedure, or CCCP \parencite{yuille2003concave}.
This will give us a local minimum of the
    PDG scoring function $\bbr{\dg M}_\gamma$,
    without requiring us to write this scoring function
    in a way that proves its convexity,
    (as is necessary in order to specify a disciplined convex program).
At this point, if we happen to know that the problem is convex (or even
    just pseudo-convex) for other reasons, then finding this distribution
    suffices for inference.
We now describe how this can be done in more detail.

Suppose
    $\beta_a < \gamma \alpha_a$ some $a \in \Ar$.
In this case $\bbr{\dg M}_\gamma$ may not be convex, in general.%
\footnote{
    Consider the PDG $({\to} X,\, Y {\gets})$ for instance, which
    has arcs to $X$ and $Y$, both with $\alpha = 1$ and $\beta = 0$.
    The minimizers of
     $\bbr{{\to} X,\, Y {\gets}}_\gamma$
        are the distributions that make $X$ and $Y$ is independent.
    It is easily seen that this set
    is not convex: $X$ and $Y$ are independent if either variable is deterministic,
    and every distribution is a convex combination of deterministic distributions.
    }
However, we do know how to decompose $\bbr{\dg M}_\gamma$ into a sum
of a convex function $f(\mu)$ and a concave one $g(\mu)$.
Concretely: each term on the second line of \eqref{eq:altscore} is either convex or concave, depending on the sign of the quantity $\gamma \alpha_a - \beta_a$.
Once we sort the terms into convex terms $f(\mu)$ and strictly concave terms $g(\mu)$,
    the CCCP tells us to repeatedly solve $f$ plus a
        linear approximation to $g$.
In more detail, the algorithm proceeds as follows.
First, choose an initial guess $\mu_0$, and iteratively use the convex solver
    as in the main paper to compute
\begin{align*}
    \mu_{t+1} &:= \argmin_{\mu} f(\mu) + (\mu - \mu_{t})^{\sf T}
        \nabla g(\mu_t).
\end{align*}
This can be slow because each iteration of the solver is expensive.
Still, it is guaranteed to make progress, since
\def\tplus1{{t\mskip-2mu+\mskip-2mu1}}
\begin{align*}
    f(\mu_\tplus1) \!+\! g(\mu_\tplus1) &<  f(\mu_\tplus1) \!+\! (\mu_\tplus1 \!-\! \mu_t)^{\sf T} \nabla g(\mu_t) \!+\! g(\mu_t)
        \\
    &\le  f(\mu_t) + (\mu_t - \mu_{t})^{\sf T}\nabla g(\mu_t)  + g(\mu_t)
    \\
&= f(\mu_t) + g(\mu_t).
\end{align*}
Furthermore, because in our case $g$ is bounded,
    the process eventually converges a local minimum
    of $\bbr{\dg M}_\gamma$.
This alone, however, is not sufficient for inference,
because we may not be able to use this local minimum
    to answer queries in a way that is true of \emph{all} minimizing distributions.
But, if it happens there is a unique local minimum, then the CCCP will
    find it, leading to an inference procedure.

Notice that if $\bbeta \ge \gamma\balpha$, then the concave part $g$ is identically
zero, and CCCP converges after making just one call to the convex solver.
Therefore, in the cases we could already handle, this extension
reduces to the algorithm we described before.
For this reason, all of our code that handles problems
 \eqref{prob:joint-small-gamma} and \eqref{prob:cluster-small-gamma}
is augmented with the CCCP.

Compared to the black-box optimization baselines (Adam and LBFGS),
 which also only find one minimum,
    the CCCP still has some advantages.
One can see in \cref{fig:gamma-v-gap-fine}, for example, that when $\gamma = 2 > 1 = \max_a {(\beta_a / \alpha_a)}$,
CCCP performs better than the baselines.
In fact, the CCCP-augmented
solver
could probably even higher accuracy,
    if were we not limiting it to
    a maximum of only five iterations.

\section{Details on the Empirical Evaluation}\label{sec:expt-setup}
Imagine a very steep $V$-shaped canyon, and inside a small slow-moving stream at a gentle incline. The end of the river may be very far away, and the whole landscape may be smooth and strongly convex, but the gradient will still almost always point perpendicular to it, and rather towards the center of the river.
This intuition may help explain why, even though $\bbr{\dg M}_\gamma$ is infinitely differentiable in $\mu$ and $\gamma$-strongly convex, it can still be challenging to optimize, especially when the $\beta$'s are very different, or when $\gamma$ is small.
For example, a solution to \eqref{prob:cluster-inc} finds a minimizer of $\OInc$, but such minimizers may be very far away from $\bbr{\dg M}_{0^+}^*$, despite sharing an objective value.

We now see how this is true even when working with very small PDGs and joint distributions.

\begin{figure}
    \includegraphics[width=\linewidth]{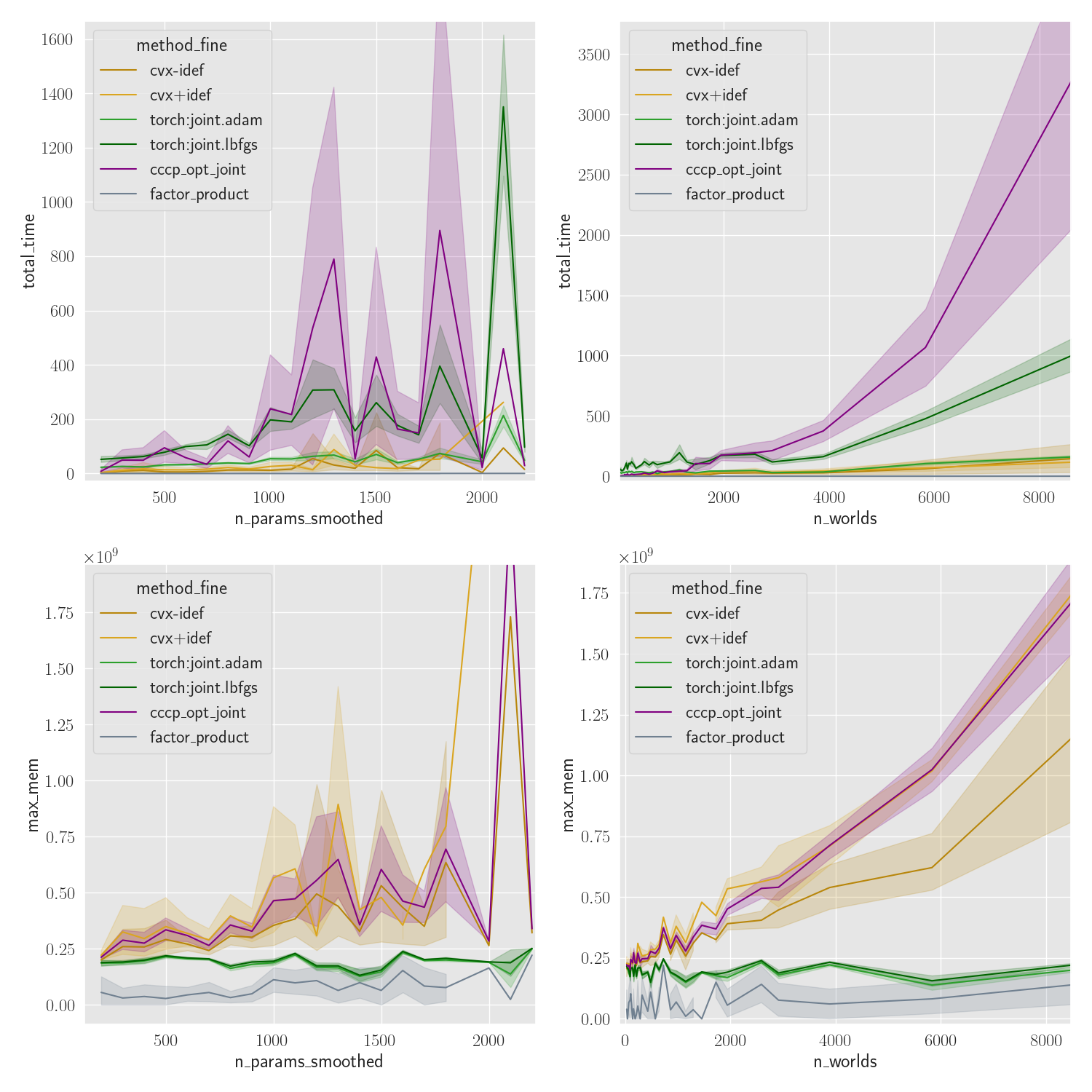}
    \caption{\small
        Resource costs for the joint-distribution optimization setting of \cref{sec:inf-as-cvx-program}.
        We measure computation time (\texttt{total\_time}, top) and maximum memory usage (\texttt{max\_mem}, bottom) for the various optimization methods (by color), as the size of the PDG increases, as measured by the number of parameters in the PDG (\texttt{n\_params}$\,=\V\!\Ar$, left), and the size of a joint distribution over its variables (\texttt{n\_worlds}$\,=\V\!\X$, right).
        Note that the convex solvers for the 0 and $0^+$ semantics are significantly faster than LBFGS, and on par with Adam.
        However, all three convex-solver based approaches require significantly more memory than the black-box optimizers.
     }\label{fig:resources}
\end{figure}

\subsection{Synthetic Experiment: Comparison with Black-Box Optimizers, on Joint Distributions.} \label{sec:joint-expt-details}
\begin{wrapfigure}{R}{4.8cm}
    \centering
    \vspace{-3ex}
    \includegraphics[width=4.5cm]{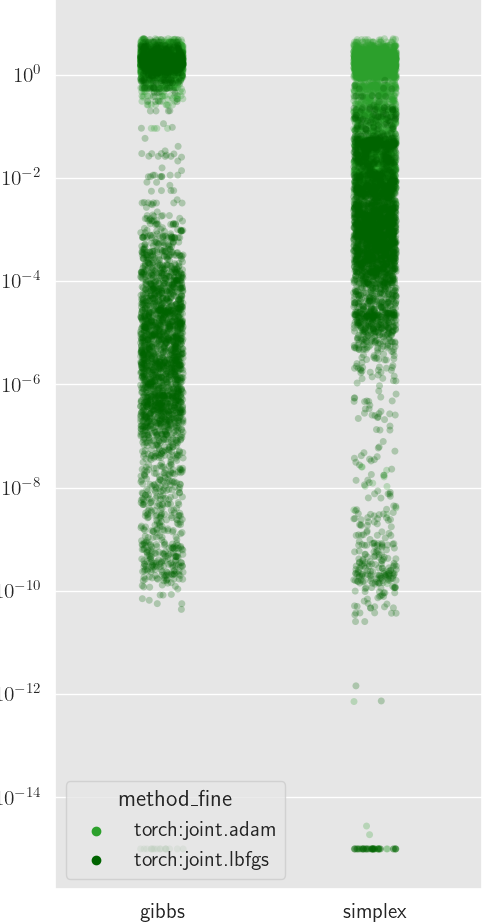}
    \caption{\small differences in performance between the Gibbs and simplex parameterizations of probabilities.}%
    \label{fig:representation}%
\end{wrapfigure}

Here is a more precise description of our first synthetic experiment,
on joint distributions, which contrasts the convex optimization approaches of \cref{sec:inf-as-cvx-program} with black-box optimizers.

\begin{itemize}
    \item generate 300 PDGs, each of which has the following quantities, to each of which we choose the following natural numbers uniformly at random:
    \begin{itemize}
        \item $N \in \{5,\ldots,9\}$ of variables
            (so that $\X := \{1, \ldots, N\}$),
        \item $V_X \in \{2, 3\}$ values per variable
            (so that $|\V X| = V_X$ for each $X \in \X$)
        \item $A \in \{7, \ldots, 14\}$ hyperarcs,
        each $a \in \{1, \ldots, A\} =: \Ar$ of which has
        \item $N^S_a \in \{0, 1, 2, 3\}$ sources, and
        \item $N^T_a \in \{1,2\}$ targets.
    \end{itemize}
    \item For each arc $a \in \Ar$, $N^S_a$ of the $N$ variables are choosen without replacement to be sources $S_a \subseteq N$, and $N^T_a$ of remaining variables are chosen to be targets. Finally, to each value of $S_a$ and $T_a$, a number $p_{a,s,t} \in [0,1]$ is chosen uniformly at random, and the cpd
    \[
     \p_a(\Tgt a {=} t \mid \Src a {=} s) =
        \frac{p_{a,s,t}}{\displaystyle\sum_{t' \in \V(T)} p_{a,s,t'}}
     \qquad\text{ is given by normalizing appropriately.}
    \]
    This defines a PDG $\dg M = (\X, \Ar, \mathbb P, \mat 1, \mat 1)$, that
    has $\balpha = \bbeta = \mat 1$, which will allow us to comapre against
    belief propogation and other graphical models at $\gamma = 1$.
    The complexity of this PDG is summarized by two numbers:
    \begin{itemize}[nosep]
        \item \texttt{n\_params}$\,:= \V\!\Ar$, the total number of parameters in all cpds of $\dg M$, and
        \item \texttt{n\_worlds}$\,:= \V\!\X$, the dimension of joint distributions over $\dg M$'s variables.
    \end{itemize}
\end{itemize}

\begin{itemize}
    \item Run MOSEK on \eqref{prob:joint-inc} to find a distribution that minimizes $\OInc$; we refer to this method as \texttt{cvx-idef}
    \item Use the result to run MOSEK on \eqref{prob:joint+idef} to find the special distribution $\bbr{\dg M}^*_{0^+}$; we refer to this method as \texttt{cvx+idef}. These names are due to the fact that $\SInc$ is called $\IDef{}$ in previous work \parencite{pdg-aaai,one-true-loss};
    thus, this refers to using the convex solver to compute minimizers of $\OInc$ with and without considering $\IDef{}$.

    \item Run the \texttt{pytorch} baselines.
    Let $\theta = [\theta_{\mat x}]_{\mat x \in \V\!\X} \in \mathbb R^{\V\X}$ be a vector of optimization variables, and choose a representation of the joint distribution, either by
    \[
        \Big(\begin{array}{c}\text{renormalized}\\
        \text{simplex}\end{array}\Big)\quad
        \mu_{\theta}(\mat x) = \frac{\max\{\theta_{\mat x}, 0\} }
            {\sum_{\mat y \in \V\!\X}\max\{\theta_{\mat y}, 0\} }
    \qquad\text{or}\qquad
    \mu_{\theta}(\mat x) = \frac{\exp(\theta_{\mat x})}{\sum_{\mat y \in \V\!\X} \exp(\theta_{\mat y})}\quad(\text{Gibbs})
    \qquad\quad\text{(see \cref{fig:representation})}
    \]
    \item
    For each value of the trade-off paramteter 
    $\gamma \in \{0, 10^{-8}, 10^{-4}, 10^{-2}, 1\}$, and each learning rate $\texttt{lr} \in 1E-3, 1E-2, 1E-1, 1E0$, and each optimizer $\mathit{opt} \in \{\texttt{adam}, \texttt{L-BFGS}\}$,
    run $\mathit{opt}$ over the parameters $\theta$ to minimize $\bbr{\dg M}_\gamma(\mu_{\theta})$
     until convergence (or a maximum of 1500 iterations)

     \item We collect the following data about the resulting distribution and the process of computing it:
     \begin{itemize}[nosep]
         \item the total time taken to arrive at $\mu$;
         \item the maximmum memory taken by the process computing $\mu$;
         \item the objective and its component values:
         \vspace{-1ex}
         \[
            \texttt{inc} := \SInc_{\dg M}(\mu),
            \qquad \texttt{idef} := \SInc_{\dg M}(\mu),
            \qquad \texttt{obj} := \OInc_{\dg M}(\mu) + \gamma \SInc_{\dg M}(\mu) = \bbr{\dg M}_\gamma(\mu)
        \]
     \end{itemize}
\end{itemize}

The numbers can then be recreated by running our experimental script as follows:
\begin{verbatim}
python random_expts.py -N 300 -n 5 9 -e 7 14 -v 2 3
    --ozrs lbfgs adam
    --learning-rates 1E0 1E-1 1E-2 1E-3
    --gammas 0 1E-8 1E-4 1E-2 1E0
    --num-cores 20
    --data-dir random-joint-data
\end{verbatim}
which creates a folder called \verb|random-joint-data|,
and fills it with \verb|.mpt| files corresponding to each distribution
and the method / parameters that gave rise to it.

\textbf{Analyzing the Results.}
Look at  \cref{fig:resources}.  Our theoretical analysis, and in particular the proof of \cref{lem:cluster-inc-polytime}, suggest that the magnitudes of $\V\!\X$ and $\V\!\Ar$ play similar roles in the asymptotic complexity of PDG inference.
Our experiments reveal that, at least for random PDGs, the number of worlds is the far more important of the two; observe how much more variation there is on the left side of the figure than the right---and now note that the left side has been smoothed, while the right side has not.
The black-box py-torch based approaches clearly have an edge in that they can handle larger models, as evidenced by the cut-offs on the right-hand side of \cref{fig:gap-resource-fine-old}, when with 5GB memory.

Note that the exponential-cone-based methods for the observational limit (gold) are actually faster than L-BFGS (the black-box optimizer with the lowest gap), and also seem to be growing at a slower rate.
However, they use significantly more memory, and cannot handle large models.
In addition to being faster, our techniques also seem to be more precise; they achieve objective values that are consistently much better than the black-box methods.

Now look at \cref{fig:joint-gap-vs-time-by-gamma},
which contains a break-down of the information in \cref{fig:joint-gap-time}. The bottom half of the figure is just the same information, but with each value of $\gamma$ separated out, so that the special cases of the factor product and $0^+$ inference become clear, while the top half shows why it's more important to look at the gap than the actual objective value for these random PDGs.
\Cref{fig:joint-gap-vs-time-by-gamma} also makes it clearer how larger problems take longer, and especially so for \texttt{cccp} (violet), which solves the most complex version of the problem \eqref{prob:joint-small-gamma}.

\begin{figure}
    \includegraphics[height=0.45\textheight]{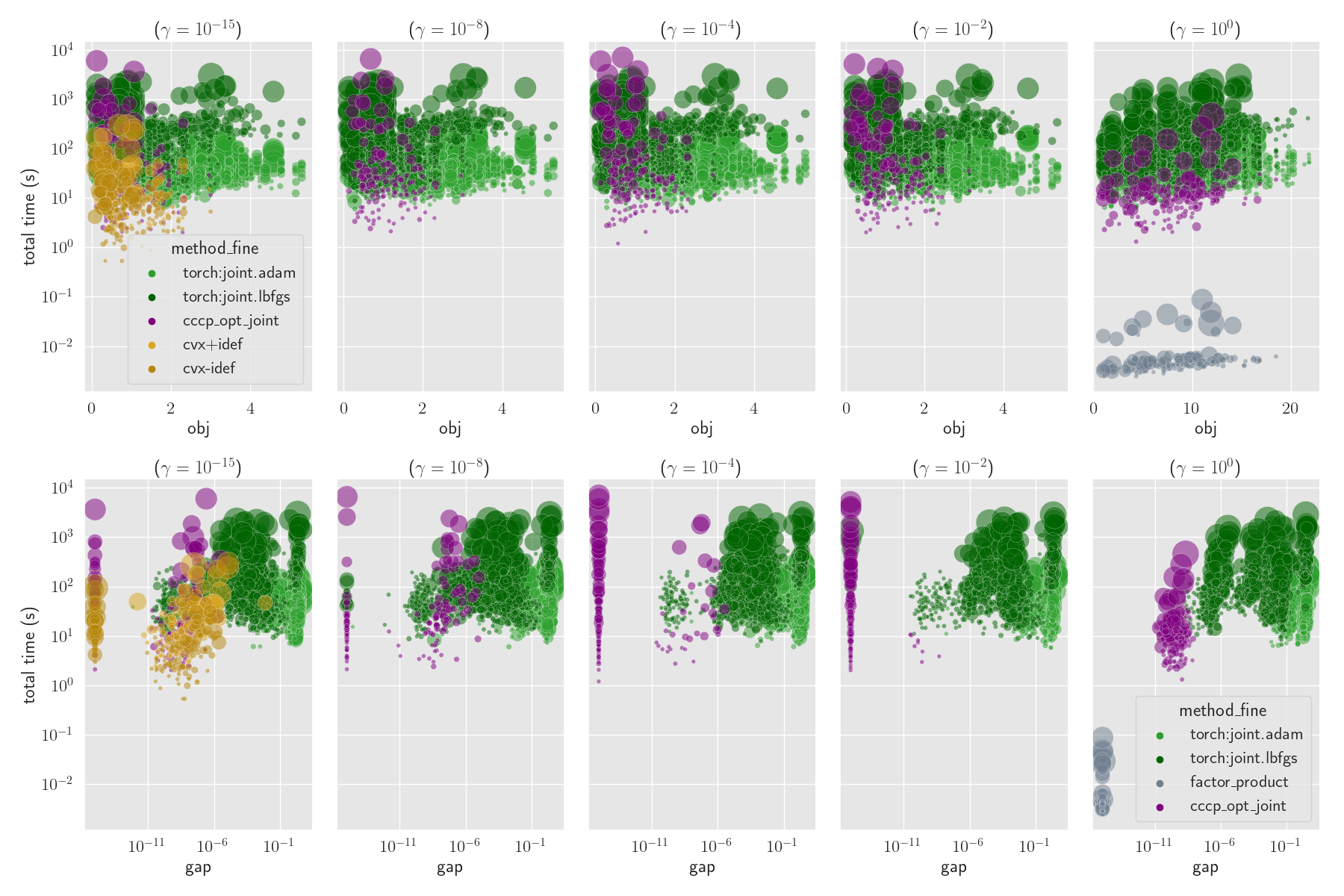}
    \caption{\small
        An un-compressed version of the information in \cref{fig:joint-gap-time}, that groups by the value of $\gamma$, and also gives the absolute values of the objectives (top row) in addition to the relative gaps (bottom row).
    }\label{fig:joint-gap-vs-time-by-gamma}
\end{figure}

\discard{\begin{figure}
    \includegraphics[height=0.45\textheight]{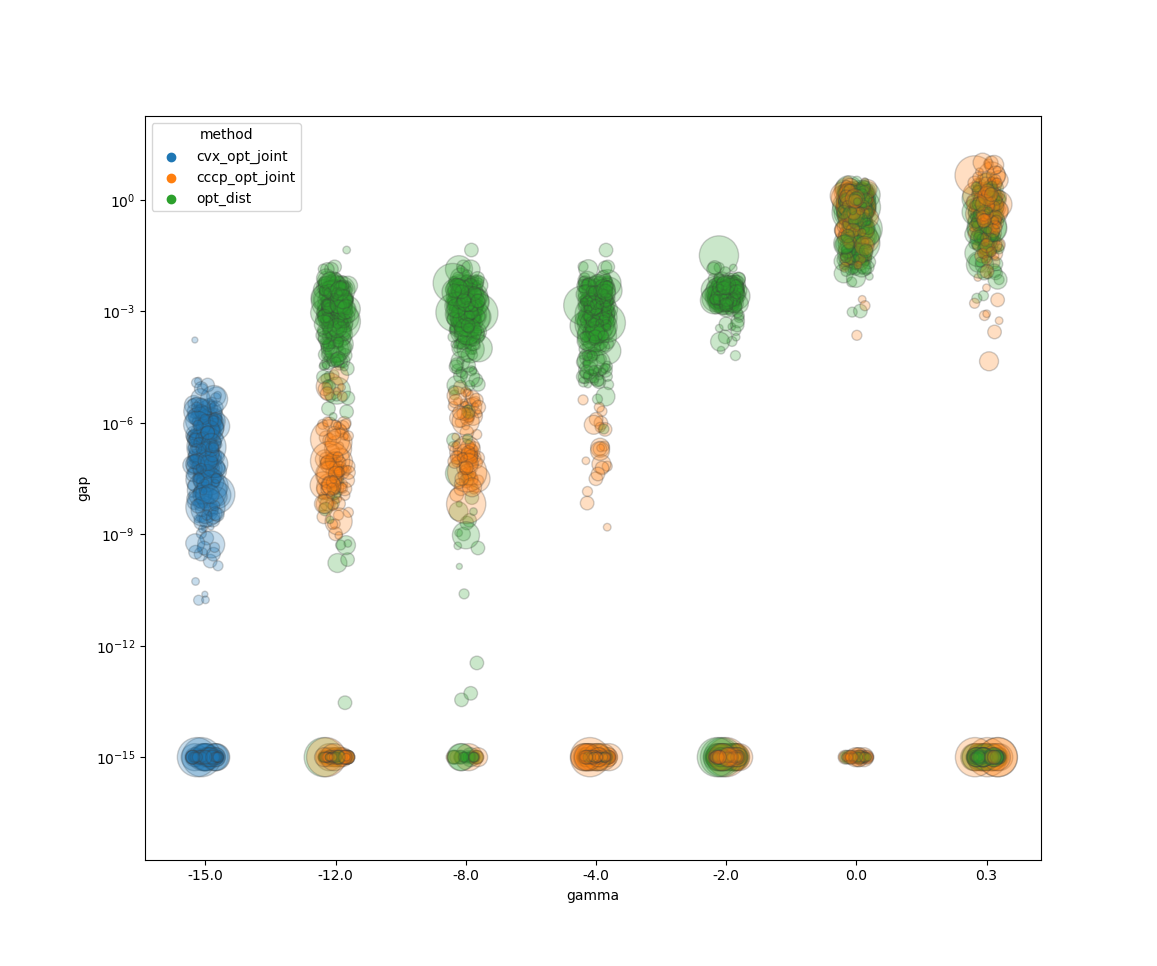}
    \caption{
        A graph of the gap (the difference between the attained objective value, and the best objective value obtained across all methods for that value of $\gamma$),
        as $\gamma$ varies. The x-axis is $\log_{10} ( \gamma + 10^{-15})$.
        As before, colors indicate the optimization method; here blue corresponds to \eqref{prob:joint+idef}, while orange corresonds to \eqref{prob:joint-small-gamma}, and green, as before, correponds to all optimization baselines.
        The size of the circle illustrates the relative number of worlds.
        See \cref{fig:gamma-v-gap-fine} for a more detailed breakdown.
    }\label{fig:gamma-v-gap}
\end{figure}}

\begin{figure}
    \includegraphics[width=\linewidth]{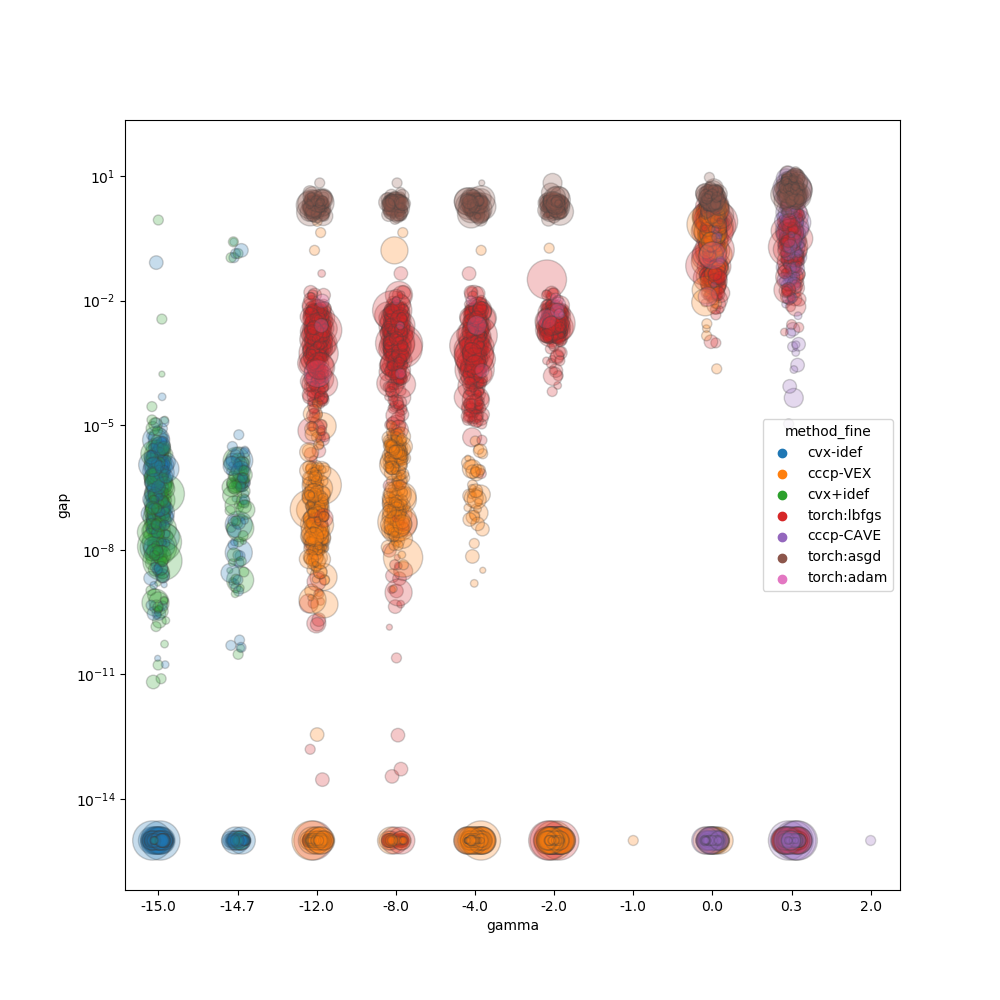}
    \caption{
        A graph of the gap (the difference between the attained objective value, and the best objective value obtained across all methods for that value of $\gamma$),
        as $\gamma$ varies. The x-axis is $\log_{10} ( \gamma + 10^{-15})$.
        As before, colors indicate the optimization method, and
        the size of the circle illustrates the number of optimization variables (i.e., the number of possible worlds).
        \texttt{cvx-idef} corresponds to just solving \eqref{prob:joint-inc}, and \texttt{cvx+idef} corresponds to then solving problem \eqref{prob:joint+idef} afterwards.
        The CCCP runs are split into regimes where the entire problem is convex ($\gamma \le 1$, labeled \texttt{cccp-VEX}), and the entire problem is concave ($\gamma > 1$, labeled \texttt{cccp-CAVE}).
        The optimization approaches \texttt{opt\_dist} are split into three different optimizers: LBFGS, Adam, and also a third one that
        performs relatively poorly: accelerated gradient descent.
        Note that for small $\gamma$, the exponential-cone based methods significantly outperform the gradient-based ones.
    }\label{fig:gamma-v-gap-fine}
\end{figure}

\subsection{Synthetic Experiment: Comparing with Black-Box Optimizers, on \AcTree s} \label{sec:clus-expt-details}

\begin{figure}
    \centering
    \includegraphics[width=0.67\linewidth]{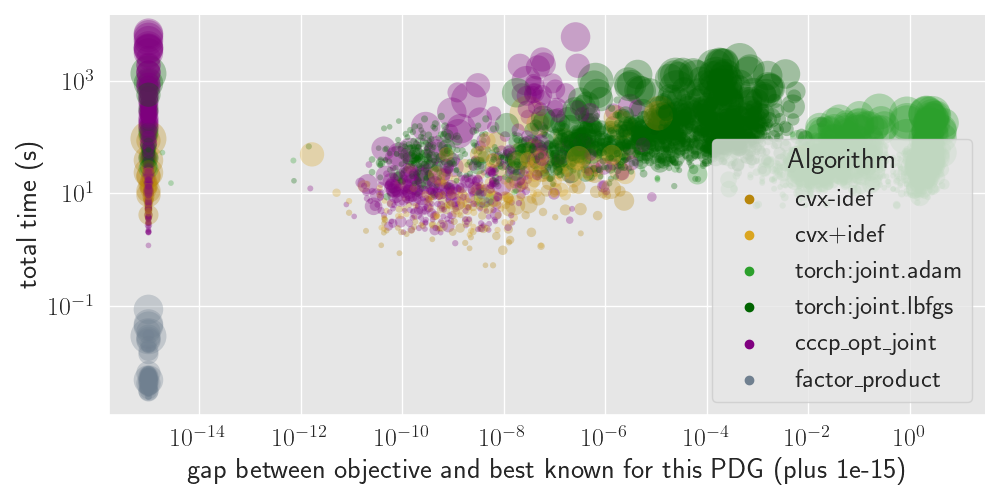}
    \caption{An analogue of \cref{fig:joint-gap-time}, for the cluster setting.
    Note that there is even more separation between the exponential-cone based approaches, and the black-box optimization based ones.
    The new grey points on the bottom correspond to belief propogation, which is both faster and typically the most accurate.}
    \label{fig:clus-gap-vs-time--appendix}
\end{figure}
\begin{figure}
    \centering
    \includegraphics[width=0.67\linewidth]{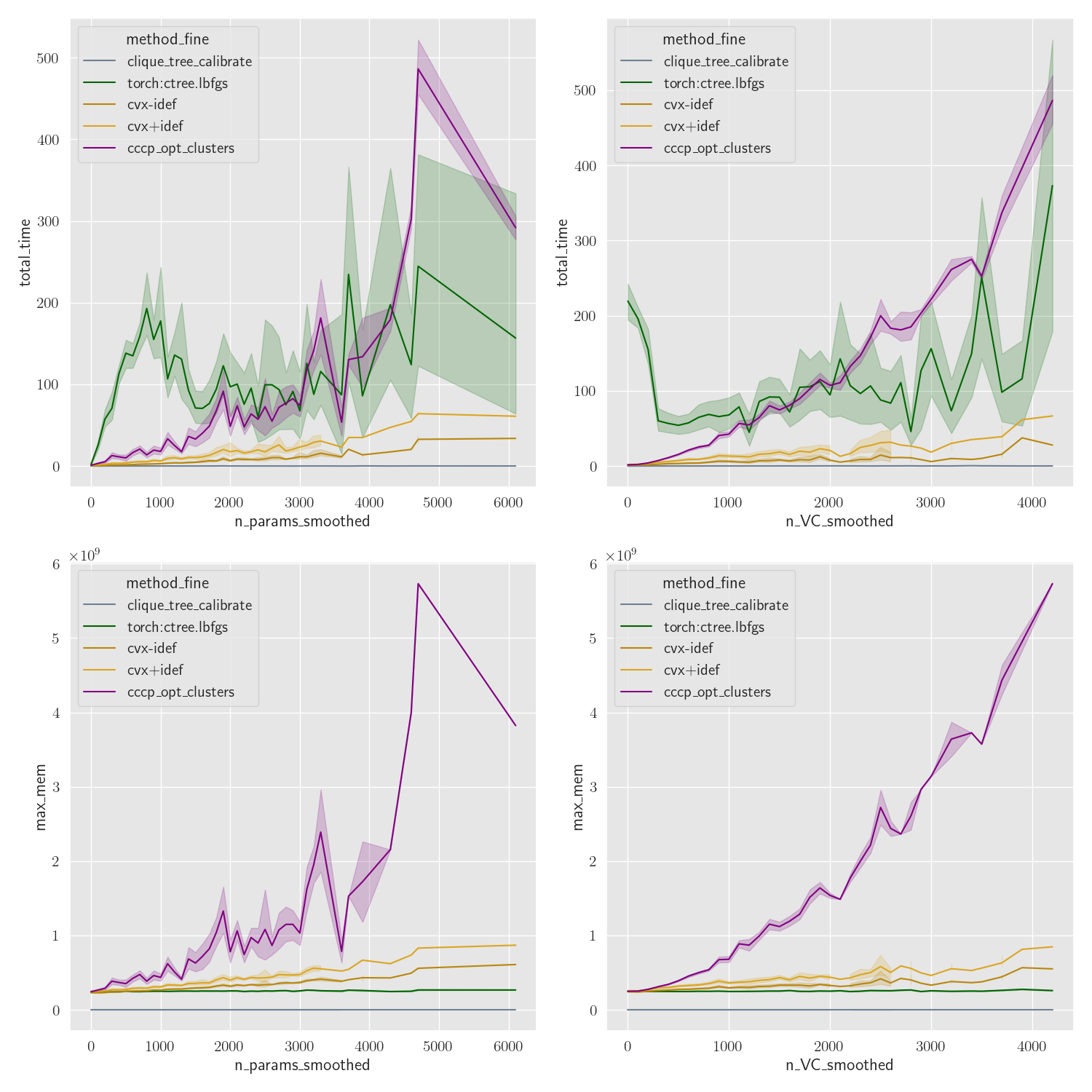}
    \caption{Resource costs for the cluster setting. Once again, the $\OInc$-optimimzing exponential cone methods are in gold, the small-gamma and CCCP is in violet, and the baselines are in green. The bottom line is belief propogation, which is significantly faster and requires very little memory, but also only gives the correct answer under very specific cirucumstances.}
    \label{fig:clus-resource-costs}
\end{figure}

\begin{enumerate}
    \item Choose a number of variables $N \in \{ 8, \ldots, 32 \}$, and a treewidth $k \in \{1, \ldots, 4\}$ uniformly at random.
    Then draw a random $k$-tree and corresponding tree of clusters $(\C, \mathcal T)$, as follows:
    \begin{enumerate}
        \item Initialize $G \gets K_{k+1}$ to a complete graph on $k+1$ vertices, and $\C \gets\{ K_{k+1} \}$ to be set containing a single cluster, and $\mathcal T\gets \emptyset$.
        \item Until there are $N$ vertices: add a new vertex $v$ to $G$, then randomly select a size $k$-clique (fully-connected subgraph) $U \subset G$, and add edges between $v$ and every vertex $u \in U$.
        Add $U \cup \{v\}$ to $\C$, and add edges to every other cluster $C \in \C$ such that $U \subset C$.
    \end{enumerate}
    \item Draw the same parameters $V_X \in \{2,3\}$, $A \in \{8, \cdots, 120\}$, $N_a^S \in \{0,1,2,3\}$, and $N^{T}\in \{1,2\}$
    as in \cref{sec:joint-expt-details} uniformly at random.
    While $N_a^S + N^T_a > k+1$, for any $a$, resample $N_a^S$ and $N_a^T$.

    \item Form a PDG whose structure $\Ar$ can be decomposed by $(\C, \cal T)$, as follows:
    for each edge $a \in \Ar$, sample a cluster $C \in \C$ uniformly at random; then select $N_a^S$ nodes from that cluster without replacement as sources, and $N_a^T$ nodes as targets; this is possible because each cluster has $k+1$ nodes, and $N_a^S + N_a^T \le k+1$ by construction.
    \item Fill in the probabilities by drawing uniform random numbers and re-normalizing, just as before, to form a PDG $\dg M$

    \item The black-box optimization baselines work in much the same way also, although now the optimization variables include not one distribution $\mu$ but a collection $\bmu$ of them;
    this time, we use only the simplex representation of $\bmu_\theta$.
    More importantly, we want these clusters to share appropriate marginals; to encourage this, we add a terms to the loss function, so overall, it is
    \[
        \ell(\theta) := \bbr{\dg M}_{\gamma}(\bmu_{\theta}) + \sum_{C{-}D \in \mathcal T} \exp \left(\sum_{w \in \V(C\cap D)} \Big(\mu_C(C\cap D{=}w) - \mu_D(C \cap D {=}w)\Big)^2 \right) - 1.
    \]
    This is admittedly pretty ad-hoc; the point is just that it is zero and does not contribute to the gradient if $\bmu_\theta$ is calibrated, and otherwise quickly becomes overwhelmingly important.
\end{enumerate}

\textbf{Analyzing the Results.}
Observe in \cref{fig:clus-gap-vs-time} that the separation between the \actree\ convex solver and the black-box algorithms is even more distinct. This is because, in this case, the penalty for violating constraints was too small, and the optimization effort was largely wiped out by the calibration before evalution.

\begin{figure}
    \centering
    \includegraphics[width=0.7\textwidth]{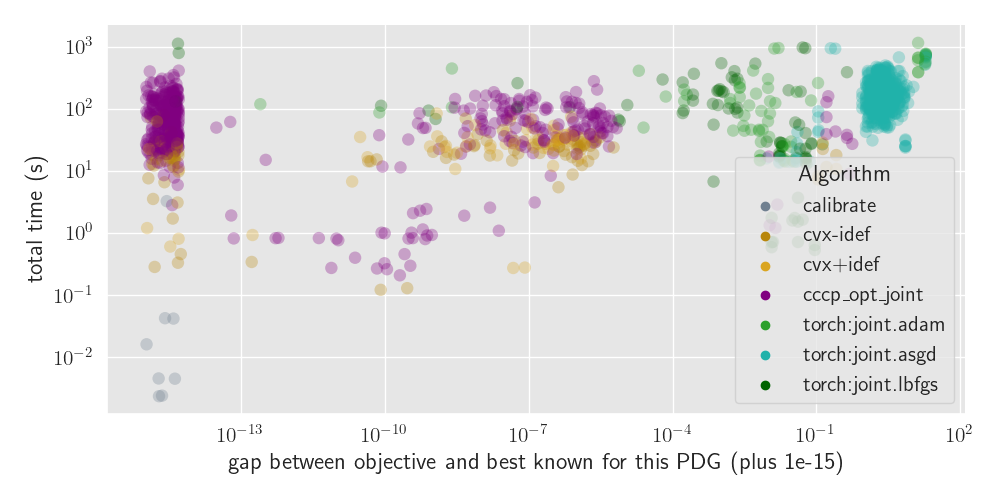}
    \caption{Gap vs inference time for the small PDGs in the \href{https://www.bnlearn.com/bnrepository/}{\texttt{bnlearn}} repository}
    \label{fig:bn-gap-v-time}
\end{figure}

This illustrates another general advantage that the convex solver has over black-box optimizers: it is much less brittle and reliant and exactly tuning parameters correctly. Note that even in this minimal example, there were many hyper-parameters that require tuning:
the regularization strengths that enforce soft constraints (\actree\ calibration, normalization), as well as learning rate, not to mention
various other structural choices: the optimizer, the representation of the distribution, and the maximum number of iterations, none of which are clear-cut choices, but rather require first being tuned to the data.
While the convex solver does have internal parameters (tolerences and such) these do not need to be tuned to the problem under normal circumstances.

\subsection{Comparing to Belief Propagation, on \AcTree s.}
    \label{sec:bn-expt-details}
 Since PDGs generalize other graphical models, one might wonder how our method stacks up against algorithms tailored to the more traditional models. In brief: our algorithm is much slower, and only handle much smaller networks.
 Concretely, our methods can handle all of the ``small'' networks, and some of the ``medium'' ones, from the \href{https://www.bnlearn.com/bnrepository/}{\texttt{bnlearn}} repository.
  In these cases, we have verified that the two methods yield the same results.
  \Cref{fig:bn-gap-v-time} contains the analogue of \cref{fig:joint-gap-time,fig:clus-gap-vs-time}
  for the Bayesian Nets. This graph looks qualitatively quite similar to the other graphs we've seen, suggesting that the results in our synthetic experiments hold more broadly for small real-world models as well.

 \begin{figure}
     \includegraphics[height=0.45\textheight]{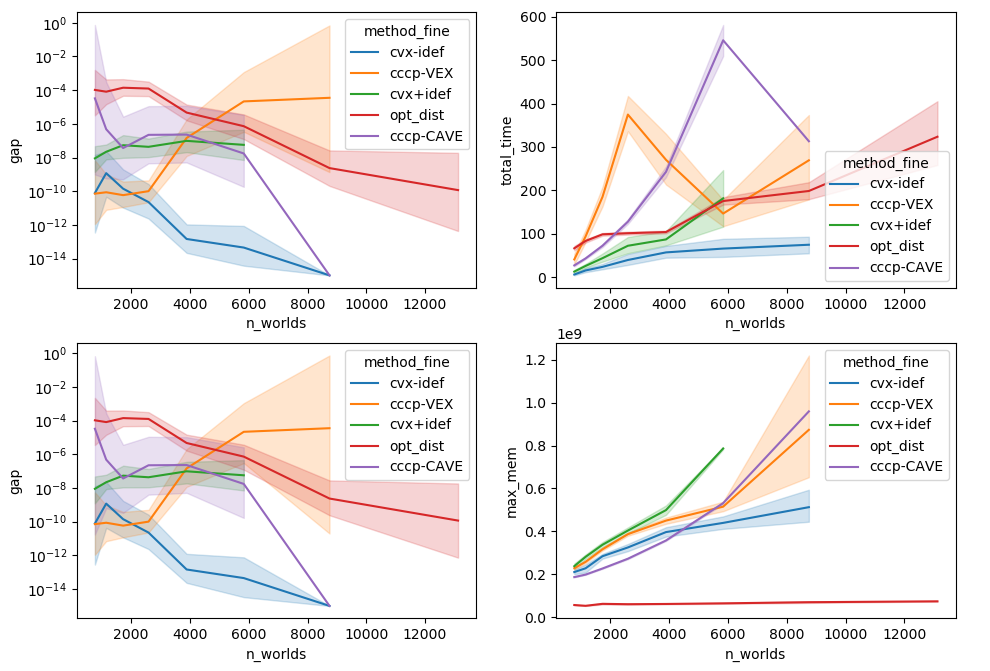}
     \caption{
         A variant of \cref{fig:resources}, with
         with gap (accuracy) information on the left, and slightly different parameter settings.
     }\label{fig:gap-resource-fine-old}
 \end{figure}

\discard{

\section{OTHER APPROACHES TO PDG INFERENCE} \label{sec:other-inference}

\subsection{}
A stronger result than \cref{prop:markov-property} holds as well.
\begin{prop}\label{prop:same-set-dists}
    Let $\Ed$ be a set of (hyper)edges over $\X$.
    For every PDG $\dg M$ over $\X$ with edges $\Ed$, every $\gamma > 0$, and every optimum $\mu^* \in \bbr{\dg M}_\gamma^*$ of $\dg M$'s scoring function at $\gamma$,
    there is a factor graph $\Phi$ with factors along $\Ed$ such that $\Pr_\Phi = \mu^*$.
\end{prop}

In other words: every distribution that a PDG can pick out as optimal (for any choice of $\gamma > 0$ and also in the limit as $\gamma \to 0$), can also be described as a factor graph with the same structure as that PDG.
How do we square this with the \citeauthor{pdg-aaai}'s claim that PDGs are more general than factor graphs?

\TODO[TODO: answer this question.\\
    The short answer: PDGs still compose differently, and in a way that respects the meaning of the probabilities. And just because you can find a factor graph that would have given you the right distribution after the fact, doesn't mean you could have specified the component factors.]

\TODO[Also: don't get lost; figure out how to continue as below:]
\cref{prop:same-set-dists} suggests another approach to avoiding an exponential representation of $\mu$: given a PDG, fit a factor graph that has the same structure to it.

\subsection{Approximate Inference}
\textbf{Relaxing the marginal polytope.}
Just as it is possible to do belief propagation on cluster graphs that are not trees (e.g., loopy belief propagation)
so too is it possible to drop the requirement that the cluster that we use is indeed a tree decomposition.
This program is smaller, and will converge, but it will only be an approximate solution.
Like the original PDG itself, it might be inconsistent.

\subsubsection{Variational Approaches}

}

\subsubsection*{Acknowledgements}
Halpern and Richardson were supported in part by MURI grant
W911NF-19-1-0217 and ARO grant W911NF-17-1-0592.
De Sa was supported by NSF RI-CAREER award 2046760.

\ifbiblatex
    \subsubsection*{References}
    \printbibliography
\else
    \bibliography{inference-refs}
    \fi

\end{document}